\documentclass[a4paper,onecolumn,11pt,accepted=2026-01-13]{quantumarticle}
\pdfoutput=1
\usepackage[utf8]{inputenc}

\usepackage{amsmath,amsthm,amssymb}
\usepackage{nicefrac}

\usepackage{tikz}
\usetikzlibrary{quantikz}

\usepackage{dsfont}
\usepackage{xcolor}
\usepackage{algorithm2e}
\RestyleAlgo{ruled}
\usepackage{bm}

\newcommand*{\dP}{\mathds{P}_n}
\newcommand*{\dI}{\mathbb{I}}

\newcommand*{\cE}{\mathcal{E}}
\newcommand*{\cI}{\mathcal{I}}
\newcommand*{\cL}{\mathcal{L}}
\newcommand*{\cT}{\mathcal{T}}
\newcommand*{\cO}{\mathcal{O}}
\newcommand*{\cU}{\mathcal{U}}
\newcommand*{\cN}{\mathcal{N}}
\newcommand*{\cH}{\mathcal{H}}
\newcommand*{\cM}{\mathcal{M}}
\newcommand*{\cG}{\mathcal{G}}
\newcommand*{\cB}{\mathcal{B}}
\newcommand*{\fS}{\mathfrak{S}}

\usepackage{hyperref}
\newcommand{\bigzero}{\mbox{\normalfont\huge $0$}}
\newcommand{\bigeye}{\mbox{\normalfont\huge $\dI$}}
\newcommand{\rvline}{\hspace*{-\arraycolsep}\vline\hspace*{-\arraycolsep}}
\usepackage{cleveref}

\usepackage[backend=biber, style=alphabetic, backref=true, hyperref=true, maxbibnames=99]{biblatex}
\usepackage[babel,english=british]{csquotes}
\DefineBibliographyStrings{english}{%
    backrefpage  = {cited on p.}, 
    backrefpages = {cited on pp.} 
}
\usepackage{authblk}

\newtheorem{theorem}{Theorem}[section]
\newtheorem*{theorem*}{Theorem}
\newtheorem{lemma}[theorem]{Lemma}

\newtheorem{corollary}[theorem]{Corollary}
\newtheorem{definition}[theorem]{Definition}
\newtheorem{proposition}[theorem]{Proposition}
\newtheorem{remark}[theorem]{Remark}
\newtheorem{problem}[theorem]{Problem}

\DeclareMathOperator{\W}{Wg}

\DeclareMathOperator{\Ber}{Bern}
\DeclareMathOperator{\Haar}{Haar}
\DeclareMathOperator{\unif}{Uniform}
\DeclareMathOperator{\KL}{KL}
\DeclareMathOperator{\TV}{TV}
\newcommand*{\eps}{\varepsilon}
\newcommand*{\id}{\mathrm{id}}
\newcommand*{\tr}{\mathrm{Tr}}
\newcommand*{\pr}[1]{\mathbb{P}\left(#1 \right)}
\newcommand*{\ex}[1]{\mathbb{E}\left[#1 \right]}
\newcommand{\norm}[1]{\left\lVert#1\right\rVert}

\newcommand*{\ptr}[2]{\mathrm{Tr}_{#1}\left[#2\right]}
\newcommand*{\exs}[2]{\mathbb{E}_{#1}\left[#2 \right]}
\newcommand*{\prs}[2]{\mathbb{P}_{#1}\left[#2 \right]}

\newcommand*{\spr}[2]{\langle #1 | #2 \rangle}

\newcommand{\nnorm}[1]{{\left\vert\kern-0.25ex\left\vert\kern-0.25ex\left\vert #1 
    \right\vert\kern-0.25ex\right\vert\kern-0.25ex\right\vert}}

\usepackage[normalem]{ulem}
\newcommand{\stkout}[1]{\ifmmode\text{\sout{\ensuremath{#1}}}\else\sout{#1}\fi}
\newif\ifverbose
 \verbosefalse
\newcommand{\ins}[1]{\ifverbose\textcolor{blue}{#1}\else#1\fi}
\newcommand{\edit}[2]{\ifverbose\textcolor{red}{\stkout{#1} #2}\else#2\fi}
\newcommand{\del}[1]{\ifverbose\textcolor{red}{\stkout{#1}}\fi}

\title{Hamiltonian Property Testing}
\author[1]{Andreas Bluhm}
\affil[1]{\small{Univ.\ Grenoble Alpes, CNRS, Grenoble INP, LIG, 38000 Grenoble, France}}
\author[2,3]{Matthias C.~Caro}
\affil[2]{\small{Dahlem Center for Complex Quantum Systems, Freie Universit\"at Berlin, Berlin, Germany}}
\affil[3]{\small{Department of Computer Science, University of Warwick, Coventry, UK}}
\author[4,5,6]{Aadil Oufkir}
\affil[4]{\small{Institute for Quantum Information, 
  RWTH Aachen University,
  Aachen, Germany}}
  \affil[5]{\small{Univ Lyon, Inria, ENS Lyon, UCBL, LIP,
  Lyon, France}}
  \affil[6]{\small Mohammed VI Polytechnic University, Rocade Rabat-Salé, Technopolis, Morocco}
\date{}

\begin{document}

\maketitle

\begin{abstract}   
    Locality is a fundamental feature of many physical time evolutions.
    Assumptions on locality and related structural properties also underlie recently proposed procedures for learning an unknown Hamiltonian from access to the induced time evolution.    
    However, no protocols to rigorously test whether an unknown Hamiltonian is in fact local were known.    
    We investigate Hamiltonian locality testing as a property testing problem, where the task is to determine whether an unknown $n$-qubit Hamiltonian $H$ is $k$-local or $\varepsilon$-far from all $k$-local Hamiltonians, given access to the time evolution along $H$. 
    First, we emphasize the importance of the chosen distance measure: With respect to the operator norm, a worst-case distance measure, incoherent quantum locality testers require $\tilde{\Omega}(2^n)$ many time evolution queries and an expected total evolution time of $\tilde{\Omega}(2^n / \varepsilon)$, and even coherent testers need $\Omega(2^{n/2})$ many queries and $\Omega(2^{n/2}/\varepsilon)$ total evolution time.
    In contrast, when distances are measured according to the normalized Frobenius norm, corresponding to an average-case distance, we give a 
    computationally efficient incoherent Hamiltonian locality testing algorithm with query complexity $\mathcal{O}(1/\varepsilon^4)$ and total evolution time $\mathcal{O}(1/\varepsilon^3)$, based on randomized measurements. In fact, our procedure can be used to simultaneously test a wide class of Hamiltonian properties beyond locality.
    Finally, we prove that learning a general Hamiltonian remains exponentially hard with this average-case distance, thereby establishing an exponential separation between Hamiltonian testing and learning.
    Our work initiates the study of property testing for quantum Hamiltonians, demonstrating that a broad class of Hamiltonian properties is efficiently testable even with limited quantum capabilities, and positioning Hamiltonian testing as an independent area of research alongside Hamiltonian learning.
\end{abstract}

\section{Introduction}

\edit{Time evolution of a quantum system according to the unitary group generated by some self-adjoint Hamiltonian is a crucial ingredient to quantum physics, going back to the early days of quantum mechanics \cite{schrodinger1926quantisierung, schrodinger1926undulatory, stone1932one-parameter}.}{In quantum mechanics, systems evolve according to the unitary group generated by some self-adjoint Hamiltonian \cite{schrodinger1926quantisierung, schrodinger1926undulatory, stone1932one-parameter}.}
This makes extracting information about an unknown Hamiltonian from access to the corresponding time evolution a central task when it comes to understanding fundamental processes in physics\ins{, which is studied for instance in quantum sensing \cite{allen2025quantumcomputingenhancedsensing}}.
Moreover, given recent progress in experimental implementations of early quantum devices, \edit{the same task}{extracting Hamiltonian information} also gains technological relevance, for instance for benchmarking applications in quantum simulation and quantum computing.

In many physically relevant quantum systems, the Hamiltonian describing the time evolution is not an arbitrary self-adjoint operator but has additional structure. In particular, as physical processes often arise from local interactions, $k$-local Hamiltonians -- Hamiltonians that can be written as a sum of terms that act non-trivially only on $k$ out of the overall $n$ many subsystems, often with $k=\mathcal{O}(\log(n))$ or $k=\mathcal{O}(1)$ -- play an important role in modeling real-world systems. Additionally, some proposed quantum computing architectures \ins{such as \cite{arute2019quantum, google2025quantum}} also come with natural locality constraints, and are thus naturally described by (possibly even geometrically) local Hamiltonian\edit{s}{ interaction terms}.
In addition to fundamental and practical relevance, locality also carries theoretical importance. For instance, a majority of recent works establishing rigorous guarantees for learning an unknown Hamiltonian from either its Gibbs state or from access to its time evolution (see \Cref{subsection:related-work} for an overview) \edit{relies}{rely} on structural assumptions on the Hamiltonian that in particular require locality, such as having a bounded-degree interaction graph.
However, no protocols for testing whether an unknown Hamiltonian satisfies such assumptions were known. \ins{The testing protocol we propose in this work can hence be used to check whether a given unknown Hamiltonian satisfies this assumption and thus serve as a preprocessing step for learning. Furthermore, by iteratively applying our protocol for different locality parameters, we can estimate the locality parameter and thereby quantify the amount of resources needed for Hamiltonian learning. Additionally, our protocol can be used to check the implementation of a known Hamiltonian, as any implementation will very likely be imperfect due to cross-talk between qubits and other sources of noise.}

In this work, we investigate the tasks of testing locality and more general properties of an unknown Hamiltonian\ins{, such as sparsity and being low-intersection (see \Cref{remark:application-sparse}),} in the framework of property testing. 
Concretely, given access to the time evolution\footnote{\ins{We note that we do not assume access to inverse time evolution. While inverse time evolution access is natural if the unknown quantum process is realized via a quantum circuit, reversing time can be hard in physically motivated scenarios; as argued in \cite{tang2025amplitudeamplificationestimationrequire}, inverse time evolution access is a non-trivial resource.}} according to an unknown Hamiltonian $H$, we aim to determine whether $H$ is $k$-local or far from all $k$-local Hamiltonians.
In fact, versions of this Hamiltonian locality testing task have already been proposed as interesting problems, albeit not studied, in \cite{montanaro2016survey, she2022unitary}.
We demonstrate that the feasibility of Hamiltonian locality testing crucially relies on how distances between Hamiltonians are measured. On the one hand, we establish hardness of locality testing with the distance measured by the operator norm. On the other hand, for the normalized Frobenius norm as distance measure, we give an efficient Hamiltonian locality tester. In fact, we show that our algorithm can be modified to efficiently test for a variety of Hamiltonian properties, specified by subsets of all possible Pauli strings. Finally, we highlight a crucial difference between Hamiltonian testing and learning: While we achieve efficient Hamiltonian property testing with respect to the normalized Frobenius norm, we show hardness of learning an arbitrary Hamiltonian with the same notion of distance.

\subsection{Problem statement: Hamiltonian locality testing}

Throughout, we consider the Pauli expansion $H=\sum_{P\in\mathds{P}_n} \alpha_P P$ of an $n$-qubit Hamiltonian $H$, where $\mathds{P}_n=\{\mathbb{I}, X,Y,Z\}^{\otimes n}$ is the set of $n$-qubit Paulis and where the $\alpha_P = \nicefrac{\tr[HP]}{2^n}$ are the (real) Pauli basis coefficients of $H$.
We call the Hamiltonian $H$ $k$-local if $\alpha_P=0$ holds for all $P\in\mathds{P}_n$ with $\lvert P\rvert >k$.
Here, $\lvert P\rvert$ denotes the weight of a Pauli string $P$, that is, the number of non-identity tensor factors in $P$.
\ins{We w.l.o.g.~assume that $\alpha_{\mathbb{I}^{\otimes n}}=0$ to avoid physically irrelevant global phases.}
We phrase the problem of testing whether an unknown Hamiltonian is (at most) $k$-local as a property testing problem:

\begin{definition}[Hamiltonian locality testing]\label{inf-def:hamiltonian-locality-testing}
    Given a locality parameter $1\leq k\leq n$, a norm $\nnorm{\cdot}$, and an accuracy parameter $\varepsilon\in (0,1)$, the Hamiltonian $k$-locality testing problem, denoted as $\mathcal{T}_{\nnorm{\cdot}}^{\mathrm{loc}}(\varepsilon)$, is the following task:
    Given access to the time evolution according to an unknown Hamiltonian $H$, decide, with success probability $\geq 2/3$, whether
    \begin{enumerate}
        \item[(i)] $H$ is $k$-local, or
        \item[(ii)] $H$ is $\varepsilon$-far from being $k$-local, that is, $\nnorm{ H - \Tilde{H}}\geq \varepsilon$ for all $k$-local Hamiltonians $\Tilde{H}$.
    \end{enumerate}
    If $H$ satisfies neither (i) nor (ii), then any output of the tester is considered valid.
\end{definition}

Here, we can use different norms $\nnorm{\cdot}$ to measure the distance between two Hamiltonians.
Motivated by operational interpretations as worst-case and average-case notions of distance, respectively, we focus on the Schatten $\infty$-norm $\norm{\cdot}_\infty$ (aka operator norm) and the normalized Schatten $2$-norm $\frac{1}{\sqrt{2^n}}\norm{\cdot}_2$ (aka normalized Frobenius norm). 

In this work, we study how easy or hard it is to achieve Hamiltonian locality testing \edit{w.r.t.~}{with respect to} these different norms.
On the one hand, we consider incoherent quantum algorithms, which can only perform measurements on single copies of time-evolved states, and which cannot interleave Hamiltonian time evolution with control operations. 
Here, we do, however, allow for adaptively chosen input states and measurements in the different experiments performed for testing.
On the other hand, we also consider more general coherent quantum algorithms, describing the most general quantum experiments that can be performed when given access to an unknown Hamiltonian time evolution.
Additionally, we may restrict the algorithms to have access to no or only few auxiliary qubits. 
See \Cref{subsection:types-of-strategies} for a detailed discussion of the kinds of quantum protocols that we consider.
Our negative results cover both the incoherent and the coherent case.
Our positive results are phrased in the framework of simple incoherent quantum testing algorithms. As such, they serve as examples of the power of simple-to-implement protocols for testing properties of a Hamiltonian, which may be feasible in short- or mid-term quantum devices.

\subsection{Main results}

An immediate approach towards solving the testing problem from \Cref{inf-def:hamiltonian-locality-testing} is to (approximately) learn the unknown Hamiltonian to a sufficient accuracy and to then decide based on whether the learned hypothesis Hamiltonian has the desired property or not.
Given that the unknown Hamiltonian $H$ can be arbitrary, instantiating this approach requires Hamiltonian learning protocols that work without any structural assumptions. Such procedures were recently proposed in \cite{caro2023learning, castaneda2023hamiltonian}.
However, as we discuss in more detail in \Cref{sec:upper-bounds-from-hamiltonian-learning}, this naive ``testing via learning'' strategy with the protocols of \cite{caro2023learning, castaneda2023hamiltonian} has an undesirable feature: 
When working with any norm except for the $\ell_\infty$-norm on the coefficient vector $(\alpha_P)_P$, $\lVert\cdot\rVert_{\mathrm{Pauli},\infty}$, it uses a number of queries to the unknown Hamiltonian and a total evolution time that both scale exponentially in $n$.
This raises the question: 

\begin{center}
    \emph{For which norms is it possible to efficiently test Hamiltonian locality?}
\end{center}

We provide the first rigorous answers to this question
In our first result \ins{(see \Cref{thm:lower-Schatten-indep,sequential-LB-op-norm} for a formal statement)}, we show that the locality testing task is hard with the Schatten $\infty$-norm:

\begin{theorem}[Hardness of Hamiltonian locality testing \edit{w.r.t.~}{with respect to }the operator norm -- Informal]\label{inf-thm:hamiltonian-locality-testing-hardness}
    For $k\leq\Tilde{\mathcal{O}}(n)$, any ancilla-free, incoherent, adaptive quantum algorithm that solves the $k$-locality testing problem $\mathcal{T}_{\norm{\cdot}_\infty}^{\mathrm{loc}}(\varepsilon)$, even only under the additional promise that the unknown Hamiltonian $H$ satisfies $\tr[H]=0$ and $\norm{H}_\infty\leq 1$, has to make at least $N\geq \tilde{\Omega}\left(2^n\right)$ queries to the unknown Hamiltonian and has to use an expected total evolution time of at least $\mathbb{E}[T]\geq \tilde{\Omega}\left(\frac{2^n}{\varepsilon}\right)$.
    Even any coherent quantum algorithm achieving the same has to make at least $N\geq \Omega\left(2^{n/2}\right)$ many queries and has to use a total evolution time of at least $T\geq \Omega\left(\frac{2^{n/2}}{\varepsilon}\right)$.
\end{theorem}

Here and throughout, we use $\Tilde{\Omega}$ to hide factors that are polylogarithmic in the leading order term.
\Cref{inf-thm:hamiltonian-locality-testing-hardness} says that testing $k$-locality is not possibly query-efficiently or with efficient evolution time when using the Schatten $\infty$-norm as a distance measure. In fact, we show this for all Schatten $p$-norms, $p\geq 1$.
With our next result, we demonstrate that this changes significantly when considering the normalized Schatten $2$-norm instead, which (see \Cref{appendix:normalized-frobenius-norm}) corresponds to a an average-case distance measure for unitaries and channels \cite{nielsen2002simple} that has recently been considered in a variety of testing and learning contexts \cite{montanaro2016survey, huang2022quantum, caro2022out-of-distribution, caro2023learning, zhao2023learning, bao2023testing,  huang2024learning, nadimpalli2024paulispectrum, vasconcelos2024learningshallowquantumcircuits}. 
With this distance measure, we can achieve efficiency in terms of the number of queries, the total evolution time, and even with respect to the classical post-processing time.
In particular, we give a procedure that tests whether an unknown Hamiltonian is $k$-local and that achieves this with a polynomial number of queries, a polynomial overall evolution time, and with polynomial classical post-processing time.

\begin{theorem}[Efficient Hamiltonian locality testing \edit{w.r.t.~}{with respect to }normalized Frobenius norm -- Informal]\label{inf-thm:hamiltonian-locality-testing-normalized-frobenius}
    Let $k\leq \Tilde{\mathcal{O}}(n)$.
    When promised that the unknown Hamiltonian $H$ satisfies $\tr[H]=0$ and $\norm{H}_\infty\leq 1$, there is an ancilla-free, incoherent, non-adaptive quantum algorithm that solves the Hamiltonian $k$-locality testing problem $\mathcal{T}_{\frac{1}{\sqrt{2^n}}\norm{\cdot}_2}^{\mathrm{loc}}(\varepsilon)$ using $\mathcal{O}\left(\varepsilon^{-4}\right)$ many queries to the unknown Hamiltonian, a total evolution time of $\mathcal{O}\left(\varepsilon^{-3}\right)$, and a classical post-processing time of $\mathcal{O}\left(\frac{n^{k+3}}{\varepsilon^4}\right)$. 
    Moreover, the testing algorithm uses only stabilizer states as inputs and stabilizer basis measurements at the output.
\end{theorem}

In fact, while we state \Cref{inf-thm:hamiltonian-locality-testing-normalized-frobenius} \ins{(the formal version of which can be found in \Cref{corollary:locality-testing-upper-bound})} in terms of locality testing, our procedure allows us to establish a more general Hamiltonian property testing result. 
In particular, we show \ins{(see \Cref{thm:upper bound on testing} for a detailed statement)}: Let $S\subset\mathds{P}_n$ be a subset with $\lvert S\rvert\leq \mathcal{O}(\mathrm{poly}(n))$, let $\varepsilon\geq \Omega(1/\mathrm{poly(n)})$.
Then, we can efficiently test, even tolerantly, whether $H$ consists only of Pauli terms in $S$ or whether $H$ is $\varepsilon$-far \edit{w.r.t.~}{with respect to }$\frac{1}{\sqrt{2^n}}\norm{\cdot}_2$ from the set of all such Hamiltonians. 
More precisely, we can do so with $\mathcal{O}\left(1/\varepsilon^4\right)\leq \mathcal{O}(\mathrm{poly}(n))$ many queries to $H$, with a total evolution time of $\mathcal{O}\left(1/\varepsilon^3\right)\leq \mathcal{O}(\mathrm{poly}(n))$, and with a classical post-processing time of $\mathcal{O}\left(n^2 \lvert S\rvert/\varepsilon^4\right)\leq \mathcal{O}(\mathrm{poly}(n))$.
Again, this testing procedure uses only stabilizer state inputs and stabilizer basis measurements at the output, both of which are efficiently implementable.
Thus, by choosing a suitable set $S$, we can for example test whether an unknown Hamiltonian is exactly $k$-local, (at most or exactly) geometrically $k$-local, or whether it has a desired interaction graph.

Our algorithm achieving the guarantees in \Cref{inf-thm:hamiltonian-locality-testing-normalized-frobenius} and its extended version to more general Hamiltonian properties, specified by some subset $S$, are novel additions to the randomized measurement framework \cite{elben2022randomized}. In particular, we inherit the ``measure first, ask questions later'' feature. That is, the data in our testing algorithm can be collected even without knowing the Hamiltonian property that is to be tested as long as an a priori bound on its size and the desired accuracy are known in advance; and once collected the data can be used to test multiple properties simultaneously.
This then allows us to test properties such as whether a Hamiltonian has a sparse Pauli basis expansion or whether it is a low-intersection Hamiltonian.

\ins{While \Cref{inf-thm:hamiltonian-locality-testing-normalized-frobenius} assumes $\norm{H}_\infty\leq 1$, we can easily extend the result. Namely, assuming $\norm{H}_\infty\leq B$, we can rescale the accuracy to $\tilde{\varepsilon}=\varepsilon/B$ and test whether $\Tilde{H}=H/B$ is $k$-local or $\Tilde{\varepsilon}$-far from $k$-local. As all relevant complexities scale inverse-polynomially in the accuracy, this results in polynomial dependencies on $B$. Hence, our procedure is efficient for $B\leq\mathcal{O}(\mathrm{poly}(n))$, which is for instance the case if $H$ is a sum of polynomially many Pauli terms with bounded Pauli coefficients.}

In our final result, we highlight a large separation between Hamiltonian testing and learning. 
On the one hand, we have shown that normalizing the Frobenius norm makes it possible to efficiently test arbitrary Hamiltonian properties. 
On the other hand, even after normalizing the distance measure, learning a general Hamiltonian from time evolution access remains hard:

\begin{theorem}[Hardness of Hamiltonian learning \edit{w.r.t.~}{with respect to }normalized Frobenius norm -- Informal]\label{inf-thm:hardness-hamiltonian-learning}
    Any (even coherent) quantum algorithm that, when given time evolution access to an arbitrary $n$-qubit Hamiltonian $H$, promised to satisfy $\tr[H]=0$ and $\norm{H}_\infty\leq 1$, with success probability $\geq 2/3$, outputs (the classical description of) a Hamiltonian $\hat{H}$ such that $\frac{1}{\sqrt{2^n}}\norm{H - \hat{H}}_2\leq\varepsilon$ has to make at least $\tilde{\Omega}\left(2^{2n}\right)$ many queries to $H$.
    Any non-adaptive incoherent quantum algorithm achieving the same without auxiliary qubits has to use a total evolution time of at least $\tilde{\Omega}\left(\frac{2^{2n}}{\varepsilon}\right)$.
\end{theorem}

Juxtaposing \Cref{inf-thm:hamiltonian-locality-testing-normalized-frobenius}, and its extension to general properties, with \Cref{inf-thm:hardness-hamiltonian-learning} \ins{(and its formal version, given as \Cref{thm:hardness-hamiltonian-learning})}, we see that the naive ``testing via general learning'' approach fails for Hamiltonian property testing if we do not have prior promises on the structure of the unknown Hamiltonian.
In fact, Hamiltonian locality testing, and even more general Hamiltonian property testing, is significantly easier than the infeasible task of general Hamiltonian learning, but requires approaches tailored specifically to testing.
Finally, we note that the reasoning behind \Cref{inf-thm:hardness-hamiltonian-learning} also gives rise to an ${\Omega}\left(\frac{2^{2k}}{k}\right)$ query complexity lower bound for any coherent quantum algorithm  that learns an unknown $k$-local Hamiltonian to accuracy $\varepsilon$ in normalized Frobenius norm. Thus we obtain a superexponential-in-$k$ separation between learning under a $k$-locality promise and testing for that same promise, which constitutes a learning versus testing separation in the sense commonly considered in property testing.

\subsection{Related work}\label{subsection:related-work}

\paragraph{Classical and quantum property testing.}

Since its origins \cite{blum1993self, rubinfeld1996robust, goldreich1998property}, property testing has evolved into an important area of theoretical computer science, with connections to, among others, learning theory and probabilistically checkable proofs \cite{ron2008property, goldreich2017introduction, bhattacharyya2022property}.
Among the plethora of classical property tasks, those of low-degree testing \cite{alon2003testing} and junta testing \cite{fischer2004testing,blais2009testing} can be viewed as counterparts of quantum Hamiltonian locality testing. Similarly, testing for Fourier sparsity of a Boolean function \cite{gopalan2011testing} can be viewed as a classical version of testing whether a Hamiltonian has a sparse Pauli basis expansion.
More recently, the field of quantum property testing \cite{montanaro2016survey} has emerged. 
It includes a long line of works on testing properties of classical objects from quantum data access \cite{deutsch1985quantum, deutsch1992rapid, simon1997power, atici2007quantum, buhrman2008quantum, chakraborty2010newresults, ambainis2011quantum, bravyi2011quantum, hillery2011quantum, aaronson2018forrelation, ambainis2016efficient, gilyen2020distributional}; investigations into testing properties of quantum states \cite{harrow2013testing, odonnell2015quantum, harrow2017sequential, carmeli2017probing, buadescu2020lower, gross2021schur, soleimanifar2022testing, grewal2023improved}; proofs of proximity for unitary properties \cite{dallagnol2022quantumproofsof}; unitary and channel versions of junta testing \cite{chen2023testing, bao2023testing}; unitary property testing more broadly \cite{laborde2022quantum, she2022unitary}; and Hamiltonian symmetry testing \cite{laborde2022quantum}. 

\paragraph{Hamiltonian learning.}
One way to infer properties of a Hamiltonian clearly is to learn the coefficients of the Hamiltonian. There has been a lot of work on this topic, hence the references below should not be understood as a complete review of the field. For our purposes, there are two approaches to Hamiltonian learning. In the first, one has access to the unitary dynamics of the Hamiltonian of interest \cite{Silva2011Practical, Bairey2019Learning, Zubida2021Optimal, haah2022optimal, wilde2022scalably, yu2023robust, caro2023learning, Dutkiewicz.2023, huang2023heisenberg, castaneda2023hamiltonian, li2023heisenberglimited, möbus2023dissipationenabled, franca2024efficient, Gu2022Practical}. We are allowed to prepare appropriate initial states, choose how long the system should evolve, and perform measurements of our choosing afterwards. In the second, our aim is to learn the Hamiltonian from copies of Gibbs states, i.e., thermal equilibrium states \cite{anshu2021sample, haah2022optimal, rouze2023learning, onorati2023efficient, bakshi2023learning, Gu2022Practical}. For these algorithms to be efficient, it is usually assumed that the Hamiltonian to be learned is local. Our results can therefore be seen as complementary to the above protocols for Hamiltonian learning: we first determine the locality of the Hamiltonian using our results and then run an appropriate learning protocol.

\paragraph{Coherent vs.~incoherent quantum learning.}
Quantum-enhanced learning algorithms can use advanced quantum processing, such as multi-copy measurements and coherent long-time evolutions with interleaving control operations, to achieve a quantum advantage over conventional algorithms.
Recent work has investigated such advantages in learning and testing quantum states \cite{bubeck2020entanglement, huang2021information-theoretic, chen2022exponential, huang2022quantum, arunachalam2023optimal, chen2023does, fawzi_adaptivity_2023} as well as unitaries and quantum channels \cite{aharonov2022quantum, chen2022exponential, huang2022quantum, caro2023learning, chen2023unitarity, oufkir2023sample, fawzi2023lower,fawzi2023quantum-channel-certification}, and in learning Hamiltonians \cite{huang2023heisenberg, li2023heisenberglimited}. 
Our lower bounds constitute an addition to the toolkit developed in these prior works.
Notably, the algorithms achieving our upper bounds use only simple quantum processing as is common in the randomized measurement paradigm \cite{huang2020predicting, elben2022randomized}.

\subsection{Techniques and proof overview}
Let $d=2^n$ be the dimension of an $n$-qubit system. 

\paragraph{Hamiltonian locality testing lower bound.}

To prove the first Hamiltonian locality testing lower bound of \Cref{inf-thm:hamiltonian-locality-testing-hardness}, we identify an underlying Hamiltonian many-vs-one distinguishing problem that can be solved with high success probability by any successful locality tester, and then establish lower bounds for this distinguishing task.
Concretely, we consider the following task: Given access to the time evolution along a Hamiltonian $H$ that is promised to satisfy either (i) $H=0$ or (ii) $H=\varepsilon (V\proj{0}V^\dagger-\dI/d)$, where $V$ is a Haar-random $n$-qubit unitary, decide whether (i) or (ii) is the case.

We show that Hamiltonian locality testing indeed suffices to solve this distinguishing problem via a concentration of measure argument. Namely, using the concentration of a Lipschitz function of Haar-random unitaries around its mean \cite{meckes2013spectral}, we show that $\varepsilon (V\proj{0}V^\dagger-\dI/d)$  with Haar-random $V$ is $(\varepsilon/2)$-far \edit{w.r.t.~}{with respect to }$\norm{\cdot}_\infty$ from any fixed, traceless Hamiltonian $K$ with $\norm{K}_\infty\leq 1$ with probability $\geq 1 - \exp(- \Omega(2^n))$.
As the set  of $k$-local Hamiltonians admits an $\norm{\cdot}_\infty$-covering net $\mathcal{H}_\varepsilon$ whose size satisfies $\log\lvert\mathcal{H}_\varepsilon\rvert\leq \min\{(k+1)(3n)^k,4^n\}$, a union bound now implies that $\varepsilon (V\proj{0}V^\dagger-\dI/d)$  with Haar-random $V$ is simultaneously $(\varepsilon/4)$-far, again \edit{w.r.t.~}{with respect to }$\norm{\cdot}_\infty$, from all $k$-local Hamiltonians with high probability.
Therefore, any high-probability algorithm for testing Hamiltonian $k$-locality to accuracy $\varepsilon/4$ in the operator norm also manages to distinguish between (i) $H=0$ and (ii) $H=\varepsilon (V\proj{0}V^\dagger-\dI/d)$ with Haar-random $V$ with high success probability.
In particular, any lower bound for this distinguishing task immediately implies a lower bound for Hamiltonian locality testing.

To establish such a lower bound, we follow \cite{fawzi2023quantum-channel-certification}. To explain the proof of the lower bound, we
use the learning tree representation of \cite{chen2022exponential}. That is, we think of the possible outcomes of an adaptive distinguishing algorithm as leaves in a tree, where the observed measurement outcomes in each round determine how the learner moves from the root to a leaf.
Viewed this way, Le Cam's two-point method implies that solving the distinguishing task with constant success probability requires the leaf distributions induced by the different hypothesis Hamiltonians to have at least a constant total variation (TV) distance. 
For technical reasons, we in fact consider a slightly different test: Distinguish between the channels (i) $\cU_t=\id$ or (ii) $\cU_t(\cdot)=\alpha\, \id(\cdot) + (1-\alpha)\mathrm{e}^{-\mathrm{i}tH}(\cdot )\mathrm{e}^{\mathrm{i}tH}$ for a random $H$ as above. 
We show that as long as $\alpha\le \frac{1}{10N}$\ins{, where $N$ is the number of queries}, the leaf distributions induced by these two different hypotheses still have at least a constant TV distance. 
Next, we use Pinsker's inequality to upper bound (the square of) this TV distance by the Kullback Leibler (KL) divergence, as the KL divergence is more amenable to decoupling the dependence between the random observations at different steps. 
Taking a mixture of the identity channel and the time evolution under the alternate hypothesis as in (ii) is important to make a second order Taylor expansion of the logarithm function that occurs in the KL divergence possible. 
However, this alone is not sufficient since the (expected) second order term in this expansion can diverge if the input states and measurement projectors are orthogonal. 
To address this issue, at each step $k\in [N]$, we distinguish between two types of paths of length $k$: those for which the overlap between the input state and the measurement operator (corresponding to the final node in the path) is too small, and those for which the overlap is not. 
When the overlap is too small, we use a simple lower bound on Born's probability under $\cU_t(\cdot)=\alpha\, \mathrm{id}(\cdot) + (1-\alpha)\mathrm{e}^{-\mathrm{i}tH}(\cdot )\mathrm{e}^{\mathrm{i}tH}$ in terms of the probability under $\cU_t=\id$ (here again taking the  mixture proves to be essential).  When the overlap is not too small, we can Taylor expand the logarithm to second order and control the resulting term using Weingarten calculus. 
Overall, we achieve a TV distance upper bound of $\sum_{k=1}^N \mathcal{O}\left(\frac{\log(N)}{d}\cdot\mathbb{E}[\min(1,\varepsilon t_k)]\right)$, which when compared to the constant lower bound implies the claimed bounds on $N$, the number of queries, and on $\mathbb{E}[\sum_{k=1}^N t_k]$, the expected total evolution time. This completes the sketch of the lower bound for incoherent learners, which we view as our main contribution for testing lower bounds. 

For coherent testers, we obtain lower bounds that are quadratically weaker, using the same distinguishing problem. Lower bounds with the same scaling but based on a simpler distinguishing problem have appeared in the literature on quantum phase estimation, see for example \cite[Section 3]{made2023tight}. While similar in spirit, the simpler distinguishing task does not seem to yield the stronger linear-in-$d$ scaling for the incoherent case. Also, it has previously been considered for access models different from ours. For example, the lower bound in \cite{made2023tight} assumes access to the unknown unitary only at a fixed time, but requires also access to its inverse.
\paragraph{Hamiltonian locality testing upper bound -- Commuting case.}

For clarity of exposition, we begin with a simpler setting, focusing on commuting Hamiltonians consisting of terms from $\{ \dI, X\}^{\otimes n}$. That is, we assume that we can expand the Hamiltonian as $H = \sum_{P \in \{ \dI, X\}^{\otimes n}} \alpha_P P$.
Our concrete task under consideration thus becomes: Given access to the time evolution along an unknown Hamiltonian $H = \sum_{P \in \{ \dI, X\}^{\otimes n}} \alpha_P P$ with $\tr[H]=0$ and $\norm{H}_\infty\leq 1$, decide, with success probability $\geq 2/3$, whether $H$ is $k$-local or $\varepsilon$-far from $k$-local in normalized Schatten $2$-norm distance.
The underlying idea for locality testing in this special case will carry over to the general setting.

Our algorithm for solving this property testing problem is as follows:
We prepare the state $\proj{0}$ and let it evolve under the unknown Hamiltonian for time $t = \mathcal O(\eps)$. At the end of this evolution, we perform a measurement in the computational basis $\{\proj{i}\}_{i \in \{0,1\}^n}$. 
We repeat this procedure $N = \mathcal O(1/\mathrm{poly}(\eps))$ times. If at least one of the $N$ rounds produces an $n$-bit string with at least $k+1$ non-zero entries, i.e., with Hamming weight $\geq k+1$, as measurement outcome, we conclude that the Hamiltonian $H$ is $\varepsilon$-far from being $k$-local. Otherwise, we claim that $H$ is $k$-local.

The proof of correctness for this algorithm has the following structure. 
Observe that $\ket{j} = X^j \ket{0}$ for any $j \in \{0,1\}^n$ and that $U_t = \mathrm{e}^{\mathrm{i}tH} \approx \dI+ \mathrm{i}t H$ holds for short times $t$. Here, we write $X^j$ to mean $X^{j_1} \otimes \ldots \otimes X^{j_n}$, where $X^0 = \dI$.
Consequently, we have $\bra{j} U_t \ket{0} \approx \delta_{j,0} + \mathrm{i}t \alpha_{X^j}$ for all $j \in \{0,1\}^n$.
In particular, for any $n$-bit string $j$ with weight $|j| > 0$, 
\begin{equation*}
    |\bra{j} U_t \ket{0}|^2 \approx  t^2 |\alpha_{X^j}|^2\, .
\end{equation*}
If $H$ is indeed $k$-local, then $\alpha_{X^j} = 0$ holds whenever $|j|>k$, and we find that 
\begin{equation*}
    \sum_{j: |j| > k}  |\bra{j} U_t \ket{0}|^2 \approx 0\, ,
\end{equation*}
so the probability that the algorithm falsely claims that $H$ is far from $k$-local when it is indeed $k$-local is approximately zero.

Conversely, if the Hamiltonian $H$ is $\eps$-far from any $k$-local Hamiltonian in normalized Schatten $2$-norm, this means that $\sum_{j: |j| > k}  |\alpha_{X^j}|^2 \geq \eps^2$ because of Parseval. We infer that
\begin{equation*}
      \sum_{j: |j| > k}  |\bra{j} U_t \ket{0}|^2 \gtrapprox t^2 \eps^2 \, .
\end{equation*}
Note that we have to choose $t$ small enough for the approximation to be correct.
By repeating the algorithm $\mathcal O(t^{-2} \eps^{-2})=\mathcal{O}(\eps^{-4})$ many times, we can increase the probability of the algorithm to be correct to a constant. To make the above reasoning precise, we use the Taylor expansion of $\mathrm{e}^{\mathrm{i}tH}$ in order to bound the error in the approximation $\mathrm{e}^{\mathrm{i}tH} \approx \dI+ \mathrm{i}t H$. 

\paragraph{Hamiltonian locality testing upper bound -- General case.}
In the case of a general Hamiltonian $H$, we can no longer limit ourselves to preparing and measuring in the computational basis. Instead, we need to prepare and measure in several different bases. A convenient choice are some $d+1$ \emph{mututally unbiased bases} (MUBs) $\mathcal B_i$, which are known to exist in our case since $d=2^n$ is a prime power \cite{wootters1989optimal}. We write
\begin{equation*}
    \mathcal B_i = \{\ket{\phi_{i,j}}\}_{j \in \{1, \ldots, d\}}, ~1\leq i\leq d+1\, .
\end{equation*}
MUBs are known to be particularly suited for determining the state of a quantum system \cite{wootters1989optimal}. Moreover, MUBs were used for Pauli channel learning \cite{flammia2020efficient}. They can be explicitly constructed from covering the Pauli group with $d+1$ stabilizer groups that only have the identity element in common but are otherwise disjoint.

Motivated by our previous discussion, we consider for a Pauli operator $P$ with weight $|P|$ the overlap $|\bra{\phi_{i,\ell}} P \ket{\phi_{i,j}}|$ and find that it is either $0$ or $1$.
This motivates the following algorithm:
We choose $(i,j) \in [d] \times [d+1]$ uniformly at random and prepare the state $\proj{\phi_{i,j}}$. We let it evolve under the unknown Hamiltonian for time $t = \mathcal O(\eps)$. At the end of this evolution, we perform a measurement in the basis $\mathcal B_i$. We repeat this procedure $N=\mathcal O(1/\mathrm{poly}(\eps))$ times. 
If at least one of the $N$ rounds produces an output $\ell$ such that all Pauli strings $P$ with $|P|\le k$ satisfy $|\bra{\phi_{i,\ell}} P \ket{\phi_{i,j}}|=0$ -- which we denote by $\ket{\phi_{i,\ell}} \nsim_k \ket{\phi_{i,j}}$, meaning that we detected a non-locality --, then we conclude that the Hamiltonian is $\eps$-far from being local.
Otherwise, we claim that $H$ is $k$-local.

By construction, if $\ket{\phi_{i,\ell}} \nsim_k \ket{\phi_{i,j}}$, then $|\bra{\phi_{i,\ell}} H^m \ket{\phi_{i,j}}| = 0$ for $m \in \{0,1\}$. Using $\mathrm{e}^{\mathrm{i}tH} \approx \dI+ \mathrm{i}t H$, we find that if the Hamiltonian is $k$-local, then in any single round, the probability of our procedure falsely detecting a non-locality is
\begin{equation*}
    \pr{\ket{\phi_{i,\ell}} \nsim_k \ket{\phi_{i,j}} } =\frac{1}{d(d+1)} \sum_{i=1}^{d+1}\sum_{j\neq \ell}|\bra{\phi_{i,\ell}} \mathrm{e}^{\mathrm{i}tH} \ket{\phi_{i,j}}|^2 \mathbf{1}\left(\left\{ \ket{\phi_{i,\ell}} \nsim_k \ket{\phi_{i,j}}  \right\}\right) \approx 0 \, .
\end{equation*}
If the Hamiltonian is $\epsilon$-far from being $k$-local, we can use the fact that $d+1$ MUBs in dimension $d$ form a 2-design \cite{klappenecker2005mutually} to infer the following upper bound on the probability of not detecting the non-locality in any single round:
\begin{equation} \label{eq:2-design-used}
    \pr{\ket{\phi_{i,\ell}} \sim_k \ket{\phi_{i,j}}} \leq \sum_{P: |P|\le k} \frac{d}{d(d+1)} + \sum_{P: |P|\le k} \frac{\left|\tr\left(P \mathrm{e}^{\mathrm{i}tH}\right)\right|^2}{d(d+1)} \, .
\end{equation}
We use $\mathrm{e}^{\mathrm{i}tH} \approx \dI+ \mathrm{i}t H$ to conclude $\tr\left(P \mathrm{e}^{\mathrm{i}tH}\right) \approx \mathrm{i}t d \alpha_P$ if $P \neq \dI$. Furthermore, expanding $\mathrm{e}^{\mathrm{i}tH}$ to second order, we can approximate
\begin{equation*}
    |\tr\left(\mathrm{e}^{\mathrm{i}tH}\right)|^2 \approx d^2 - d^2 t^2 \sum_P |\alpha_P|^2 \, .
\end{equation*}
Combining this with $H$ being $\varepsilon$-far from $k$-local in normalized Frobenius norm, which means $\sum_{P: |P| > k} |\alpha_P|^2 \geq \eps^2$, the second term in \Cref{eq:2-design-used} can be upper bounded by $1-t^2\eps^2$. The first term can be seen to quickly approach $0$ as $n$ grows. Hence,
\begin{equation*}
    \pr{\ket{\phi_{i,\ell}} \sim_k \ket{\phi_{i,j}}} \lessapprox 1- \frac{t^2\eps^2}{2} \, .
\end{equation*}
Note that we again have to choose $t$ small enough for the approximation to be correct. By repeating the algorithm $N= \mathcal O(t^{-2} \eps^{-2}) = \mathcal{O}(\eps^{-4})$ many times, we reduce the error probability to a small enough constant.  
Again, we use the Taylor expansion of $\mathrm{e}^{\mathrm{i}tH}$ to make the above reasoning precise. Choosing $t = \mathcal O(\eps)$ lets us control the higher order terms appearing in this expansion. 

\paragraph{Hamiltonian learning lower bound.}

Similarly to the reasoning behind \Cref{inf-thm:hamiltonian-locality-testing-hardness}, we follow~\cite{flammia2012quantum,haah2017sample,bubeck2020entanglement,lowe2022lower,fawzi2023lower,oufkir2023sample} and begin by identifying a distinguishing problem, whose existence we guarantee through a probabilistic argument, and that any successful general Hamiltonian learner can solve. We then establish lower bounds for that distinguishing task through information-theoretic arguments. 

To set up our distinguishing problem, let $O=\mathrm{diag}(+1,\ldots,+1, -1,\ldots,-1)$ be the diagonal $2^n\times 2^n$ matrix with half of the diagonal entries equal to $+1$ and the other half equal to $-1$. We consider Hamiltonians $H$ of the form $H=\varepsilon U O U^\dagger$, where $U$ is a Haar-random $n$-qubit unitary.
Using well known expressions for the first and second moments of the Haar measure, we show that the expected square of the normalized Frobenius distance between two such Hamiltonians satisfies $\mathbb{E}_{U,V\sim\mathrm{Haar}_n}[\frac{1}{2^n}\norm{H_U - H_V}_2^2]=2\varepsilon^2$, and that the second moment of this quantity is bounded as $\mathbb{E}_{U,V\sim\mathrm{Haar}_n}[\frac{1}{2^{2n}}\norm{H_U - H_V}_2^4]\leq 6\varepsilon^2$. Hölder's inequality then implies that the expected normalized Frobenius distance satisfies $\mathbb{E}_{U,V\sim\mathrm{Haar}_n}[\frac{1}{\sqrt{2^n}}\norm{H_U - H_V}_2] > 1.1\varepsilon$. 
Combining this with a Lipschitz concentration argument, we conclude that $\frac{1}{\sqrt{2^n}}\norm{H_U - H_V}_2>\varepsilon$ holds with probability $\geq 1 - \exp(-\Omega(2^{2n}))$.
Therefore, by a union bound, there exists a set of $M=\exp(\Omega(4^n))$ unitaries $U_x$, $1\leq x\leq M$, such that the Hamiltonians $H_x= \varepsilon U_x O U_x^\dagger$ are pairwise $\varepsilon$-far apart \edit{w.r.t.~}{with respect to }$\frac{1}{\sqrt{2^n}}\norm{\cdot}_2$. In particular, any algorithm that can learn an unknown Hamiltonian to average-case accuracy $\varepsilon$ is able to distinguish between these $M$ candidate Hamiltonians. Via Fano's inequality, this implies the mutual information lower bound $\mathcal{I}(X:Y)\geq \Omega(\log M) \geq \Omega(4^n)$, where $X\sim \mathrm{Uniform}([M])$ and where the random variable $Y$ describes the outcomes observed by the learner.

We then provide complementary mutual information upper bounds for two different scenarios. First, we consider coherent procedures that use the Hamiltonian time evolution sequentially interspersed with control channels. The mutual information can be upper bounded by the Holevo information $\chi$ \cite{holevo1973bounds}. This quantity can be decomposed as $\chi= \sum_{k=1}^N \chi_k$ where $\chi_k$ reflects the amount of information acquired after the $k$-th use of the Hamiltonian time evolution. Using standard properties of the von Neumann entropy, we can bound $\chi_k\leq 2n$. Comparing this with the previous mutual information lower bound, we conclude that any such learner has to query the Hamiltonian time evolution at least $N\geq \Omega ( 4^n/n)  $ many times. 

Second, for incoherent learners that use neither an auxiliary system nor adaptivity in their choice of experiments, we decompose the overall mutual information as $\mathcal{I}(X:Y) \le \sum_{\ell=1}^N \mathcal{I}(X:Y_\ell)$, where the random variable $Y_\ell$ describes the measurement outcome that the learner observes in the $\ell^{\mathrm{th}}$ experiment. That is, we have $\mathbb{P}[Y_\ell = y_\ell|X=x] = \lambda_{y_\ell}\bra{\phi^\ell_{y_\ell}}\mathrm{e}^{-it_\ell H_x}\rho_\ell \mathrm{e}^{it_\ell H_x}\ket{\phi^\ell_{y_\ell}}$, where $t_\ell$ is the evolution time, $\rho_\ell$ is the input state, and the measurement is described by $\{\lambda_{y_\ell}\proj{\phi^\ell_{y_\ell}}\}_{y_\ell}$. Using that $H_x^2=\eps^2\dI$ and therefore $\mathrm{e}^{\mathrm{i}tH_x}=\cos(t\eps)\dI+ \mathrm{i} \sin(t\eps)U_xOU_x^\dagger $, we show via a careful analysis that $\mathcal{I}(X:Y_\ell)\leq \tilde{\mathcal{O}}(t_\ell \varepsilon)$ and hence we obtain $\mathcal{I}(X:Y)\leq \mathcal{O}(\varepsilon\sum_{\ell=1}^N t_\ell + \sqrt{\mathrm{poly}(2^n)/M})$. Comparing this to our Fano-based mutual information lower bound of $\mathcal{I}(X:Y)\geq \Omega(4^n)$ and recalling that $M$ is doubly exponential in $n$, we obtain the lower bounds on the total evolution time and the query complexity stated in \Cref{inf-thm:hardness-hamiltonian-learning}.

\subsection{Directions for future work}

Motivated by the question of how to test whether an unknown Hamiltonian is $k$-local, we proposed a framework for Hamiltonian property testing.
Here, we considered different ways of measuring distances between Hamiltonians.
With worst-case distance measures, exemplified by the operator norm, we showed that exponentially many queries as well as exponentially long time are required to solve the locality testing problem.
In contrast, for the normalized Frobenius norm, leading to an average-case notion of distance, we gave a broadly applicable Hamiltonian testing algorithm that uses only single-copy measurements on short-time evolutions of certain randomized input states, and that is resource-efficient \edit{w.r.t.~}{with respect to }the number of queries, total evolution time, and computation time involved.
Finally, still in the regime of average-case distances, we showed that learning is exponentially harder than testing.

There are several promising ways of extending our Hamiltonian property testing framework, for instance by changing the notions of distance and access.
First, one may consider other physically motivated notions of distance, such as quantum Wasserstein distances \cite{depalma2021quantum} or distances relative to (some distribution over) a set of input states and a set of output observables of interest. The former would lead to Hamiltonian testing (and learning) that takes an underlying locality structure into account, wheres the latter would be reminiscent of classical shadows \cite{huang2020predicting, huang2023learning} and shadow tomography \cite{aaronson2020shadow, buadescu2021improved}.
Second, one may change the form of access to the unknown Hamiltonian from time evolution access to access to copies of a Gibbs state, an access model already well studied in Hamiltonian learning (compare \Cref{subsection:related-work}).
Also, exploring the possibility for quantitative improvements in our bounds seems important.
For instance, concerning upper bounds, coherent quantum algorithms may be able to achieve Heisenberg-limited scaling for testing, as they already have for Hamiltonian learning \cite{huang2023heisenberg, li2023heisenberglimited} and unitary tomography \cite{haah2023queryoptimal, zhao2023learning}.
Regarding lower bounds, it would be interesting to generalize  \Cref{inf-thm:hamiltonian-locality-testing-hardness,inf-thm:hardness-hamiltonian-learning} to ancilla-assisted incoherent quantum algorithms.
Finally, we highlight another possible extension to the Hamiltonian testing framework proposed here: Whereas we focus on \edit{Hamiltonian testing}{property testing for time-independent Hamiltonians}, more general \del{property }testing questions for \ins{time-dependent Hamiltonians or even} GKLS generators \cite{lindblad1976generators, gorini1976completely} from access to the generated quantum dynamical semigroups could be of interest. 

\paragraph{Note added.} After the first version of our work appeared on the arXiv, Francisco Escudero Gutiérrez shared with us his approach to tolerant Hamiltonian property testing and to local Hamiltonian learning based on entangled inputs \cite{escudero2024testing}. We are grateful for this exchange, which motivated us to improve upon the first version of our work, tightening the analysis underlying our upper bounds to remove the $n$-dependence in the total evolution time and number of experiments as well as to achieve a tolerant version of our tester.

\section{Preliminaries}

\subsection{Notation and basic definitions}
Let $\log$ denote the natural logarithm and $\log_2$ the logarithm to base $2$. For compactness, for $n \in \mathbb N$ we will abbreviate $[n]:=\{1, \ldots, n\}$. We write $\mathbf{1}(\mathcal X)$ for the indicator function on the set or event $\mathcal X$. For a complex number $z \in \mathbb C$, we will denote by $\Re(z)$ its real part and $\Im(z)$ its imaginary part. 

We will make extensive use of the Schatten $p$-norms $\|\cdot\|_p$, for $p \in \mathbb N \cup \{\infty\}$. They are defined for any $n \times m$ matrix $X$ as
\begin{equation*}
    \|X\|_p := \operatorname{Tr}[|X|^p]^{\frac{1}{p}} =\operatorname{Tr}[(\sqrt{X^\dagger X})^p]^{\frac{1}{p}} \, .
\end{equation*}
Here, $X^\dagger$ is the Hermitian conjugate of $X$. The case $p=\infty$ is the operator norm of $X$. The Schatten $2$-norm is also known as the Frobenius norm, whereas the Schatten $1$-norm also goes by the names of trace or nuclear norm.  Here, we can employ different measures of distance between Hamiltonians. Moreover, we consider the normalized Schatten $p$-norms $\frac{1}{2^{n/p}}\lVert H\rVert_p$. And finally, we use $\lVert H\rVert_{\mathrm{Pauli},p} \equiv \left(\sum_{P\in\mathds{P}_n} \lvert\alpha_P\rvert^p\right)^{1/p}$ to denote the norm induced by the $\ell_p$-norm of the coefficient vector of $H$. Note that $\frac{1}{\sqrt{2^n}}\lVert H\rVert_2 = \lVert H\rVert_{\mathrm{Pauli},2}$ by Parseval's identity.

Depending on context, we will mean by a quantum state either a unit vector $\ket{\psi} \in \mathbb C^d$ for some appropriate dimension $d \in \mathbb N$ or a density matrix $\rho$, i.e., a positive-semidefinite $d \times d$ matrix of unit trace. A linear map $\mathcal N$ from $d_1 \times d_1$ to $d_2 \times d_2$ matrices will be called positive if it maps positive-semidefinite matrices to positive-semidefinite matrices. Moreover, it will be called completely positive if $\mathcal{N} \otimes \mathrm{id}_n$ is positive for all $n$, where $\mathrm{id}_n$ is the identity map on an additional $n$-dimensional system. If a completely positive map is in addition trace-preserving, we will call it a quantum channel. For general background on quantum information theory, we refer the reader to a textbook such as \cite{watrous2018theory}.

Finally, we will need some objects from classical information theory. Given two probability distributions $P$, $Q$ on a finite alphabet $\mathcal Y$, their total variation distance is
\begin{equation*}
    \TV(P, Q) := \frac{1}{2} \sum_{y \in \mathcal Y} |P(y) - Q(y)| \, .
\end{equation*}
Their Kullback-Leibler divergence is 
\begin{equation*}
    \KL(P \| Q) := \sum_{y \in \mathcal Y} P(y) \log\left(\frac{P(y)}{Q(y)}\right)
\end{equation*}
if $\operatorname{supp} P \subseteq \operatorname{supp} Q$, where we define $0 \log 0 := 1$, and $+\infty$ otherwise.

Given two random variables $X$, $Y$ on finite alphabets $\mathcal X$, $\mathcal Y$, respectively, and joint probability distribution $P_{XY}$ on $\mathcal X \times \mathcal Y$ with marginals $P_X$ on $\mathcal X$ and $P_Y$ on $\mathcal Y$, their mutual information is 
\begin{equation*}
    \mathcal I(X:Y) := \KL(P_{XY} \| P_X \otimes P_Y)\, .
\end{equation*}

\subsection{Mutually unbiased bases of stabilizer states}

We define the Pauli matrices in the standard way
\begin{equation*}
    X = \begin{pmatrix}
        0 & 1 \\ 1 & 0
    \end{pmatrix}\, , \qquad
    Y = \begin{pmatrix}
        0 & -\mathrm{i} \\ \mathrm{i} & 0
    \end{pmatrix}\, , \qquad
    Z = \begin{pmatrix}
        1 & 0 \\ 0 & -1
    \end{pmatrix}\, .
\end{equation*}
The set of $n$-strings formed by the Pauli operators together with the identity matrix will be written as $\mathds{P}_n=\{\dI, X,Y,Z\}^{\otimes n}$. Note that unlike the set 
\begin{equation*}
    \left\{ \mathrm{e}^{\mathrm{i}\theta \pi/2 }\sigma_1\otimes \dots \otimes \sigma_n \;\big|\; \theta = 0,1,2,3 , \; \sigma_i \in \{\dI, X,Y,Z\}\right\}\, ,
\end{equation*} 
which is the Pauli group, the set $\mathds{P}_n$ is not a group. We will also need the quotient group of the Pauli group with its centralizer. We will write this Abelian group as
\[\mathbf{P}_n=\left\{ \mathrm{e}^{\mathrm{i}\theta \pi/2 }\sigma_1\otimes \dots \otimes \sigma_n \;\big|\; \theta = 0,1,2,3 , \; \sigma_i \in \{\dI, X,Y,Z\} \right\}/\{\pm1, \pm \mathrm{i}\} \, . \] 
There is a bijection between elements in $\mathbf{P}_n$ and elements in $\mathds{P}_n$, and we will write $P(g)$ for the element in $\mathds{P}_n$ corresponding to $g \in \mathbf{P}_n$.

We will make extensive use of the following known two lemmas, see e.g. \cite{bandyopadhyay2002new,flammia2020efficient}. 

\begin{lemma}
	$\mathbf{P}_n$ can be covered by $d+1$ stabilizer  groups $G_1,\dots, G_{d+1}$ satisfying for all $i\neq j$:
	\begin{itemize}
		\item $|G_i|=d$,
		\item $C_{G_i}=G_i$,
		\item $G_i\cap G_j= \{\dI\}$.
	\end{itemize}
 Here, $C_{G_i} =\{g \in \mathbf{P}_n: [P(g), P(h)] = 0~\forall h \in G_i \}$. It can be thought of as the centralizer of $G_i$ in the Pauli group, where we are only interested in elements with a fixed choice of sign.
\end{lemma}

From these stabilizer groups, we can construct sets of pure states: However, let us first note that $\{P(g) : g \in G_i\}$ are not groups. We can, however, turn them into groups: there are signs $\zeta_{i,g} \in \{\pm 1\}$ such that $\{\zeta_{i,g} P(g) : g \in G_i\}$ is an Abelian group, the stabilizer group. 
There are two ways to see this. Either, we start with a set of generators for $G_i$ and keep track of the signs when multiplying the corresponding $P(g)$ in order to generate the group. Alternatively, we choose a common eigenstate $\ket{\varphi}$ of all the $P(g)$, $g \in G_i$, which exists since all the $P(g)$ commute and define $P(g) \ket{\varphi} = \zeta_{i,g}\ket{\varphi}$.
In this case, $\{S(g) :=\zeta_{i,g} P(g) : g \in G_i\}$ will be a stabilizer group for $\ket{\varphi}$.

\begin{lemma}
	Let $G\in \{G_1,\dots, G_{d+1}\}$. The set $\cM_G  =  \{M_G^r  =  \frac{1}{d}\sum_{p\in G} (-1)^{p\circ r}S(p)\}_{r\in A_G}$ forms an orthonormal basis consisting of rank-$1$ stabilizer states.
\end{lemma}
Here we use the notation $A_G= \mathbf{P}_n/G$ and $p\circ r=0$ if $P(p) P(r)=P(r)P(p)$ and $p\circ r=1$ if $P(p)P(r)=-P(r)P(p)$ (note that we could have put arbitrary signs in front of the $P(p)$, $P(r)$ without changing the value of $p\circ r$). 
To not overload the notation, we define $A_{G_i}= \{ r_j^i\}_{j=1}^d$  and introduce the notation
\begin{align} \label{eq:stabilizer_MUB}
	\proj{\phi_{i,j}} = M_{G_i}^{r_j^i}= \frac{1}{d}\sum_{p\in G_i} (-1)^{p\circ r^i_j}S(p)
\end{align}
Now we have $d+1$ mutually unbiased bases (MUBs) $\cB_i=\{\ket{\phi_{i,j}}\}_{j=1}^d$ for $i=1, \dots, d+1$. We will use the fact that MUBs of stabilizer states form a $2$-design. $2$-designs were used for the problem of testing the mixedness of states by \cite{yu2021sample}. Stabilizer states and measurements were used for the problem of Pauli channel learning by \cite{flammia2020efficient}. 

\begin{proposition}[Pauli group, MUB and $2$-design]
    The bases $\left\{\cB_i=\{\ket{\phi_{i,j}}\}_{j=1}^d\right\}_{i=1}^{d+1}$ form an MUB and a $2$-design. That is:
\begin{align*}
	\frac{1}{d(d+1)}\sum_{i=1}^{d+1}\sum_{j=1}^d \proj{\phi_{i,j}}\otimes \proj{\phi_{i,j}}= \frac{\dI+ F}{d(d+1)}
\end{align*}
where $F=\sum_{x,y} \ket{xy}\bra{yx}$ is the flip operator.
\end{proposition}

\begin{proof}One way to see this is to apply \cite[Theorem 1]{klappenecker2005mutually}. We only need to prove that:
\begin{align*}
   \sum_{i,j} \frac{1}{d(d+1)}\sum_{k,\ell} \frac{1}{d(d+1)} |\spr{\phi_{i,j}}{\phi_{k,\ell}}|^4= \int_{\Haar} |\spr{0}{\phi}|^4 d\phi= \frac{2}{d(d+1)}.
\end{align*}
This identity can be checked easily since $|\spr{\phi_{i,j}}{\phi_{k,\ell}}|=\frac{1}{\sqrt{d}}$ if $i\neq k$ and $|\spr{\phi_{i,j}}{\phi_{i,\ell}}|=\mathbf{1}(\{j=l\})$. Indeed, we can check  $|\spr{\phi_{i,j}}{\phi_{k,\ell}}|=\frac{1}{\sqrt{d}}$ if $i\neq k$:
\begin{align*}
    |\spr{\phi_{i,j}}{\phi_{k,\ell}}|^2 &= \tr\left( \proj{\phi_{i,j}} \proj{\phi_{k,\ell}}\right)
    \\&= \tr\left( \frac{1}{d}\sum_{p\in G_i}(-1)^{p\circ r^i_j}S(p) \cdot  \frac{1}{d}\sum_{p^\prime\in G_k}(-1)^{p^\prime\circ r^k_l} S(p^\prime)\right)
        \\&= \frac{1}{d}\sum_{p\in G_i\cap G_k}(-1)^{p\circ r^i_j}(-1)^{p\circ r^k_l}=\frac{1}{d}
\end{align*}
since $G_i\cap G_k=\{\dI\}$. Alternatively, we could have argued using \cite[Theorem 3]{klappenecker2005mutually}.
\end{proof}

Finally, the operation $p \circ q$ has a nice property that we'll need later, which has already been used, for example, in \cite{flammia2020efficient}. 
\begin{lemma} \label{lem:summing-signs-over-groups}
    Let $q \in \mathbf P_n$ and let $G$ be any subgroup of $\mathbf P_n$. Then
    \begin{equation*}
        \frac{1}{|G|} \sum_{p \in G} (-1)^{p \circ q} = \mathbf{1}(q \in C_G).
    \end{equation*}
\end{lemma}

\subsection{Problem statement: Hamiltonian property testing}

We consider $n$-qubit Hamiltonians $H$ expanded in the Pauli basis, $H=\sum_{P\in\mathds{P}_n} \alpha_P P$, with $\mathds{P}_n=\{\dI, X,Y,Z\}^{\otimes n}$ the set of $n$-qubit Paulis and with (real) coefficients $\alpha_P = \nicefrac{\tr[HP]}{2^n}$. 
Throughout this work, the Hamiltonian properties of interest are characterized by a subset $S\subseteq \mathds{P}_n$, and we say that $H$ has property $\Pi_S$, write $H\in \Pi_S$, if $\alpha_P=0$ for all $P\not\in S$.
For instance, the property of being (at most) $k$-local is characterized by the subset $S_{k\mathrm{-loc}} = \{P\in\mathds{P}_n: \lvert P\rvert \leq k\}$, where we use $\lvert P\rvert$ to denote the weight of $P$, that is, the number of non-identity tensor factors.
Other examples of properties that our framework encompasses are geometric locality or having a given interaction graph.

We now define Hamiltonian property testing of $\Pi_S$ as the task of deciding, given access to the time evolution according to an unknown Hamiltonian $H$, whether $H$ has property $\Pi_S$ or whether $H$ is far from all Hamiltonians that have property $\Pi_S$.

\begin{problem}[Hamiltonian property testing]\label{def:hamiltonian-testing}
    Given a property $\Pi_S$ associated to a subset $S\subseteq \mathds{P}_n$, a norm $\nnorm{\cdot}$, and an accuracy parameter $\varepsilon\in (0,1)$, we denote by $\mathcal{T}_{\nnorm{\cdot}}^{\Pi_S}(\varepsilon)$ the following Hamiltonian property testing problem:
    Given access to the time evolution according to an unknown Hamiltonian $H$, decide, with success probability $\geq 2/3$, whether
    \begin{enumerate}
        \item[(i)] $H$ has property $\Pi_S$, that is, $H\in \Pi_S$, or 
        \item[(ii)] $H$ is $\varepsilon$-far from having property $\Pi_S$, that is, $\forall \Tilde{H}\in \Pi_S$: $\nnorm{ H - \Tilde{H}}\geq \varepsilon$\, .
    \end{enumerate}
     If $H$ satisfies neither (i) nor (ii), then any output of the tester is considered valid. 
\end{problem}

\begin{remark}
    A short discussion of the physical interpretation for the different norms used above is in order. 
    Among the Schatten $p$-norms, we highlight the Schatten $\infty$-norm (aka operator norm) as a worst-case distance. Namely, in \Cref{appendix:operator-norm}, we show that, at least for short evolution times, the Hamiltonian distance $\norm{H - \Tilde{H}}_\infty$ is tightly related to, among others, the worst-case output fidelity $\max_{\ket{\psi}}\lvert\bra{\psi}\mathrm{e}^{-\mathrm{i}tH}\mathrm{e}^{\mathrm{i}t\Tilde{H}}\ket{\psi}\rvert^2$ and the diamond norm distance between the unitary time evolution channels $\mathrm{e}^{-\mathrm{i}tH}(\cdot)\mathrm{e}^{\mathrm{i}tH}$ and $\mathrm{e}^{-\mathrm{i}t\Tilde{H}}(\cdot)\mathrm{e}^{\mathrm{i}t\Tilde{H}}$.

    Second, to illustrate the relevance of the normalization factor in normalized Schatten $p$-norms, we single out the normalized Frobenius norm and interpret it as an average-case distance measure. Here, we demonstrate in \Cref{appendix:normalized-frobenius-norm} that $\frac{1}{\sqrt{2^n}}\norm{H - \Tilde{H}}_2$ is, again for short times, connected to the average output fidelity $\mathbb{E}_{\ket{\psi}\sim\mathrm{Haar}_n}\left[\lvert\bra{\psi}\mathrm{e}^{-\mathrm{i}tH}\mathrm{e}^{\mathrm{i}t\Tilde{H}}\ket{\psi}\rvert^2\right]$, the normalized Frobenius norm distance between the time evolutions $\mathrm{e}^{-\mathrm{i}tH}$ and $\mathrm{e}^{-\mathrm{i}t\Tilde{H}}$, and other average-case distance measures.

    Finally, the norms $\lVert \cdot\rVert_{\mathrm{Pauli},p}$, in particular for $p=2,\infty$, have featured in recent work on Hamiltonian learning, see \Cref{subsection:related-work}.
    They make intuitive sense for learning and testing tasks in which the interest is in estimating or validating properties of interaction strength parameters in a Hamiltonian. 
\end{remark}

The central goal of this work is to understand how easy or hard this Hamiltonian property testing problem is.
More specifically, we are interested in the number of queries to the time evolution as well as in the total evolution time necessary and sufficient to solve the testing task. Additionally, we consider whether one can successfully approach Hamiltonian testing via Hamiltonian learning.

\subsection{Types of strategies for Hamiltonian property testing}\label{subsection:types-of-strategies}

In order to solve the Hamiltonian property testing problem, we consider several different scenarios, depending on what kind of access to the quantum system and how many additional resources we grant the tester.

\paragraph{Incoherent strategies.}
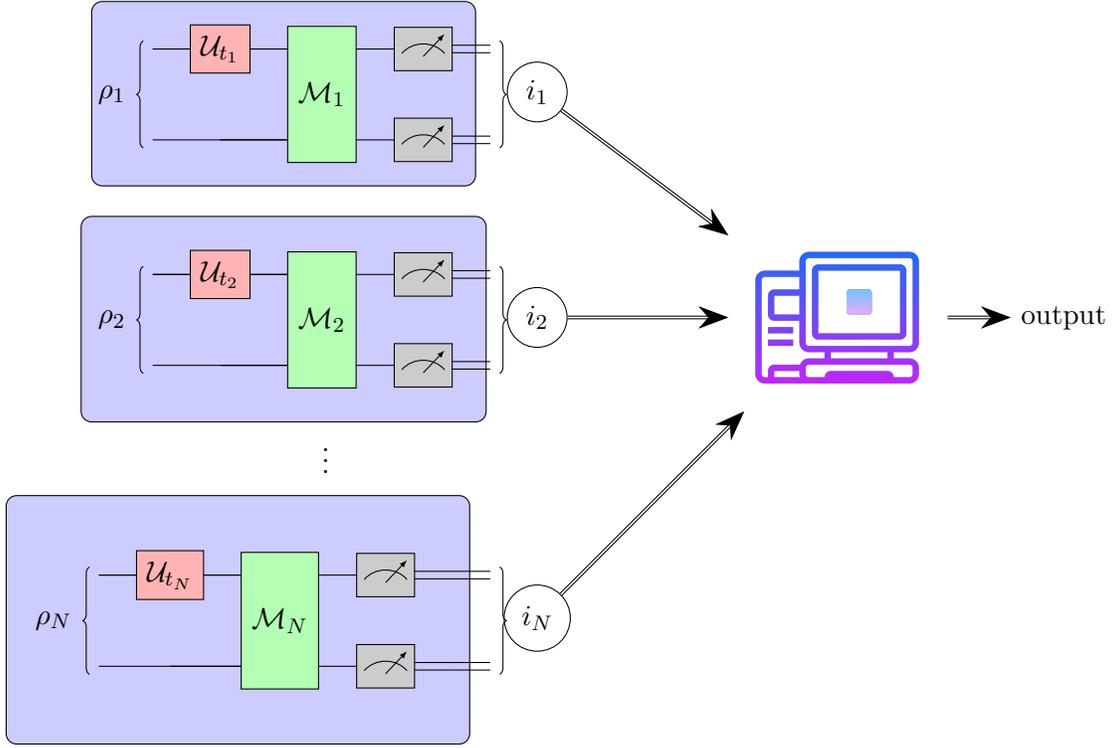
\begin{figure}
	\centering
	\tikzset{
		meter/.append style={fill=black!20}
	}
	\begin{tikzpicture}
	\node[circle, draw] at (4, 7)   (a) {$i_1$};
	\node[circle, draw] at (4, 4)   (b) {$i_2$};
	\node[circle, draw] at (4, 0)   (c) {$i_N$};
	\filldraw[black] (4, 7.1)  node[anchor=east]{	\begin{quantikz}[thin lines] 
		\gategroup[wires=2,steps=5,style={rounded corners,fill=blue!20,inner sep=5pt},background]{}&\lstick[2]{$\rho_{1}$} & \gate[style={fill=red!30}]{\cU_{t_1}} \qw& \gate[2, style={fill=green!30}]{\cM_1} \qw &\meter{}   &\cw\rstick[2]{ } \\
		&&  \qw&  \qw &\meter{}   &\cw
		\end{quantikz}};
	\filldraw[black] (4,4.1)  node[anchor=east]{\begin{quantikz}[thin lines] 
		\gategroup[wires=2,steps=5,style={rounded corners,fill=blue!20,inner sep=9pt},background]{}&\lstick[2]{$\rho_{2}$} & \gate[style={fill=red!30}]{\cU_{t_2}} \qw& \gate[2, style={fill=green!30}]{\cM_2} \qw &\meter{}   &\cw\rstick[2]{ } \\
		&&  \qw&  \qw &\meter{}   &\cw
		\end{quantikz}};
	\filldraw[black] (1,2.2) node[anchor=west]{$\vdots$}; 
	\filldraw[black] (4,0.1)  node[anchor=east]{
		\begin{quantikz}[thin lines] 
		\gategroup[wires=2,steps=5,style={rounded corners,fill=blue!20,inner sep=17pt},background]{}&\lstick[2]{$\rho_{N}$} & \gate[style={fill=red!30}]{\cU_{t_N}} \qw& \gate[2, style={fill=green!30}]{\cM_N} \qw &\meter{} &\cw  &\cw\rstick[2]{ } \\
		&&  \qw&  \qw &\meter{}   &\cw&\cw
		\end{quantikz}};
	\node[] at (8, 4)   (d) {$\includegraphics[width=.16\linewidth]{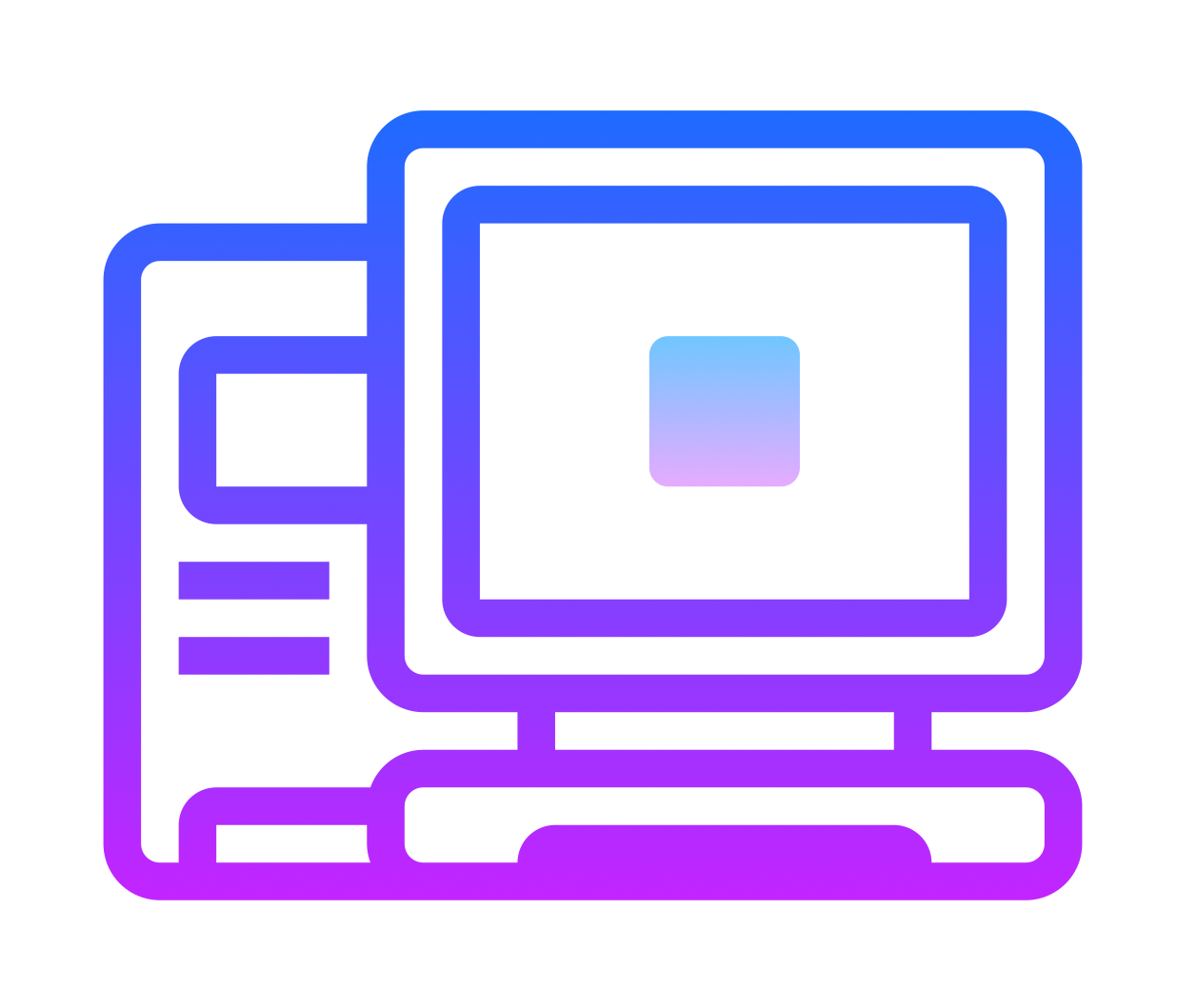} $};
	\node[] at (11, 4) (e) { output};
	\draw[->] [-{Implies},double][-{Stealth[scale=2]}]   (d) --  (e)   ;
	\draw[->] [-{Implies},double][-{Stealth[scale=2]}]   (c) --  (d)   ;
	\draw[->]   [-{Implies},double][-{Stealth[scale=2]}]  (b) -- (d)   ;
	\draw[->]   [-{Implies},double][-{Stealth[scale=2]}]  (a) -- (d)   ;
	\end{tikzpicture}
	\caption{Illustration of a non-adaptive   incoherent  strategy for learning/testing properties of a Hamiltonian $H$ from its time evolution channel $\cU_t(\cdot)= \mathrm{e}^{-\mathrm{i}tH} (\cdot) \mathrm{e}^{\mathrm{i}tH}$. It is called ancilla-free if the auxiliary systems have dimension $1$, otherwise it is called ancilla-assisted. 
	The classical computer  processes the observations $(i_1, \dots, i_N)$ to distinguish between two hypotheses $H_0/H_1$ (in testing) or to produce an approximate Hamiltonian $\hat{H}$ (in learning).}
	\label{fig: NonAdap-ancill-ass}
\end{figure}
\begin{figure}
	\centering
	
	\tikzset{
		meter/.append style={fill=black!20}
	}
    \scalebox{0.93}{
	\begin{tikzpicture}
	\node[circle, draw] at (4, 6)   (a) {$i_1$};
	\node[circle, draw] at (5.5, 2.7)   (b) {$i_2$};	
	\node[circle, draw] at (7.4, -1.2)   (c) {$i_N$};	
	\draw[-]  [-{Implies},double]  (a) -- (4,4.5) -- node[pos=0.9,below=-0.1cm]{$i_1$}(-1,4.5) --(-1,2.8)--(-0.65, 2.8);
	\filldraw[black] (a)++(0,0.1)  node[anchor=east]{	\begin{quantikz}[thin lines] 
		\gategroup[wires=2,steps=5,style={rounded corners,fill=blue!20,inner sep=5pt},background]{}&\lstick[2]{$\rho_{1}$} & \gate[style={fill=red!30}]{\cU_{t_1}} \qw& \gate[2, style={fill=green!30}]{\cM_1} \qw &\meter{}   &\cw\rstick[2]{ } \\
		&&  \qw&  \qw &\meter{}   &\cw
		\end{quantikz}};
	\filldraw[black] (b)++(0,0.1)  node[anchor=east]{	\begin{quantikz}[thin lines] 
		\gategroup[wires=2,steps=5,style={rounded corners,fill=blue!20,inner sep=9pt},background]{}&\lstick[2]{$\rho_{2}^{i_1}$} & \gate[style={fill=red!30}]{\cU_{t_2}} \qw& \gate[2, style={fill=green!30}]{\cM_2^{i_1}} \qw &\meter{}   &\cw\rstick[2]{ } \\
		&&  \qw&  \qw &\meter{}   &\cw
		\end{quantikz}};
	\draw[-] [-{Implies},double]  (b)  --  (5.5,1) --(5.4,1)  (3.5,1)-- node[pos=0.6,below=-0.08cm]{$i_1, \dots, i_{N-1}$}(-0.3, 1)--(-0.3, -1.05)  -> (0.15, -1.05)  ;
	\draw[dotted]  [double] (5.4, 1) -- (3.5, 1)     ;
	\filldraw[black] (c)++(-0.05,0.1)   node[anchor=east]{	\begin{quantikz}[thin lines] 
		\gategroup[wires=2,steps=5,style={rounded corners,fill=blue!20,inner sep=17pt},background]{}&\lstick[2]{$\rho_{N}^{i_{<N}}$} & \gate[style={fill=red!30}]{\cU_{t_N}} \qw& \gate[2, style={fill=green!30}]{\cM_N^{i_{<N}}} \qw &\meter{} &\cw  &\cw\rstick[2]{ } \\
		&&  \qw&  \qw &\meter{}   &\cw&\cw
		\end{quantikz}};
	
	\node[] at (10, 2.7)   (d) {$\includegraphics[width=.16\linewidth]{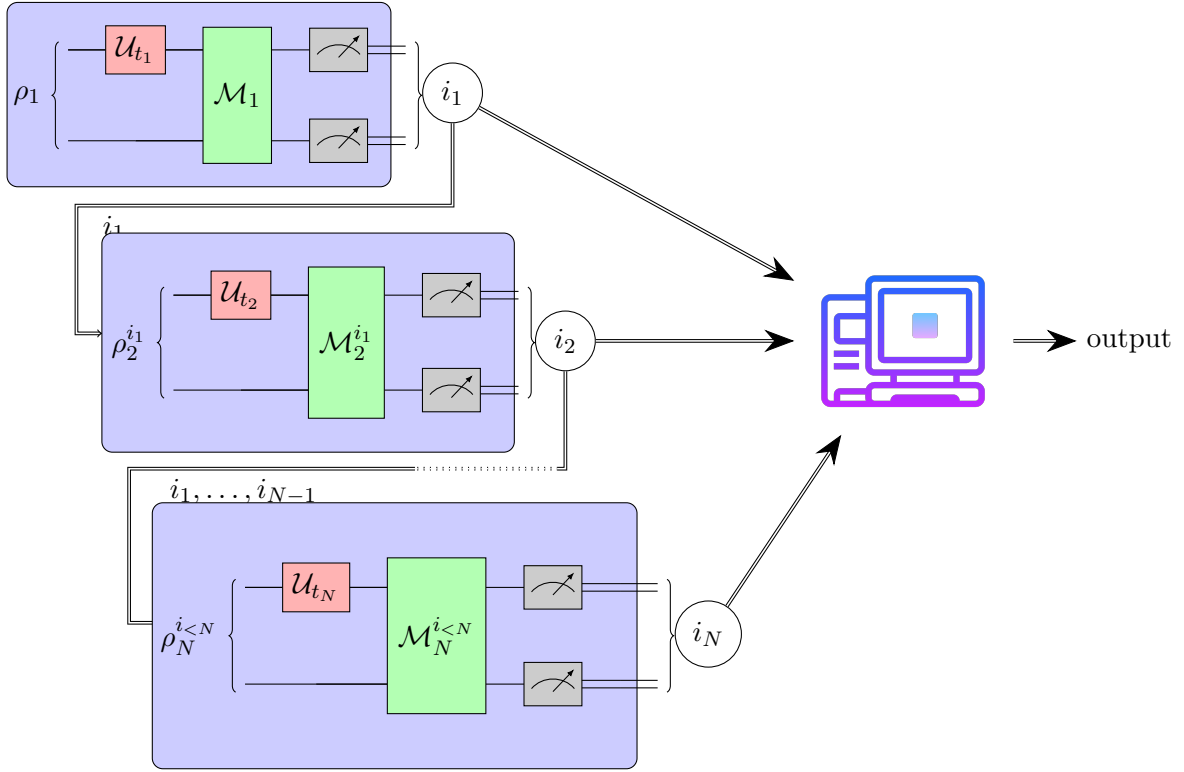} $};
	\draw[->] [-{Implies},double][-{Stealth[scale=2]}]   (c) --(d)   ;
	\draw[->]   [-{Implies},double][-{Stealth[scale=2]}]  (b) -- (d)   ;
	\draw[->]   [-{Implies},double][-{Stealth[scale=2]}]  (a) -- (d)   ;
	\node[] at (13, 2.7) (e) { output};
	\draw[->] [-{Implies},double][-{Stealth[scale=2]}]   (d) --  (e)   ;
	\end{tikzpicture}}
	\caption{Illustration of an adaptive   incoherent  strategy for learning properties of a Hamiltonian $H$ from its time evolution channel $\cU_t(\cdot)= \mathrm{e}^{-\mathrm{i}tH}(\cdot) \mathrm{e}^{\mathrm{i}tH}$. It is called ancilla-free if the auxiliary systems have dimension $1$, otherwise it is called ancilla-assisted. 
	The classical computer  processes the observations $(i_1, \dots, i_N)$ to distinguish between two hypotheses $H_0/H_1$ (in testing) or to produce an approximate Hamiltonian $\hat{H}$ (in learning).}
	\label{fig: Adap-ancill-ass}
\end{figure}
\begin{figure}
	\centering
    \scalebox{0.9}{
	\tikzset{
		meter/.append style={fill=black!20}
	}\begin{tikzpicture}
    \node[circle,draw] at (5.4,0) (i) {$i$};
    \filldraw[black] (i)++(0.1,0.1)   node[anchor=east]{	\begin{quantikz}[thin lines]
		\gategroup[wires=2,steps=9,style={rounded corners,fill=blue!20,inner sep=5pt},background]{}	&\lstick[2]{$\rho$}& \gate[style={fill=red!30}]{\cU_{t_1}} & \gate[2, style={fill=green!30}]{\cN_1} & \ \ldots\ \qw &  \gate[2, style={fill=green!30}]{\cN_{N-1}}&\gate[style={fill=red!30}]{\cU_{t_N}}  &
		\gate[2, style={fill=green!30}]{\cM} & \meter{}&\cw \rstick[2]{} &\\
		&& \qw& & \ \ldots\ \qw &  \qw& \qw&&\meter{}& \cw& 
		\end{quantikz}};
	\node[] at (7.3,0) (b) {$\includegraphics[width=.08\linewidth]{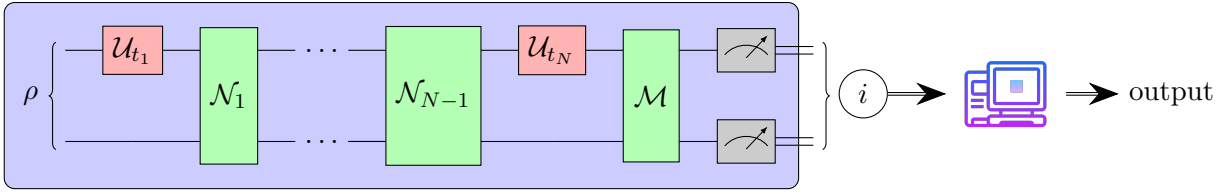}$};
	\node[] at (9.5,0) (c) {output};
	\draw[->] [-{Implies},double][-{Stealth[scale=2]}]  (i) -- (b)  ;
	\draw[->] [-{Implies},double][-{Stealth[scale=2]}]  (b)-- (c) ;
	\end{tikzpicture}}
	\caption{Illustration of a coherent strategy for learning properties of a Hamiltonian $H$ from its time evolution channel $\cU_t(\cdot)= \mathrm{e}^{-\mathrm{i}tH} (\cdot) \mathrm{e}^{\mathrm{i}tH}$. The classical computer processes the observation $i$ to distinguish between two hypotheses $H_0/H_1$ (in testing) or to produce an approximate Hamiltonian $\hat{H}$ (in learning).}
	\label{fig:coherent}
\end{figure}
We start by presenting \emph{incoherent} strategies. That is, the tester can use the quantum system only once per step. Therefore, in round $k$ of the overall procedure, the tester can prepare a quantum state $\rho_k$, let the time evolution $\mathcal U_t(\cdot) = \mathrm{e}^{- \mathrm{i} t H} (\cdot) \mathrm{e}^{\mathrm{i} t H}$ run for a chosen time $t_k$, and finally perform a measurement $\mathcal M_k = \left\{\lambda^{(k)}_i \proj{\phi^{(k)}_i}\right\}_{i \in \mathcal I_k}$, where $\mathcal I_k$ is a set of measurement outcomes.
Note that we could have taken a general positive operator-valued measure (POVM) here, but that without loss of generality, we can assume the POVM to consist of rank-one elements. If the input state, the duration of the time evolution, and the choice of measurement are allowed to depend on the outcomes of previous measurements, the incoherent strategy is \emph{adaptive}, otherwise it is \emph{non-adaptive}. We can also add ancilla qubits to the strategy by preparing the input state $\rho_k$ on a $(d \times d_{\mathrm{aux}})$-dimensional system instead of restricting to dimension $d$, subsequently letting the time evolution act as $\mathcal U_{t_k} \otimes \mathrm{id}$, where the identity acts on the ancilla qubits, and finally measuring both systems with $\mathcal M_k$. If we add the ancilla, we call the strategy \emph{ancilla-assisted}, otherwise we will speak of \emph{ancilla-free} strategies. 
See \Cref{fig: NonAdap-ancill-ass} for an illustration of non-adaptive strategies and \Cref{fig: Adap-ancill-ass} for an illustration of adaptive strategies. 
\paragraph{Coherent strategies.}
The second type of strategies are the \emph{coherent} strategies. Here, the tester is allowed to access the time-evolution multiple times before measuring, possibly interleaving these time evolutions with the application of quantum channels. Thus, the tester will prepare an initial quantum state $\rho$ on a $(d \times d_{\mathrm{aux}})$-dimensional quantum system and then choose times $t_1, \ldots, t_N$ and quantum channels $\mathcal N_1, \ldots, \mathcal N_{N-1}$ acting on quantum systems of dimension $d \times d_{\mathrm{aux}}$ such that the output state before the measurement is
\begin{equation*}
    \rho^{\mathrm{output}} = [\mathcal U_{t_N} \otimes \mathrm{id}]\circ  \mathcal N_{N-1}   \circ [\mathcal U_{t_{N-1}} \otimes \mathrm{id}] \circ \ldots \circ \mathcal N_1 \circ [\mathcal U_{t_1} \otimes \mathrm{id}] (\rho) \, ,
\end{equation*}
where the identities act on the ancilla qubits. Finally, the tester measures $\rho^{\mathrm{output}}$ using a POVM acting on $(d \times d_{\mathrm{aux}})$-dimensional quantum systems. See \Cref{fig:coherent} for an illustration.

\section{Lower bounds for Hamiltonian locality testing}\label{sec:hamiltonian-testing-lower-bounds}

In this section, we will give lower bounds for the Hamiltonian locality testing problem, thereby proving \Cref{inf-thm:hamiltonian-locality-testing-hardness}. We will first consider lower bounds for incoherent testing strategies in \Cref{sec:lower-indep-setting}. Subsequently, we consider the more general coherent strategies in \Cref{sec:lower-seq-setting}, for which we can only give weaker bounds. 

\subsection{Incoherent setting} \label{sec:lower-indep-setting}

In this section, we will prove the following hardness result for Hamiltonian locality testing with respect to unnormalized Schatten norms:
\begin{theorem} \label{thm:lower-Schatten-indep} 
    Let $n\geq \Omega(1)$, $k\le \cO\left(\frac{n}{\log(n)}\right)$ and $p\geq 1$.  
    Suppose that $N\le \exp(\cO(n))/\eps^{\cO(1)}$. 
    The problem $\cT^{\mathrm{loc}}_{\|\cdot\|_p}(\eps)$, even under the additional promise that the unknown Hamiltonian $H$ satisfies $\tr[H]=0$ and $\norm{H}_\infty\leq 1$, requires an expected total evolution time of $\ex{\sum_{k=1}^N t_k}=\Omega\left( \frac{2^n}{\eps(n-\log\eps)}\right)$ and a total number of independent experiments $N=\Omega\left( \frac{2^{n}}{n}\right)$ in the adaptive ancilla-free incoherent setting.
\end{theorem}

In order to prove \Cref{thm:lower-Schatten-indep}, we will consider the following test: 
\begin{align}\label{Toy problem II}
    \cH_0: \cU_t(\rho)= \id(\rho)= \mathrm{e}^{-\mathrm{i}t\cdot 0}\rho \mathrm{e}^{\mathrm{i}t\cdot 0} ~~~~~~\text{vs} ~~~~~~~~ \cH_1: \cU_t(\rho)=  \mathrm{e}^{-\mathrm{i}tH}\rho \mathrm{e}^{\mathrm{i}tH}
\end{align}
where $H=\eta (\proj{v} -  \dI/d)$ for $\eta \ins{=\eta(\varepsilon)}>0$ \ins{a function of $\varepsilon$ to be chosen later,} and $\ket{v}=V\ket{0}, V\sim \Haar(d)$. We will need a couple of lemmas to determine the relation of this toy problem of distinguishing these two time evolutions to the problem $\cT^{\mathrm{loc}}_{\|\cdot\|_p}(\eps)$ that we care about.

First, we show the that the operator norm between $H$ and any fixed traceless Hamiltonian (of bounded operator norm) is bounded below with high probability. 
\begin{lemma}\label{lem:concentration-operator-norm}
	Let $d \geq 4$, $H=\eta (\proj{v} -  \dI/d)$ for $\eta>0$ and $\ket{v}=V\ket{0}, V\sim \Haar(d)$ and let $K$ be a fixed traceless Hamiltonian such that $\|K\|_\infty \le 1$. We have that
	\begin{align*}
		\pr{\|H-K\|_{\infty} \le \frac{\eta}{2}}\le  \exp\left(-\frac{d}{1728}\right).
	\end{align*} 
\end{lemma}
\begin{proof}
	Consider the function:
	\begin{align*}
		f(V)=\bra{v}K\ket{v}=\bra{0}V^\dagger K V \ket{0}. 
	\end{align*}
	We have $\ex{f(V)}= \frac{\tr{K}}{d}=0$. Moreover $f$ is $2$-Lipschitz with respect to the Euclidean norm. Indeed, let $U$ and $V$ be two unitaries and $\ket{u}=U\ket{0}, \ket{v}=V\ket{0}$.
    We can first use the triangle inequality  and then the Cauchy-Schwarz inequality to show that:
	\begin{align*}
		|f(U)-f(V)| &= |\bra{u}K\ket{u} - \bra{v}K\ket{v}|
    	\\&\le |\bra{u-v}K\ket{u}|+ |\bra{v}K\ket{u-v}|
    	\\&\le  |\spr{u-v}{u-v} |^{1/2}|\bra{u}K^2\ket{u}|^{1/2}+ |\spr{u-v}{u-v} |^{1/2}|\bra{v}K^2\ket{v}|^{1/2}
    	\\&\le 2|\bra{0}(U-V)^{\dagger}(U-V)\ket{0}|^{1/2}\|K^2\|^{1/2}_{\infty}
    	\\&\le 2\|U-V\|_2 \|K\|_{\infty}.
	\end{align*}
	Hence $f$ concentrates around its mean~\cite[Corollary 17]{meckes2013spectral}, for all $s>0$:
	\begin{align*}
		\pr{ f(V)-\ex{f(V)}\ge s } = \pr{ f(V)\ge s }\le \exp\left(-\frac{ds^2}{48\|K\|_{\infty}^2}\right)
	\end{align*}
	Therefore using $\bra{v}(H-K)\ket{v} \le \|H-K\|_{\infty} $ and $\|H\|_{\infty}\leq \eta$:
	\begin{align*}
		\pr{\|H-K\|_{\infty} \le \frac{\eta}{2}} &\le 	\pr{\bra{v}(H-K)\ket{v} \le \frac{\eta}{2}}\mathbf{1}\left(\left\{\|K\|_{\infty}\le \frac{3\eta}{2}\right\}\right)
		\\&= \pr{\left(1-\frac{1}{d}\right)\eta- f(V) \le \frac{\eta}{2}}\mathbf{1}\left(\left\{\|K\|_{\infty}\le \frac{3\eta}{2}\right\}\right)
		\\&= \pr{f(V)\ge \left(\frac{1}{2} - \frac{1}{d}\right)\eta}\mathbf{1}\left(\left\{\|K\|_{\infty}\le \frac{3\eta}{2}\right\}\right)
		\\&\le \exp\left(-\frac{d\eta^2}{768\|K\|_{\infty}^2}\right)\mathbf{1}\left(\left\{\|K\|_{\infty}\le \frac{3\eta}{2}\right\}\right)
		\\&\le \exp\left(-\frac{d}{1728}\right).
	\end{align*}
 Here, we used $1/d \leq 1/4$ in the second to last inequality.
\end{proof}

\begin{lemma}\label{lem:eta-far-from-k-local}
	Let $n \geq 2$, $k\le  \cO\left(\frac{n}{\log(n)}\right)$ and $H=\eta  (\proj{v} -  \dI/d)$ for $\eta>0$ and $\ket{v}=V\ket{0}, V\sim \Haar(d)$. Then, $H$ is $(\eta/4)$-far (in the operator norm) from all $k$-local traceless Hamiltonians of operator norm at most $1$ with high probability. 
\end{lemma}

\begin{remark}\label{remark:from-infty-to-p-norm}
	Since for any $p\geq 1$ and any Hermitian operator $X$ we have $\|X\|_p \ge \|X\|_\infty$, the (random) Hamiltonian  $H$ is $(\eta/4)$-far in the $p$-norm from all $k$-local traceless Hamiltonians of operator norm at most $1$. Therefore, we can focus on the case $p = \infty$.
\end{remark}

\begin{proof}[Proof of \Cref{lem:eta-far-from-k-local}]
We shall take $\cH_\eps$ an $\eps$-net in the space of $k$-local Hamiltonians for the $\infty$-norm at the level of the coefficients. For each $K=\sum_{P\in \dP: |P|\le k} \alpha_P P \in \Pi_{S_{k\mathrm{-loc}}}$ such that $\|K\|_{\infty}\le 1$ we have $\|\alpha\|_\infty=\max_P |\alpha_P| = \max_P \frac{\tr(PK)}{d} \le \max_P \frac{\tr(|P|)\|K\|_\infty}{d} \le1$. 
For $k$-local Hamiltonians, many entries of $\alpha$ will be zero, namely all the ones that correspond to $P$ such that $|P|>k$. In this case, we can construct an $\eps$-net on the non-zero entries of cardinality at most $|\cH_\eps|\le (1/\eps)^{N_k}$ where $N_k= \sum_{s=0}^k\binom{n}{s} 3^s $. Now if $K=\sum_{P\in \dP: |P|\le k} \alpha_P P\in \Pi_{S_{k\mathrm{-loc}}}$, then we can find a similar element $\widetilde{K}=\sum_{P\in \dP: |P|\le k} \tilde{\alpha}_P P\in \cL_k$ in our $\eps$-net such that for all $|P|\le k$ we have $|\alpha_P -\Tilde{\alpha}_P|\le \eps$. Thus, 
\begin{align*}
\|K-\widetilde{K}\|_\infty	\le \|K-\widetilde{K}\|_2\le \sqrt{d\eps^2 N_k} \le \frac{\eta}{4}
\end{align*}
for $\eps= \eta/(4\sqrt{dN_k})$.  
Now we apply the union bound and \Cref{lem:concentration-operator-norm} to bound
\begin{align*}
	\mathbb{P}_{V\sim \Haar(d)}\left(\exists K\in \Pi_{S_{k\mathrm{-loc}}}:  \|H-K\|_\infty \le \frac{\eta}{4}\right)&\le \mathbb{P}_{V\sim \Haar(d)}\left( \exists K\in \cH_\eps:\|H-K\|_\infty\le \frac{\eta}{2}\right)
	\\&\le \sum_{K\in \cH_\eps}\mathbb{P}_{V\sim \Haar(d)}\left(\|H-K\|_\infty\le \frac{\eta}{2}\right)
	\\&\le |\cH_\eps|\exp\left(-\frac{d}{1728} \right)
	\\&= \left(\frac{1}{\eps}\right)^{N_k} \exp\left(-\frac{d}{1728} \right)
	\\&= \left(\frac{4\sqrt{dN_k}}{\eta}\right)^{N_k} \exp\left(-\frac{d}{1728} \right).
\end{align*}
Moreover, as $k\leq n/2$, we have the following simple upper bound on $N_k$:
\begin{align*}
	N_k&= \sum_{s=0}^k\binom{n}{s} 3^s\le\min \left( (k+1)\binom{n}{k}3^k, 4^n \right)\le\min \left( (k+1)(3n)^k, 4^n \right).
\end{align*}
Hence, 
\begin{align*}
	\left(\frac{4\sqrt{dN_k}}{\eta}\right)^{N_k} \exp\left(-\frac{d}{1728} \right)\le \exp\left( (k+1)(3n)^k\log\left(\frac{2^{3n}}{\eta} \right) -\frac{2^n}{1728} \right)\le \exp\left(-\Omega(d) \right)
\end{align*}
if $n^2(3n)^k \le  c\cdot2^n $ for a small constant $c$. This is valid, for instance, as long as $k\le \cO\left(\frac{n}{\log(n)}\right)$.
 Therefore, with high probability, $H= \eta (\proj{v} -  \dI/d)$ is not $(\eta/4)$-close to any $k$-local Hamiltonian if $k \le \cO\left(\frac{n}{\log(n)}\right)$.
\end{proof}
Now we proceed to prove \Cref{thm:lower-Schatten-indep}. Our proof strategy is inspired by \cite{fawzi2023quantum-channel-certification}.

\begin{proof}[Proof of \Cref{thm:lower-Schatten-indep}]
Let  $H= \eta (\proj{v} -  \dI/d)$,  $U_{v,t}= \mathrm{e}^{-\mathrm{i}\eta t\cdot \proj{v} }$ and let $\cU_{v,t}$ be the unitary channel:
\begin{align*}
      \cU_{v,t}(\rho) 
    = \mathrm{e}^{-\mathrm{i}\eta t\cdot (\proj{v}-\dI/d) }\rho \mathrm{e}^{\mathrm{i}\eta t\cdot (\proj{v}-\dI/d)}= \mathrm{e}^{-\mathrm{i}\eta t\cdot \proj{v} }\rho \mathrm{e}^{\mathrm{i}\eta t\cdot \proj{v}}=U_{v,t}\rho U_{v,t}^\dagger. 
\end{align*} 
A $1/3$-correct algorithm should distinguish between the identity channel and $\cU_{v,t}$ with at least a probability $2/3$ of success. In the incoherent setting, the tester can only choose an input $\rho_k$ at each step $k$, the time evolution $t_k$ and perform a measurement using the POVM $\cM_k=\{\lambda^{(k)}_i\proj{\phi^{(k)}_i}\}_{i\in \cI_k}$ on the output quantum state $\cU_{t_k}(\rho_k)$. These choices can depend on the previous observations, that is, the algorithm can be adaptive. Let $I_{\le N}=(I_1,\dots,I_N)$ be the observations of this algorithm where $N$ is a sufficient number of independent experiments to decide correctly with a  probability at least $2/3$. 
\\
\edit{Let $P$ (resp $Q$) be the distribution of $(I_1,\dots,I_N)$ under  the null hypothesis $\cH_0$ ($H=0$)  and the alternate hypothesis  $\cH_1$ ($H=\eta (\proj{v}- \dI/d)$).}{Let $P$ be the distribution of $(I_1,\dots,I_N)$ under  the null hypothesis $\cH_0$ ($H=0$); let $Q$ be the distribution of $(I_1,\dots,I_N)$ under  the alternate hypothesis  $\cH_1$ ($H=\eta (\proj{v}- \dI/d)$).}
\ins{More concretely: }The distribution of $(I_1,\dots,I_N)$ under $\cH_0$ is: 
 \begin{align*}
     P\edit{:=}{=}
     \Bigg\{\prod_{k=1}^N \lambda^{(k)}_{i_k}\bra{\phi^{(k)}_{i_k}}\rho_k\ket{\phi^{(k)}_{i_k}}   \Bigg\}_{i_1,\dots,i_N}.
 \end{align*}
 Moreover, the distribution of $(I_1,\dots,I_N)$ under $\cH_1$ and conditioned on the choice of unitary $V$ is:
  \begin{align} \label{eq:Le-Cam}
    Q_V
    \edit{:=}{=} \Bigg\{\prod_{k=1}^N \lambda^{(k)}_{i_k}\bra{\phi^{(k)}_{i_k}} U_{v,t_k}\rho_k U_{v,t_k}^\dagger\ket{\phi^{(k)}_{i_k}}   \Bigg\}_{i_1,\dots,i_N}.
 \end{align}
 Let $\cE$ be the event that $\eta(V\proj{0}V^\dagger-\dI/d) $ is not $(\eta/4)$-close to any $k$-local Hamiltonian. By \Cref{lem:eta-far-from-k-local}, we have $\pr{\cE}\ge 1-\exp(-\Omega(d))$. Under this event, every correct algorithm should be able to distinguish between $P$ and $Q_V$ so we can apply Le Cam's method \cite{lecam1973convergence}.
 
For technical reasons, we will use instead of $Q_V$ the following parameterized distribution for $\alpha\in [0,1]$:
  \begin{align*}
    Q_V^\alpha:= \Bigg\{\prod_{k=1}^N \lambda^{(k)}_{i_k}\bra{\phi^{(k)}_{i_k}}(\alpha\rho_k+(1-\alpha) U_{v,t_k}\rho_k U_{v,t_k}^\dagger)\ket{\phi^{(k)}_{i_k}}   \Bigg\}_{i_1,\dots,i_N}.
 \end{align*}
 So we need to generalize Le Cam's inequality for small $\alpha$:
\begin{lemma}[Generalized Le Cam]\label{lemma:gen-le-cam-main-text}
      Let $n\geq \Omega(1)$.
      Let $k\le \cO\left(\frac{n}{\log(n)}\right)$.
      \edit{We have for $\alpha\le \frac{1}{10N}:$}{For any $\alpha\le \frac{1}{10N}$, if there is an incoherent algorithm that correctly distinguishes between $P$ and $Q_V$ with success probability at least $2/3$, then}
     \begin{align*}
        \exs{V\sim \Haar(d)}{  \TV(P, Q_V^\alpha)}\ge \frac{1}{18}.
     \end{align*}
 \end{lemma}
 \begin{proof}
    The proof can be found in \edit{Lemma~\ref{lem:gen-le-cam}}{Appendix \ref{appendix:deferred-proofs}}.
 \end{proof}

 Hence, as long as $\alpha\le \frac{1}{10N}$ we can use $Q_V^\alpha$ instead of $Q_V$ and have a similar Le Cam separation (albeit with a worse constant). From now on we will set $\alpha= \frac{1}{10N}$.
The next step is to use Pinsker's inequality to move from the $\TV$ distance to the $\KL$ divergence, which is more suitable for studying the adaptive incoherent setting. By Jensen's and Pinsker's inequalities we have:
\begin{align*}
  \frac{2}{18^2}  \le  2\exs{V\sim \Haar(d)}{ \TV(P,  Q_V^\alpha)}^2\le 2\exs{V\sim \Haar(d)}{ \TV(P,  Q_V^\alpha)^2}\le\exs{V\sim \Haar(d)}{ \KL(P\|Q_V^\alpha)}.
\end{align*}
 The $\KL$ divergence can be expressed as follows:
 \begin{align*}
     \KL(P\|Q_V^\alpha)&=\mathbb{E}_{i\sim P} (-\log)\left( \frac{Q_V^\alpha(i)}{P(i)}\right)
     \\&=\mathbb{E}_{i\sim P} (-\log)\left(\prod_{k=1}^N \left(\frac{\bra{\phi^{(k)}_{i_k}}(\alpha \rho_k +(1-\alpha) U_{v,t_k}\rho_k U_{v, t_k}^\dagger)\ket{\phi^{(k)}_{i_k}}}{\bra{\phi^{(k)}_{i_k}}\rho_k\ket{\phi^{(k)}_{i_k}}}\right) \right)
     \\&=\mathbb{E}_{i\sim P} \sum_{k=1}^N(-\log) \left(\alpha+(1-\alpha)\frac{\bra{\phi^{(k)}_{i_k}}U_{v,t_k}\rho_k U_{v, t_k}^\dagger\ket{\phi^{(k)}_{i_k}}}{\bra{\phi^{(k)}_{i_k}}\rho_k\ket{\phi^{(k)}_{i_k}}}\right).
 \end{align*}
 
 Now, we will use a \edit{cut}{case distinction based on a probability threshold}
 \del{on $i\sim P$} as in  \cite{fawzi2023quantum-channel-certification}. Note that here, since we do not have the freedom to choose a non-unitary channel, we choose $\alpha=\frac{1}{10N}$ rather than $\alpha=\frac{1}{2}$ as in \cite{fawzi2023quantum-channel-certification}. This comes  with the cost of an additional logarithmic factor in the lower bound. 
 
For $k\in [N]$ and $i_{\le k}=(i_1, \dots, i_k)$, define the event
\begin{equation*}
    \cG(k, i_{\le k})= \left\{ \bra{\phi^{(k)}_{i_k}}\rho_k\ket{\phi^{(k)}_{i_k}} \le \frac{(1-\cos(\eta t_k))}{d^2} \right\} \,.
\end{equation*} 
We can distinguish whether the event $\cG$ is satisfied or not:
\begin{align*}
    &\mathbb{E}_{V\sim \Haar(d)}\KL(P\| Q^\alpha_V)
    = \sum_{k=1}^N \mathbb{E}_{V\sim \Haar(d)} \mathbb{E}_{i_{\le k}} (\mathbf{1}(\cG(k, i_{\le k}))+ \\&\mathbf{1}(\cG^c(k, i_{\le k})))(-\log)\left(\frac{\bra{\phi^{(k)}_{i_k}}(\alpha \rho_k +(1-\alpha) U_{v,t_k}\rho_k U_{v, t_k}^\dagger)\ket{\phi^{(k)}_{i_k}}}{\bra{\phi^{(k)}_{i_k}}\rho_k\ket{\phi^{(k)}_{i_k}}}\right) .
\end{align*}
Let us first analyze the setting when the event $\cG$ holds. Fix $k\in [N]$, observe that we have the inequality:
\begin{align*}
   & (-\log)\left(\frac{\bra{\phi^{(k)}_{i_k}}(\alpha \rho_k +(1-\alpha) U_{v,t_k}\rho_k U_{v, t_k}^\dagger)\ket{\phi^{(k)}_{i_k}}}{\bra{\phi^{(k)}_{i_k}}\rho_k\ket{\phi^{(k)}_{i_k}}}\right) 
    \\ &= (-\log) \left(\alpha+(1-\alpha)\frac{\bra{\phi^{(k)}_{i_k}}U_{v,t_k}\rho_k U_{v, t_k}^\dagger\ket{\phi^{(k)}_{i_k}}}{\bra{\phi^{(k)}_{i_k}}\rho_k\ket{\phi^{(k)}_{i_k}}}\right)\notag
     \le \log\left( \frac{1}{\alpha}\right)
\end{align*}
The last inequality follows from the monotonicity of the logarithm. Then we can control the expectation under the event $\cG$ as follows:
\begin{align}\label{eq: E1}
    &\mathbb{E}_{V\sim \Haar(d)}\mathbb{E}_{i_{\le k}} \mathbf{1}(\cG(k, i_{\le k}))(-\log) \left(\frac{\bra{\phi^{(k)}_{i_k}}(\alpha \rho_k +(1-\alpha) U_{v,t_k}\rho_k U_{v, t_k}^\dagger)\ket{\phi^{(k)}_{i_k}}}{\bra{\phi^{(k)}_{i_k}}\rho_k\ket{\phi^{(k)}_{i_k}}}\right) \notag
    \\&\le  \mathbb{E}_{V\sim \Haar(d)}\mathbb{E}_{i_{\le k}}  \mathbf{1}(\cG(k, i_{\le k}))  \log\left( \frac{1}{\alpha}\right)
    \notag
    \\&= \mathbb{E}_{V\sim \Haar(d)}\mathbb{E}_{i_{\le k-1}}  \sum_{i_k}\lambda^{(k)}_{i_k} \bra{\phi^{(k)}_{i_k}}\rho_k\ket{\phi^{(k)}_{i_k}}  \mathbf{1}(\cG(k, i_{\le k}))  \log\left( \frac{1}{\alpha}\right)  \notag
    \\&\le \mathbb{E}_{V\sim \Haar(d)}\mathbb{E}_{i_{\le k-1}} \sum_{i_k}\lambda^{(k)}_{i_k} \left( \frac{(1-\cos(\eta t_k))}{d^2}\right) \mathbf{1}(\cG(t, i_{\le k}))\log\left( \frac{1}{\alpha}\right)
    \notag
    \\&\le \mathbb{E}_{V\sim \Haar(d)}\mathbb{E}_{i_{\le k-1}}  \sum_{i_k}\lambda^{(k)}_{i_k} \left( \frac{(1-\cos(\eta t_k))}{d^2}\right) \log\left( \frac{1}{\alpha}\right)\notag
    \\&=      \frac{\log(10N)}{d} \; \mathbb{E}_{i_{\le k-1}}(1-\cos(\eta t_k))
\end{align}
where we used in the second inequality the fact that under  $\cG$ we have   $\bra{\phi^{(k)}_{i_k}}\rho_k\ket{\phi^{(k)}_{i_k}} \le \frac{(1-\cos(\eta t_k))}{d^2} $  and in the last equality the fact    $\sum_{i_k} \lambda^{(k)}_{i_k} = d$ which is an implication of the fact that  $\cM_k=\left\{\lambda^{(k)}_{i}\proj{\phi^{(k)}_{i}}\right\}_{i\in \cI_k}$ is a POVM.
Next, under $\cG^c(t, i_{\le k})$, we will use instead the following inequality valid for all $x\in \big[\frac{1}{10N},+\infty\big)$: 
\begin{align}\label{ineq-log-N}
    (-\log)(x)\le -(x-1)+2\log(10N)(x-1)^2.
\end{align}
A simple proof can be found in \Cref{lem:elem-calculations}.

We apply the inequality~\eqref{ineq-log-N} for $$x=\frac{\bra{\phi^{(k)}_{i_k}}(\alpha \rho_k +(1-\alpha) U_{v,t_k}\rho_k U_{v, t_k}^\dagger)\ket{\phi^{(k)}_{i_k}}}{\bra{\phi^{(k)}_{i_k}}\rho_k\ket{\phi^{(k)}_{i_k}}}=\alpha+(1-\alpha) \frac{\bra{\phi^{(k)}_{i_k}}U_v\rho_k U_v^\dagger\ket{\phi^{(k)}_{i_k}}}{\bra{\phi^{(k)}_{i_k}}\rho_k\ket{\phi^{(k)}_{i_k}}}\ge\frac{1}{10N}. $$ 
Let $M_{v, t_k}= \dI- U_{v, t_k}= -(\mathrm{e}^{-\mathrm{i}\eta t_k}-1)\proj{v}$  and $S_{v, t_k}= \dI-\frac{1}{2}M_{v, t_k}$. The first term of the upper bound of Inequality~\eqref{ineq-log-N} is
\begin{align}\label{eq: E2}
    -(x-1)&= 1-\frac{\bra{\phi^{(k)}_{i_k}}(\alpha \rho_k +(1-\alpha) U_{v,t_k}\rho_k U_{v, t_k}^\dagger)\ket{\phi^{(k)}_{i_k}}}{\bra{\phi^{(k)}_{i_k}}\rho_k\ket{\phi^{(k)}_{i_k}}} \notag
    \\&= (1-\alpha)\frac{\bra{\phi^{(k)}_{i_k}}M_{v, t_k}\rho_kS_{v, t_k}^\dagger \ket{\phi^{(k)}_{i_k}}}{\bra{\phi^{(k)}_{i_k}}\rho_k\ket{\phi^{(k)}_{i_k}}}   +(1-\alpha)\frac{\bra{\phi^{(k)}_{i_k}}S_{v, t_k}\rho_kM_{v, t_k}^\dagger \ket{\phi^{(k)}_{i_k}}}{\bra{\phi^{(k)}_{i_k}}\rho_k\ket{\phi^{(k)}_{i_k}}} \notag 
    \\&=2(1-\alpha)\Re\frac{\bra{\phi^{(k)}_{i_k}}M_{v, t_k}\rho_kS_{v, t_k}^\dagger \ket{\phi^{(k)}_{i_k}}}{\bra{\phi^{(k)}_{i_k}}\rho_k\ket{\phi^{(k)}_{i_k}}}\, , 
\end{align}
and by using first the inequality $(x+y)^2\le 2(x^2+y^2)$ and then the Cauchy Schwarz inequality applied for the vectors $\sqrt{\rho_k}\ket{\phi^{(k)}_{i_k} }$ and $\sqrt{\rho_k}M_{v, t_k}^{\dagger}\ket{\phi^{(k)}_{i_k} }$, we can upper bound the second term of Inequality~\eqref{ineq-log-N} as follows:
\begin{align}\label{eq: E3}
    &(x-1)^2\\
    &= \left(\frac{\bra{\phi^{(k)}_{i_k}}(\alpha \rho_k +(1-\alpha) U_{v,t_k}\rho_k U_{v, t_k}^\dagger)\ket{\phi^{(k)}_{i_k}}}{\bra{\phi^{(k)}_{i_k}}\rho_k\ket{\phi^{(k)}_{i_k}}} -1\right)^2\notag
    \\&= (1-\alpha)^2 \left(  \frac{2\Re\bra{\phi^{(k)}_{i_k}}M_{v, t_k} \rho_k\ket{\phi^{(k)}_{i_k}}}{\bra{\phi^{(k)}_{i_k}}\rho_k\ket{\phi^{(k)}_{i_k}}} - \frac{\bra{\phi^{(k)}_{i_k}}M_{v, t_k}\rho_kM_{v, t_k}^\dagger\ket{\phi^{(k)}_{i_k}}}{\bra{\phi^{(k)}_{i_k}}\rho_k\ket{\phi^{(k)}_{i_k}}} \right)^2 \notag
     \\&\le 8(1-\alpha)^2 \left(  \frac{\left\lvert\bra{\phi^{(k)}_{i_k}}M_{v, t_k} \rho_k\ket{\phi^{(k)}_{i_k}}\right\rvert}{\bra{\phi^{(k)}_{i_k}}\rho_k\ket{\phi^{(k)}_{i_k}}}\right)^2 + 2(1-\alpha)^2 \left(\frac{\bra{\phi^{(k)}_{i_k}}M_{v, t_k}\rho_kM_{v, t_k}^\dagger\ket{\phi^{(k)}_{i_k}}}{\bra{\phi^{(k)}_{i_k}}\rho_k\ket{\phi^{(k)}_{i_k}}}     \right)^2 \notag
      \\&\le 8(1-\alpha)^2 \left(   \frac{\bra{\phi^{(k)}_{i_k}}M_{v, t_k} \rho_kM_{v, t_k}^\dagger\ket{\phi^{(k)}_{i_k}}}{\bra{\phi^{(k)}_{i_k}}\rho_k\ket{\phi^{(k)}_{i_k}}}\right) + 2(1-\alpha)^2 \left(\frac{\bra{\phi^{(k)}_{i_k}}M_{v, t_k}\rho_kM_{v, t_k}^\dagger\ket{\phi^{(k)}_{i_k}}}{\bra{\phi^{(k)}_{i_k}}\rho_k\ket{\phi^{(k)}_{i_k}}}     \right)^2 \ . 
\end{align}
 Let us compute the expectation of~\Cref{eq: E2}. Let $M, S$ such that $M_{v, t_k}=VMV^\dagger $ and  $S_{v, t_k}=VSV^\dagger $. 
 Concretely
	$$ 
	M= \begin{pmatrix} 
     1- \mathrm{e}^{-\mathrm{i}\eta t_k}& \rvline &  \bigzero  \\
	\hline
	&\rvline & \\
	\bigzero &\rvline &    \bigzero_{d-1}  
	\\
	& \rvline&
	\end{pmatrix}
	~~\text{ and }~~ 
	S= \begin{pmatrix}
    \frac{1}{2} +\frac{\mathrm{e}^{-\mathrm{i}\eta t_k}}{2} & \rvline & \bigzero \\
	\hline
	&\rvline & \\
	\bigzero &\rvline &    \bigeye_{d-1}  \\
	&\rvline & 
	\end{pmatrix}
	.$$
 Note that $\tr(M)= (1-\mathrm{e}^{-\mathrm{i}\eta t_k})$, $\tr(S)=d-\frac{1}{2}+\frac{\mathrm{e}^{-\mathrm{i}\eta t_k}}{2}$, $\tr(MS^\dagger)=\mathrm{i}\sin(\eta t_k)$,  $\tr(MS)= \frac{1-\mathrm{e}^{-2\mathrm{i}\eta t_k}}{2}$ and $MM^\dagger = M+M^\dagger= M^\dagger M $.
We have 
\begin{align*}
    \Re(\tr(M)\tr(S))&= \Re\left((1-\mathrm{e}^{-\mathrm{i}\eta t_k})\big(d-\frac{1}{2}+\frac{\mathrm{e}^{-\mathrm{i}\eta t_k}}{2}\big)\right)
    \\&=(1-\cos(\eta t_k))\big(d-\frac{1}{2}+\frac{\cos(\eta t_k)}{2}\big) + \frac{1}{2}\sin^2(\eta t_k)
    \\&\le (1-\cos(\eta t_k))d + \frac{1}{2}\sin^2(\eta t_k)\, .
\end{align*}
Moreover,
\begin{equation*}
    \Re(\tr(M)\tr(S)) \geq -1/2 \sin^2(\eta t_k) \, .
\end{equation*}
Hence, 
\begin{equation*}
    |\Re(\tr(M)\tr(S))| \leq (1-\cos(\eta t_k))d + \frac{1}{2}\sin^2(\eta t_k)\, .
\end{equation*}
So we have by  Weingarten calculus \cite{collins2006integration} (see \Cref{sec:weingarten facts} for the most important facts used here):
{\small
\begin{align}
    &\left|\mathbb{E}_{V\sim \Haar(d)} \left( \Re\frac{\bra{\phi^{(k)}_{i_k}}M_{v, t_k}\rho_kS_{v, t_k}^\dagger \ket{\phi^{(k)}_{i_k}}}{\bra{\phi^{(k)}_{i_k}}\rho_k\notag\ket{\phi^{(k)}_{i_k}}}    \right)\right|
    \\&=\left|\frac{\Re \mathbb{E}_{V\sim \Haar(d)} \left( \bra{\phi^{(k)}_{i_k}}VMV^\dagger \rho_kV S^\dagger V^\dagger \ket{\phi^{(k)}_{i_k}}\right)}{\bra{\phi^{(k)}_{i_k}}\rho_k\ket{\phi^{(k)}_{i_k}}}\right|\notag
    \\&=\left|\frac{\Re \sum_{\alpha, \beta \in \fS_2} \W(\alpha\beta, d) \tr_\alpha(M,S^\dagger)\tr_{\beta (12)}(\rho_k, \proj{\phi^{(k)}_{i_k}})}{\bra{\phi^{(k)}_{i_k}}\rho_k\ket{\phi^{(k)}_{i_k}}}\right|\notag
    \\&=\left|\frac{\Re\left(d\tr(MS^\dagger)+ d\tr(M)\tr(S)\bra{\phi^{(k)}_{i_k}}\rho_k\ket{\phi^{(k)}_{i_k}}- \tr(MS^\dagger )\bra{\phi^{(k)}_{i_k}}\rho_k\ket{\phi^{(k)}_{i_k}}- \tr(M)\tr(S^\dagger) \right)}{d(d^2-1)\bra{\phi^{(k)}_{i_k}}\rho_k\ket{\phi^{(k)}_{i_k}}} \right|\label{eq:Haar-S2}
   \\&= \left|\frac{1}{\bra{\phi^{(k)}_{i_k}}\rho_k\ket{\phi^{(k)}_{i_k}}}\left(\frac{  \Re(\tr(M)\tr(S)) (d\bra{\phi^{(k)}_{i_k}}\rho_k\ket{\phi^{(k)}_{i_k}}-1) }{d(d^2-1)}\right) \right|\notag
   \\&\le \frac{1}{d(d^2-1)\bra{\phi^{(k)}_{i_k}}\rho_k\ket{\phi^{(k)}_{i_k}}}\left(  \left[(1-\cos(\eta t_k))d + \frac{1}{2}\sin^2(\eta t_k)\right] (d\bra{\phi^{(k)}_{i_k}}\rho_k\ket{\phi^{(k)}_{i_k}}+1)\right) \notag
    \\&\le\del{\cO\left(\right.} \ins{2}\frac{(1-\cos(\eta t_k)) }{d}+  \ins{2}\frac{(1-\cos(\eta t_k))}{d^2 \bra{\phi^{(k)}_{i_k}}\rho_k\ket{\phi^{(k)}_{i_k}}}+ \frac{\sin^2(\eta t_k)}{d^2\bra{\phi^{(k)}_{i_k}}\rho_k\ket{\phi^{(k)}_{i_k}}}\del{\left.\right)}. \notag 
\end{align}
}

Recall the notation 
\begin{equation*}
    \mathbb{E}_{i\le t}(X(i_1, \dots, i_t)) = \sum_{i_1, \dots, i_t} \prod_{k=1}^t \lambda^{(k)}_{i_k}\bra{\phi^{(k)}_{i_k}}\rho_k\ket{\phi^{(k)}_{i_k}} X(i_1, \dots, i_t) \,.
\end{equation*}
If we take the expectation $\mathbb{E}_{i\le t}$ under the event $\cG^c(t, i_{\le k})$, we obtain
\begin{align}\label{eq: E2 UB}
     & \mathbb{E}_{V\sim \Haar(d)} \mathbb{E}_{i_{\le k}}  \mathbf{1}(\cG^c(k, i_k)) \left( \Re\frac{\bra{\phi^{(k)}_{i_k}}M_{v, t_k}\rho_kS_{v, t_k}^\dagger \ket{\phi^{(k)}_{i_k}}}{\bra{\phi^{(k)}_{i_k}}\rho_k\notag\ket{\phi^{(k)}_{i_k}}}    \right)\notag 
     \\&\le \mathbb{E}_{i_{\le k}}  \mathbf{1}(\cG^c(k, i_k)) \left|\mathbb{E}_{V\sim \Haar(d)} \left( \Re\frac{\bra{\phi^{(k)}_{i_k}}M_{v, t_k}\rho_kS_{v, t_k}^\dagger \ket{\phi^{(k)}_{i_k}}}{\bra{\phi^{(k)}_{i_k}}\rho_k\notag\ket{\phi^{(k)}_{i_k}}}    \right)\right|\notag 
     \\&\le \mathbb{E}_{i_{\le k}}  \mathbf{1}(\cG^c(k, i_k))\del{\cO}\left( \ins{2}\frac{(1-\cos(\eta t_k)) }{d}+  \frac{\ins{2}(1-\cos(\eta t_k))+ \sin^2(\eta t_k)}{d^2 \bra{\phi^{(k)}_{i_k}}\rho_k\ket{\phi^{(k)}_{i_k}} }\right)
     \notag
     \\&\le  \mathbb{E}_{i_{\le k}}  \del{\cO}\left( \ins{2}\frac{(1-\cos(\eta t_k)) }{d}+  \frac{\ins{2}(1-\cos(\eta t_k))+ \sin^2(\eta t_k)}{d^2 \bra{\phi^{(k)}_{i_k}}\rho_k\ket{\phi^{(k)}_{i_k}} }\right)
     \notag
     \\&=  \mathbb{E}_{i_{\le k-1}}  \del{\cO}\left(\ins{2}\frac{(1-\cos(\eta t_k))}{d}\right)\notag\\& \quad +  \mathbb{E}_{i_{\le k-1}} \sum_{i_k}  \lambda^{(k)}_{i_k} \bra{\phi^{(k)}_{i_k}}\rho_k\ket{\phi^{(k)}_{i_k}}\cdot  \del{\cO}\left(\frac{\ins{2}(1-\cos(\eta t_k))+ \sin^2(\eta t_k)}{d^2 \bra{\phi^{(k)}_{i_k}}\rho_k\ket{\phi^{(k)}_{i_k}} }\right)  \notag 
     \\&=  \mathbb{E}_{i_{\le k-1}} \del{\cO}\left(\frac{\ins{4}(1-\cos(\eta t_k))+\sin^2(\eta t_k)}{d}\right)\, ,
\end{align}
where we use  $\sum_{i_k} \lambda^{(k)}_{i_k} = d$. We move to the expectation $\mathbb{E}_{i_{\le k}}  $ of the first term of~\Cref{eq: E3}, it is non negative so we can safely remove the condition $\mathbf{1}(\cG^c(k, i_k)) $:
 \begin{align}\label{eq: E3-1 UB}
 & \mathbb{E}_V\mathbb{E}_{i_{\le k}}  \mathbf{1}(\cG^c(k, i_k)) \frac{\bra{\phi^{(k)}_{i_k}}M_{v, t_k}\rho_kM_{v, t_k}^\dagger\ket{\phi^{(k)}_{i_k}}}{\bra{\phi^{(k)}_{i_k}}\rho_k\ket{\phi^{(k)}_{i_k}}} \le \mathbb{E}_V\mathbb{E}_{i_{\le k}}  \frac{\bra{\phi^{(k)}_{i_k}}M_{v, t_k}\rho_kM_{v, t_k}^\dagger\ket{\phi^{(k)}_{i_k}}}{\bra{\phi^{(k)}_{i_k}}\rho_k\ket{\phi^{(k)}_{i_k}}}\notag
  \\&= \mathbb{E}_V\mathbb{E}_{i_{\le k-1}} \sum_{i_k}\lambda^{(k)}_{i_k}  \bra{\phi^{(k)}_{i_k}}M_{v, t_k}\rho_kM_{v, t_k}^\dagger\ket{\phi^{(k)}_{i_k}}\notag
  \\&= \mathbb{E}_{i_{\le k-1}} \mathbb{E}_V \tr(M_{v, t_k}\rho_kM_{v, t_k}^\dagger)
   = \mathbb{E}_{i_{\le k-1}} \mathbb{E}_V \tr((M_{v, t_k}+M_{v, t_k}^\dagger)\rho_k) \notag
   \\&=\mathbb{E}_{i_{\le k-1}}  \frac{2(1-\cos(\eta t_k))}{d}
 \end{align}
because $\mathbb{E}_{V} M_{v, t_k} = \frac{ (1-\mathrm{e}^{-\mathrm{i}\eta t_k})}{d}\dI$.
\\Concerning the expectation of the second term of~\Cref{eq: E3}, we apply again the Weingarten calculus (\Cref{lem:Wg}) to have:
\begin{align*}
   & \mathbb{E}_{V\sim \Haar(d)}  \left( \frac{\bra{\phi^{(k)}_{i_k}}M_{v, t_k}\rho_kM_{v, t_k}^\dagger \ket{\phi^{(k)}_{i_k}}}{\bra{\phi^{(k)}_{i_k}}\rho_k\ket{\phi^{(k)}_{i_k}}}    \right)^2
   =   \frac{\mathbb{E}_{V\sim \Haar(d)}\bra{\phi^{(k)}_{i_k}}M_{v, t_k}\rho_kM_{v, t_k}^\dagger \ket{\phi^{(k)}_{i_k}}^2}{\bra{\phi^{(k)}_{i_k}}\rho_k\ket{\phi^{(k)}_{i_k}}^2} 
   \\&= \frac{\mathbb{E}_{V\sim \Haar(d)}\tr(\proj{\phi^{(k)}_{i_k}}VM V^\dagger\rho_k VM^\dagger V^\dagger \proj{\phi^{(k)}_{i_k}}VMV^\dagger \rho_k VM^\dagger V^\dagger)  }{\bra{\phi^{(k)}_{i_k}}\rho_k\ket{\phi^{(k)}_{i_k}}^2} 
   \\&= \frac{1}{\bra{\phi^{(k)}_{i_k}}\rho_k\ket{\phi^{(k)}_{i_k}}^2} \sum_{\alpha, \beta \in \fS_4}\W(\beta\alpha^{-1},d) \tr_{\beta^{-1}}(M, M^\dagger, M, M^\dagger)\\&\qquad\cdot \tr_{\alpha \gamma}( \rho_k, \proj{\phi^{(k)}_{i_k}}, \rho_k, \proj{\phi^{(k)}_{i_k}}).
\end{align*}
 Note that  
 \begin{align*}
     &\ptr{\alpha \gamma}{ \rho_k, \proj{\phi^{(k)}_{i_k}}, \rho_k, \proj{\phi^{(k)}_{i_k}}}\\& \quad \in \left\{1, \tr(\rho_k^2), \bra{\phi^{(k)}_{i_k}}\rho_k^2\ket{\phi^{(k)}_{i_k}},\bra{\phi^{(k)}_{i_k}}\rho_k\ket{\phi^{(k)}_{i_k}}, \bra{\phi^{(k)}_{i_k}}\rho_k\ket{\phi^{(k)}_{i_k}}^2 \right\},
 \end{align*}
 $\tr(\rho_k^2)\le 1$ and $\bra{\phi^{(k)}_{i_k}}\rho_k^2\ket{\phi^{(k)}_{i_k}}\le \bra{\phi^{(k)}_{i_k}}\rho_k\ket{\phi^{(k)}_{i_k}}\le 1$. 
 Moreover, it is clear that since \ins{ $M =  (1- \mathrm{e}^{-\mathrm{i}\eta t_k}) \proj{0}$ }   \del{$\tr[M]\tr[M^\dagger] = \tr[M M^\dagger]$}, we have \[|\tr_{\beta}(M, M^\dagger, M, M^\dagger)| = \ins{|1- \mathrm{e}^{-\mathrm{i}\eta t_k}|^4\tr_{\beta}(\proj{0},\proj{0},\proj{0},\proj{0})}= \del{\cO} \ins{4}(1-\cos(\eta t_k))^2.\] 
\del{ In the case $\beta$ is a $4$ cycle we have $\tr(MM^\dagger MM^\dagger )= \tr(MMM^\dagger M^\dagger )= \tr((M+M^\dagger)^2)= 4(1-\cos(\eta t_k))^2$.} Also, we know that for all $(\alpha, \beta)\in \fS_4^2$: $|\W(\beta\alpha^{-1},d)|\le \frac{10}{d^4}$~\cite{collins2006integration} for $d \geq 4$, so
 \begin{align*}
     &\left|\W(\beta \alpha^{-1},d) \tr_{\beta^{-1}}(M, M^\dagger, M, M^\dagger)\tr_{\alpha \gamma}\left(\proj{\phi^{(k)}_{i_k}}, \rho_k, \proj{\phi^{(k)}_{i_k}}, \rho_k\right)\right|\\&\quad\le \del{\cO\left(\right.}\frac{\ins{40}(1-\cos(\eta t_k))^2}{d^4}\del{\left.\right)}.
     \end{align*}
 Therefore we have:
\begin{align*}
   & \mathbb{E}_{V\sim \Haar(d)}  \left( \frac{\bra{\phi^{(k)}_{i_k}}M_{v, t_k}\rho_kM_{v, t_k}^\dagger \ket{\phi^{(k)}_{i_k}}}{\bra{\phi^{(k)}_{i_k}}\rho_k\ket{\phi^{(k)}_{i_k}}}    \right)^2
   \le \del{\cO\left( \right.}\frac{\ins{40\cdot 4!^2}(1-\cos(\eta t_k))^2}{ d^4\bra{\phi^{(k)}_{i_k}}\rho_k\ket{\phi^{(k)}_{i_k}}^2}\del{\left.\right)}.
\end{align*}
 Now if we take the expectation $\mathbb{E}_{i_{\le k}}  $ under the event 
 \begin{equation*}
 \cG^c(k, i_{\le k})= \left\{ \bra{\phi^{(k)}_{i_k}}\rho_k\ket{\phi^{(k)}_{i_k}} > \frac{(1-\cos(\eta t_k))}{d^2} \right\},
 \end{equation*}
 we obtain:
\begin{align}\label{eq: E3-2 UB}
   & \mathbb{E}_{i_{\le k}}\mathbb{E}_{V\sim \Haar(d)}\mathbf{1}(\cG^c(k, i_{\le k})) \left( \frac{\bra{\phi^{(k)}_{i_k}}M_{v, t_k}\rho_kM_{v, t_k}^\dagger \ket{\phi^{(k)}_{i_k}}}{\bra{\phi^{(k)}_{i_k}}\rho_k\ket{\phi^{(k)}_{i_k}}}    \right)^2 \notag
   \\&\le \mathbb{E}_{i_{\le k}}\mathbf{1}(\cG^c(k, i_{\le k}))\del{\cO}\left( \frac{\ins{40\cdot 4!^2}(1-\cos(\eta t_k))^2}{ d^4\bra{\phi^{(k)}_{i_k}}\rho_k\ket{\phi^{(k)}_{i_k}}}\cdot \frac{1}{ \bra{\phi^{(k)}_{i_k}}\rho_k\ket{\phi^{(k)}_{i_k}}}\right)\notag
   \\&\le \mathbb{E}_{i_{\le k}}\mathbf{1}(\cG^c(k, i_{\le k}))\del{\cO}\left( \frac{\ins{40\cdot 4!^2}(1-\cos(\eta t_k))^2}{ d^4\bra{\phi^{(k)}_{i_k}}\rho_k\ket{\phi^{(k)}_{i_k}}}\cdot \frac{d^2}{(1-\cos(\eta t_k))}\right)\notag   
   \\&=\mathbb{E}_{i_{\le k-1}} \sum_{i_k}  \lambda^{(k)}_{i_k} \bra{\phi^{(k)}_{i_k}}\rho_k\ket{\phi^{(k)}_{i_k}}\mathbf{1}(\cG^c(k, i_{\le k}))\del{\cO}\left( \frac{\ins{40\cdot 4!^2}(1-\cos(\eta t_k))}{ d^2\bra{\phi^{(k)}_{i_k}}\rho_k\ket{\phi^{(k)}_{i_k}}}\right) \notag
   \\&\le \mathbb{E}_{i_{\le k-1}} \sum_{i_k}  \lambda^{(k)}_{i_k}\del{\cO}\left( \frac{\ins{40\cdot 4!^2}(1-\cos(\eta t_k))}{ d^2}\right) \notag 
   \\&=\mathbb{E}_{i_{\le k-1}}\del{\cO}\left( \frac{\ins{40\cdot 4!^2}(1-\cos(\eta t_k))}{ d}\right) \, ,
\end{align}
where we use  $\sum_{i_k} \lambda^{(k)}_{i_k} = d$. By adding up Equations \eqref{eq: E2 UB}, \eqref{eq: E3-1 UB} and \eqref{eq: E3-2 UB}, we obtain:
\begin{align*}
   & \mathbb{E}_{i_{\le k}}\mathbb{E}_{V\sim \Haar(d)}\mathbf{1}(\cG^c(k, i_{\le k}))(-\log)\left(\frac{\bra{\phi^{(k)}_{i_k}}(\alpha \rho_k +(1-\alpha) U_{v,t_k}\rho_k U_{v, t_k}^\dagger)\ket{\phi^{(k)}_{i_k}}}{\bra{\phi^{(k)}_{i_k}}\rho_k\ket{\phi^{(k)}_{i_k}}}\right) 
    \\&\le \mathbb{E}_{i_{\le k}}\mathbb{E}_{V\sim \Haar(d)} \mathbf{1}(\cG^c(k, i_{\le k}))\left(\frac{2(1-\alpha)\Re(\bra{\phi^{(k)}_{i_k}}M_{v, t_k}\rho_kS_{v, t_k}^\dagger \ket{\phi^{(k)}_{i_k}})}{\bra{\phi^{(k)}_{i_k}}\rho_k\ket{\phi^{(k)}_{i_k}}}        \right)
    \\&\quad +  \mathbb{E}_{i_{\le k}}\mathbb{E}_{V\sim \Haar(d)} \mathbf{1}(\cG^c(k, i_{\le k}))\left( \frac{4\log(10N)(1-\alpha)^2\bra{\phi^{(k)}_{i_k}}M_{v, t_k}\rho_kM_{v, t_k}^\dagger\ket{\phi^{(k)}_{i_k}}^2}{\bra{\phi^{(k)}_{i_k}}\rho_k\ket{\phi^{(k)}_{i_k}}^2}     \right)
    \\&\quad +\mathbb{E}_{i_{\le k}}\mathbb{E}_{V\sim \Haar(d)} \mathbf{1}(\cG^c(k, i_{\le k}))  16\log(10N)(1-\alpha)^2\left(  \frac{\bra{\phi^{(k)}_{i_k}}M_{v, t_k} \rho_kM_{v, t_k}^\dagger\ket{\phi^{(k)}_{i_k}}}{\bra{\phi^{(k)}_{i_k}}\rho_k\ket{\phi^{(k)}_{i_k}}}\right)^{\del{2}}\notag 
    \\&\edit{=}{\le}  \mathbb{E}_{i_{\le k-1}}\del{\cO}\left(\frac{(\ins{8+ (160\cdot 4!^2 + 32)\log(10N) })(1-\cos(\eta t_k))+\ins{2}\sin^2(\eta t_k)}{ d}\right).
\end{align*}
Therefore using this upper bound and the upper bound  in \Cref{eq: E1} we get an upper bound on the expected $\KL$ divergence:
\begin{align*}
&\mathbb{E}_{V\sim \Haar(d)}\KL(P\| Q^\alpha_V)
\\&=  \sum_{k=1}^N\mathbb{E}_{i_{\le k}}\mathbb{E}_{V\sim \Haar(d)}(\mathbf{1}(\cG(t, i_{\le k})) + \mathbf{1}(\cG^c(k, i_{\le k})))(-\log) \left( \frac{\bra{\phi^{(k)}_{i_k}}\edit{\cU_{t_k}(}{U_{v, t_k}}\rho_k\edit{)}{U_{v, t_k}^\dagger}\ket{\phi^{(k)}_{i_k}}}{\bra{\phi^{(k)}_{i_k}}\rho_k\ket{\phi^{(k)}_{i_k}}}\right)
     \\&\le\sum_{k=1}^N  \mathbb{E}_{i_{\le k-1}}\del{\cO}\left(\frac{\ins{\log(10N)}(1-\cos(\eta t_k))}{d}\right)\\&\hphantom{\le\sum_{k=1}^N}+ \mathbb{E}_{i_{\le k-1}}\del{\cO}\left(\frac{\ins{(8+ (160\cdot 4!^2 + 32)\log(10N))}(1-\cos(\eta t_k))+\ins{2}\sin^2(\eta t_k)}{d}\right)
     \\&= \sum_{k=1}^N \mathbb{E}_{i_{\le k-1}}\del{\cO}\left(\frac{\ins{(8+ (160\cdot 4!^2 + 33)\log(10N))}(1-\cos(\eta t_k)+\ins{2}\sin^2(\eta t_k))}{d}\right)
     \\&\le \sum_{k=1}^N \mathbb{E}_{i_{\le k-1}}\del{\cO}\left(\frac{\ins{(10+ (160\cdot 4!^2 + 33)\log(10N))}\min(1,\eta t_k)}{d}\right)
\end{align*}
where we used $1-\cos(x)\le \min(1,x)$ and $\sin^2(x)\le\min(1,x)$ for $x\ge 0$. 
Finally since $\mathbb{E}_{V\sim \Haar(d)}\KL(P\| Q_V^\alpha)\ge \frac{2}{18^2}$  we deduce that:
\begin{align*}
  \sum_{k=1}^N \mathbb{E}_{i_{\le k-1}}\min(1, \eta t_k) \ins{\ge \frac{2d}{18^2(10+ (160\cdot 4!^2 + 33)\log(10N))}}= \Omega\left( \frac{d}{\log(N)}\right).
\end{align*}
In particular, we have $N\log(N)\ge\Omega\left( d\right) $ which implies that:
\begin{align*}
  N=\Omega\left( \frac{d}{\log(d)}\right).
\end{align*}
   Finally, the expected total evolution time  is lower bounded as follows:
   \begin{align*}
       \ex{ \sum_{k=1}^N t_k} = \Omega\left( \frac{d}{\eta\log(N)}\right).
   \end{align*}
   We can set $\eps = \eta/4$ and use our assumption that $N\le \exp(\cO(n))/\eps^{\cO(1)}$ to get the claimed expected total evolution time lower bound.
\end{proof}

\subsection{Coherent setting} \label{sec:lower-seq-setting}
In section, we will prove another hardness result in the more general coherent setting. However, we will pay for the greater generality with slightly weaker bounds.

\begin{theorem}\label{sequential-LB-op-norm}
    Let $n\geq 2$, $k\le \cO\left(\frac{n}{\log(n)}\right)$, and $p\geq 1$.   
    The problem $\cT^{\mathrm{loc}}_{\|\cdot\|_p}(\eta)$, even under the additional promise that the unknown Hamiltonian $H$ satisfies $\tr[H]=0$ and $\norm{H}_\infty\leq 1$, requires a total evolution time of $\sum_{k=1}^N t_k=\Omega\left( \frac{2^{n/2}}{\eps}\right)$ and a total number of independent experiments $N =\Omega\left( 2^{n/2}\right)$ in the ancilla-assisted coherent setting.
\end{theorem}

In order to prove \Cref{sequential-LB-op-norm}, we will again consider the following test: 
\begin{align}
    \cH_0: \cU_t(\rho)= \id(\rho)= \mathrm{e}^{-\mathrm{i} t\cdot 0}\rho \mathrm{e}^{\mathrm{i} t\cdot 0} ~~~~~~\text{vs} ~~~~~~~~ \cH_1: \cU_t(\rho)=  \mathrm{e}^{-\mathrm{i}tH}\rho \mathrm{e}^{\mathrm{i}tH}
\end{align}
where $H=\eta (\proj{v} -  \dI/d)$  for $\eta>0$ and $\ket{v}=V\ket{0}, V\sim \Haar(d)$. The following proof is similar in spirit to \cite{Aaronson} (see also \cite{Bennett2006Jul}).

\begin{proof}[Proof of \Cref{sequential-LB-op-norm}]
    We use here the construction from \Cref{lem:eta-far-from-k-local}. Let  $H=\eta (\proj{v}-\dI/d)$ for $\eta>0$ and $\ket{v}=V\ket{0}, V\sim \Haar(d)$. $H$ is $(\eta/4)$-far (in the operator norm) from any $k$-local trace-less Hamiltonian of unit operator norm with high probability. Let $\cE$ be the event that $H$ is $(\eta/4)$-far (in the operator norm, which we can again focus on as discussed in \Cref{remark:from-infty-to-p-norm}) from any $k$-local trace-less Hamiltonian of unit operator norm. We have by \Cref{lem:eta-far-from-k-local}:
    \begin{align*}
        \prs{V\sim \Haar(d)}{\cE^c}\le \exp\left(-\Omega(d)\right).
    \end{align*}
    In the sequential setting the tester can choose times $t_1, \dots, t_N$ and any $(d\times d_{\rm{aux}})$-dimensional operations $\cN_1,\dots, \cN_{N-1}$, a $(d\times d_{\rm{aux}})$-dimensional input state $\rho$ and a $(d\times d_{\rm{aux}})$-dimensional measurement device $\cM$. Under the null hypothesis $\cH_0$, the map $\cU_t$ is always the identity so the output state (before the measurement) is:
\begin{align*}
    \rho_0^{\text{output}}(\rho) = \cN_{N-1}\circ\dots\circ\cN_2\circ\cN_1(\rho).
\end{align*}
In contrast, under the alternate hypothesis $\cH_1$, the map $\cU_t$ is close to the identity and the output state (before the measurement) is:
\begin{align*}
    \rho_1^{\text{output}}(\rho)=  [\cU_{t_N}\otimes \id]\circ\cN_{N-1}\circ[\cU_{t_{N-1}}\otimes \id]\circ\dots\circ\cN_2\circ[\cU_{t_2}\otimes \id]\circ\cN_1\circ [\cU_{t_1}\otimes \id](\rho).
\end{align*}
We will often drop the argument $\rho$ for readability if it is not explicitly needed. On the one hand, using the correctness of the algorithm and the data processing inequality applied on the $1$-norm we have:
\begin{align*}
    \left\| \rho_0^{\text{output}}- \mathbb{E}_{V\sim \Haar(d)|\cE}\left[\rho_1^{\text{output}}\right] \right\|_1 \ge 2\TV(\Ber(1/3)\| \Ber(2/3) )= \frac{2}{3}.
\end{align*}
On the other hand, we have by the triangle inequality 
\begin{align*}
    &\left\|\mathbb{E}_{V\sim \Haar(d)}\left[\rho_1^{\text{output}}\right]-  \mathbb{E}_{V\sim \Haar(d)|\cE}\left[\rho_1^{\text{output}}\right]\right\|_1
    \\&= \frac{1}{\pr{\cE}} \left\|\pr{\cE}\mathbb{E}_{V\sim \Haar(d)}\left[\rho_1^{\text{output}}\right]-  \mathbb{E}_{V\sim \Haar(d)}\left[\rho_1^{\text{output}}\mathbf{1}(\{\cE\})\right] \right\|_1
     \\&= \frac{1}{\pr{\cE}} \left\|\pr{\cE}\mathbb{E}_{V\sim \Haar(d)}\left[\rho_1^{\text{output}}\mathbf{1}(\{\cE^c\})\right]- \pr{\cE^c} \mathbb{E}_{V\sim \Haar(d)}\left[\rho_1^{\text{output}}\mathbf{1}(\{\cE\})\right] \right\|_1
     \\&\le \left\|\mathbb{E}_{V\sim \Haar(d)}\left[\rho_1^{\text{output}}\mathbf{1}(\{\cE^c\})\right]\right\|_1
     + \frac{\pr{\cE^c}}{\pr{\cE}} \left\|  \mathbb{E}_{V\sim \Haar(d)}\left[\rho_1^{\text{output}}\mathbf{1}(\{\cE\})\right] \right\|_1
     \\&\le\pr{\cE^c}
     + \frac{\pr{\cE^c}}{\pr{\cE}} \pr{\cE}
     \\&\le 2\exp\left(-\Omega(d)\right)
\end{align*}
hence for $d=\Omega(1)$, by the triangle inequality:
\begin{align*}
   & \left\| \rho_0^{\text{output}}- \mathbb{E}_{V\sim \Haar(d)}\left[\rho_1^{\text{output}}\right] \right\|_1
   \\&\ge  \left\| \rho_0^{\text{output}}- \mathbb{E}_{V\sim \Haar(d)|\cE}\left[\rho_1^{\text{output}}\right] \right\|_1-  \left\| \mathbb{E}_{V\sim \Haar(d)}\left[\rho_1^{\text{output}}\right]- \mathbb{E}_{V\sim \Haar(d)|\cE}\left[\rho_1^{\text{output}}\right] \right\|_1
    \\&\ge\frac{2}{3}-2\exp\left(-\Omega(d)\right) \ge \frac{1}{3}.
\end{align*}
Writing the input state as $\rho= \sum_i \lambda_i \proj{\phi_i}$, e.g., using its spectral decomposition, the triangle equality implies:
\begin{align*}
   \sum_{i} \lambda_i  \left\| \rho_0^{\text{output}}(\proj{\phi_i})- \mathbb{E}_{V\sim \Haar(d)}\left[\rho_1^{\text{output}}(\proj{\phi_i}) \right]\right\|_1 \ge \frac{1}{3}.
\end{align*}
So there is a  unit vector $\ket{\phi}= \ket{\phi_i}$ such that:
\begin{align*}
 \left\| \rho_0^{\text{output}}(\proj{\phi})- \mathbb{E}_{V\sim \Haar(d)}\left[\rho_1^{\text{output}}(\proj{\phi})\right] \right\|_1 \ge \frac{1}{3}.
\end{align*}
We use the notation $U_t= \mathrm{e}^{-\mathrm{i}tH}$ and the following Kraus representation of each channel $\cN_t(\rho)= \sum_{k_i}A_{k_i} \rho A_{k_i}^\dagger $. Moreover, we will use the following shorthand notations for $\ell$, $m \in \mathbb N$, $\ell \leq m$:
\begin{align*}
\prod^\rightarrow_{i=\ell::m} X_i&:= X_{\ell} X_{\ell+1} \ldots X_{m}, \\
\prod^\leftarrow_{i=\ell::m} X_i&:= X_{m} X_{m-1} \ldots X_{\ell}.
\end{align*}
So we can write:
{\small 
\begin{align*}
 \rho_0^{\text{output}}(\proj{\phi})&=\sum_{k_1,\dots, k_{N-1}}  \left(\prod^\leftarrow_{i=1::N-1}A_{k_i}\right)\proj{\phi} \left(\prod^\rightarrow_{i=1::N-1} A_{k_i}^{\dagger}\right),
   \\ \rho_1^{\text{output}}(\proj{\phi})&=\hspace{-0.3em}  \sum_{k_1,\dots, k_{N-1}} \hspace{-0.3em}(U_{t_N}\otimes \dI)\Big(\prod^\leftarrow_{i=1::N-1}\hspace{-0.5em}A_{k_i}(U_{t_i}\otimes \dI)\Big) \proj{\phi} \Big(\prod^\rightarrow_{i=1::N-1} \hspace{-0.5em}(U_{t_i}^\dagger \otimes \dI)A_{k_i}^{\dagger} \Big)  (U_{t_N}\otimes \dI)^\dagger.
\end{align*}}
Hence, the triangle inequality implies 
{\small 
\begin{align*}
     &\left\| \mathbb{E}_{V\sim \Haar(d)}\left[\rho_1^{\text{output}}(\proj{\phi})\right]- \rho_0^{\text{output}}(\proj{\phi}) \right\|_1 
 \\&\le \mathbb{E}_{V} \sum_{k_1,\dots, k_{N-1}}\Bigg\| (U_{t_N}\otimes \dI)\Big(\prod^\leftarrow_{i=1::N-1}A_{k_i}(U_{t_i}\otimes \dI)\Big)\proj{\phi} \Big(\prod^\rightarrow_{i=1::N-1} (U_{t_i}^\dagger \otimes \dI)A_{k_i}^{\dagger} \Big)  (U_{t_N}\otimes \dI)^\dagger\\  & \qquad\qquad\qquad \qquad 
 -    \Big(\prod^\leftarrow_{i=1::N-1}A_{k_i}\Big)\proj{\phi} \Big(\prod^\rightarrow_{i=1::N-1} A_{k_i}^{\dagger} \Big)   \Bigg\|_1.
\end{align*}}
We can write the latter difference of states as a telescopic sum. A subsequent application of the triangle inequality yields:
{\small 
\begin{align}
    &\mathbb{E}_{V} \sum_{k_1,\dots, k_{N-1}}\Bigg\| (U_{t_N}\otimes \dI)\Big(\prod^\leftarrow_{i=1::N-1}A_{k_i}(U_{t_i}\otimes \dI)\Big)\proj{\phi} \Big(\prod^\rightarrow_{i=1::N-1} (U_{t_i}^\dagger \otimes \dI)A_{k_i}^{\dagger} \Big)  (U_{t_N}\otimes \dI)^\dagger \notag \\  & \qquad\qquad\qquad\qquad\qquad -    \Big(\prod^\leftarrow_{i=1::N-1}A_{k_i}\Big)\proj{\phi} \Big(\prod^\rightarrow_{i=1::N-1} A_{k_i}^{\dagger}\Big)   \Bigg\|_1 \notag
    \\&\le \mathbb{E}_{V} \sum_{s=1}^N \! \sum_{k_1,\dots, k_{N-1}}\! \Bigg\| \Big(\!\prod^\leftarrow_{s::N-1}\hspace{-0.7em}(U_{t_{i+1}}\otimes \dI) A_{k_{i}}\!\Big) [(U_{t_s}-\dI)\otimes \dI] \Big(\!\prod^\leftarrow_{i=1::s-1} \hspace{-0.7em}A_{k_i}\!\Big) \proj{\phi}  \Big(\!\prod^\rightarrow_{i=1::N-1} \hspace{-0.7em}(U_{t_i}^\dagger \otimes \dI)A_{k_i}^{\dagger}\! \Big) \Bigg\|_1 \tag{A}\label{proof-sequ}
    \\&\; + \mathbb{E}_{V} \sum_{s=1}^N \sum_{k_1,\dots, k_{N-1}}\Bigg\|  \Big(\prod^\leftarrow_{i=1::N-1} \hspace{-0.7em}A_{k_{i}}\Big)\proj{\phi} \Big(\prod^\rightarrow_{i=1::s-1} \hspace{-0.7em}A_{k_i}^{\dagger}\Big) [(U_{t_s}-\dI)\otimes \dI]^\dagger  \Big(\prod^\rightarrow_{i=s::N-1} \hspace{-0.7em}A_{k_i}^\dagger (U_{t_{i+1} }\otimes \dI)^\dagger \Big) \Bigg\|_1 \tag{B}\label{proof-sequ-B}
\end{align}}
On the one hand, using $U_t-\dI= (\mathrm{e}^{-\mathrm{i}\eta t}-1)\proj{v}$ we have 
{\small 
\begin{align*}
    &\left\|\Big(\!\prod^\leftarrow_{s::N-1}(U_{t_{i+1}}\otimes \dI) A_{k_{i}}\!\Big) [(U_{t_s}-\dI)\otimes \dI] \Big(\!\prod^\leftarrow_{i=1::s-1} A_{k_i}\!\Big)\proj{\phi}\notag  \Big(\!\prod^\rightarrow_{i=1::N-1} (U_{t_i}^\dagger \otimes \dI)A_{k_i}^{\dagger} \!\Big) \right\|_1
    \\&=|\mathrm{e}^{-\mathrm{i}\eta t_s}-1|\Bigg\|\Big(\!\prod^\leftarrow_{i=s::N-1}(U_{t_{i+1}}\otimes \dI) A_{k_{i}}\!\Big) (\proj{v}\!\otimes\!\dI) \Big(\!\prod^\leftarrow_{i=1::s-1} A_{k_i}\!\Big)\proj{\phi} \Bigg(\prod^\rightarrow_{i=1::N-1} (U_{t_i}^\dagger \otimes \dI)A_{k_i}^{\dagger} \Bigg)  \Bigg\|_1
    \\&\le|\mathrm{e}^{-\mathrm{i}\eta t_s}-1|\sqrt{\bra{\phi}\Big(\!\prod^\rightarrow_{i=1::N-1} (U_{t_i}^\dagger \otimes \dI)A_{k_i}^{\dagger}\!\Big) \Big(\!\prod^\leftarrow_{i=1::N-1} A_{k_{i}}(U_{t_i}\otimes \dI)\!\Big)  \ket{\phi}}
    \\&  \sqrt{\!\bra{\phi} \!\left[\! \prod^\rightarrow_{i=1::s-1}\hspace{-0.7em}A_{k_i}^\dagger\!\right] \![\proj{v}\!\otimes\!\dI] \!\left[\!\prod^\rightarrow_{i=s::N-1}\hspace{-0.7em}A_{k_i}^\dagger (U_{t_{i+1}}^\dagger\otimes \dI)\! \right] \!\left[\!\prod^\leftarrow_{i=s::N-1}\hspace{-0.7em} (U_{t_{i+1}}\otimes \dI) A_{k_i} \!\right]\![\proj{v}\!\otimes\!\dI] \!\left[ \!\prod^\leftarrow_{i=1::s-1}\hspace{-0.7em}A_{k_i}^\dagger\!\right]\!\ket{\phi}} 
\end{align*}
}
where we used the Cauchy-Schwarz inequality. Again, by using the  Cauchy-Schwarz inequality for the sums over $k_1, \dots, k_{N-1}$ and $\mathbb{E}_V$ as well as using the Kraus identities $\sum_{k_i} {A_{k_i}}^\dagger {A_{k_i}}= \dI $ in the last line we obtain:
{\small
\begin{align*}
   &\eqref{proof-sequ} \\ &=\sum_{s=1}^N \mathbb{E}_{V}\! \sum_{k_1,\dots, k_{N-1}}\!\left\|\left(\prod^\leftarrow_{s::N-1}\hspace{-0.7em}(U_{t_{i+1}}\otimes \dI) A_{k_{i}}\!\right) \![(U_{t_s}-\dI)\otimes \dI] \!\left(\prod^\leftarrow_{i=1::s-1} \hspace{-0.7em}A_{k_i}\!\right)\!\proj{\phi}  \!\left(\prod^\rightarrow_{i=1::N-1}\hspace{-0.7em} (U_{t_i}^\dagger \otimes \dI)A_{k_i}^{\dagger} \!\right) \right\|_1
    \\& \le\sum_{s=1}^N  \mathbb{E}_{V}  \sum_{k_1,\dots, k_{N-1}} |\mathrm{e}^{-\mathrm{i}\eta t_s}-1|\sqrt{\bra{\phi}\left(\prod^\rightarrow_{i=1::N-1} \hspace{-0.7em} (U_{t_i}^\dagger \otimes \dI)A_{k_i}^{\dagger}\right) \left(\prod^\leftarrow_{i=1::N-1} \hspace{-0.7em}A_{k_{i}}(U_{t_i}\otimes \dI)\right)  \ket{\phi}}
    \\& \sqrt{\bra{\phi}\! \left[ \prod^\rightarrow_{i=1::s-1}\hspace{-0.7em}A_{k_i}^\dagger\right] \![\proj{v}\!\otimes\!\dI] \!\left[\prod^\rightarrow_{i=s::N-1} \hspace{-0.7em}A_{k_i}^\dagger (U_{t_{i+1}}^\dagger\otimes \dI) \right]\! \left[\prod^\leftarrow_{i=s::N-1} \hspace{-0.7em}(U_{t_{i+1}}\otimes \dI) A_{k_i} \right] \! [\proj{v}\!\otimes\!\dI] \! \left[ \prod^\leftarrow_{i=1::s-1}\hspace{-0.7em} A_{k_i}^\dagger\right]\!\ket{\phi}}
    \\&\le \sum_{s=1}^N |\mathrm{e}^{-\mathrm{i}\eta t_s}-1|\sqrt{ \mathbb{E}_{V}  \sum_{k_1,\dots, k_{N-1}} \bra{\phi}\left(\prod^\rightarrow_{i=1::N-1} \hspace{-0.7em} (U_{t_i}^\dagger \otimes \dI)A_{k_i}^{\dagger}\right) \left(\prod^\leftarrow_{i=1::N-1} \hspace{-0.7em} A_{k_{i}}(U_{t_i}\otimes \dI)\right)  \ket{\phi}} \Bigg( \mathbb{E}_{V}  \sum_{k_1,\dots, k_{N-1}}
    \\& \bra{\phi} \!\left[ \prod^\rightarrow_{i=1::s-1}\hspace{-0.7em} A_{k_i}^\dagger\right] \![\proj{v}\!\otimes\!\dI] \!\left[\prod^\rightarrow_{i=s::N-1}\hspace{-0.7em} A_{k_i}^\dagger (U_{t_{i+1}}^\dagger\otimes \dI) \right] \!\left[\prod^\leftarrow_{i=s::N-1} \hspace{-0.7em} (U_{t_{i+1}}\otimes \dI) A_{k_i} \right] \! [\proj{v}\!\otimes\!\dI]  \!\left[ \prod^\leftarrow_{i=1::s-1} \hspace{-0.7em}A_{k_i}^\dagger\right]\!\ket{\phi}\!\Bigg)^{\frac{1}{2}}
    \\&= \sum_{s=1}^N |\mathrm{e}^{-\mathrm{i}\eta t_s}-1| \sqrt{\frac{1}{d}}
\end{align*}
}
as $\mathbb{E}_{V\sim \Haar(d)}\left[\proj{v}\right]= \mathbb{E}_{V\sim \Haar(d)}\left[V\proj{0}V^\dagger\right]=\frac{\tr(\proj{0})\dI}{d}=\frac{\dI}{d}$. 
On the other hand, using  $U_t-\dI= (\mathrm{e}^{-\mathrm{i}\eta t}-1)\proj{v}$ and the Cauchy-Schwarz inequality we have 
{\small
\begin{align*}
    &\left\| \left(\prod^\leftarrow_{i=1::N-1} A_{k_{i}}\right)\proj{\phi} \left(\prod^\rightarrow_{i=1::s-1} A_{k_i}^{\dagger}\right) [(U_{t_s}-\dI)\otimes \dI]^\dagger \left(\prod^\rightarrow_{i=s::N-1} A_{k_i}^\dagger (U_{t_{i+1} }\otimes \dI)^\dagger \right) \right\|_1
    \\&= |\mathrm{e}^{\mathrm{i}\eta t_s}-1|\left\|  \left(\prod^\leftarrow_{i=1::N-1} A_{k_{i}}\right)\proj{\phi} \left(\prod^\rightarrow_{i=1::s-1} A_{k_i}^{\dagger}\right) [\proj{v}\!\otimes\!\dI] \left(\prod^\rightarrow_{i=s::N-1} A_{k_i}^\dagger (U_{t_{i+1} }\otimes \dI)^\dagger \right) \right\|_1
    \\&\le |\mathrm{e}^{\mathrm{i}\eta t_s}-1| \sqrt{\bra{\phi}\left(\prod^\rightarrow_{i=1::s-1} A_{k_i}^{\dagger}\right)  \left(\prod^\leftarrow_{i=1::s-1} A_{k_i}^{\dagger}\right)\ket{\phi}}
    \\&\sqrt{\hspace{-0.3em}\bra{\phi} \hspace{-0.3em}\left[\prod^\rightarrow_{i=1::s-1} \hspace{-0.7em}A_{k_i}^{\dagger}\right]\hspace{-0.3em} [\proj{v}\!\otimes\!\dI] \hspace{-0.3em}\left[\prod^\rightarrow_{i=s::N-1} \hspace{-0.7em}A_{k_i}^\dagger (U_{t_{i+1} }\otimes \dI)^\dagger \right]\hspace{-0.5em}\left[\prod^\leftarrow_{i=s::N-1}  \hspace{-0.7em}(U_{t_{i+1} }\otimes \dI)A_{k_i} \right] \hspace{-0.3em}[\proj{v}\!\otimes\!\dI]\hspace{-0.3em}\left[\prod^\leftarrow_{i=1::s-1} \hspace{-0.7em} A_{k_i}^{\dagger}\right]\hspace{-0.3em}\ket{\phi}}.
\end{align*}
}
Hence by using the  Cauchy-Schwarz inequality for the sums over $k_1, \dots, k_{N-1}$ and the expectation $\mathbb{E}_V$ and by using the Kraus identities $\sum_{k_i} {A_{k_i}^\dagger} {A_{k_i}}= \dI $ in the last line we obtain:
{\small 
\begin{align*}
        &\eqref{proof-sequ-B}\\&=\sum_{s=1}^N\mathbb{E}_{V} \!  \sum_{k_1,\dots, k_{N-1}}\left\|  \left(\prod^\leftarrow_{i=1::N-1} \hspace{-0.7em}A_{k_{i}}\right)\proj{\phi} \left(\prod^\rightarrow_{i=1::s-1} \hspace{-0.7em}A_{k_i}^{\dagger}\right) [(U_{t_s}-\dI)\otimes \dI]^\dagger \left(\prod^\rightarrow_{i=s::N-1} \hspace{-0.7em}A_{k_i}^\dagger (U_{t_{i+1} }\otimes \dI)^\dagger \right)\right\|_1
        \\&\le \sum_{s=1}^N\mathbb{E}_{V}  \sum_{k_1,\dots, k_{N-1}}  |\mathrm{e}^{\mathrm{i}\eta t_s}-1| \sqrt{\bra{\phi}\left(\prod^\rightarrow_{i=1::s-1} A_{k_i}^{\dagger}\right)  \left(\prod^\leftarrow_{i=1::s-1} A_{k_i}^{\dagger}\right)\ket{\phi}} 
    \\&\sqrt{\bra{\phi}\! \left[\prod^\rightarrow_{i=1::s-1} \hspace{-0.7em}A_{k_i}^{\dagger}\right]\! [\proj{v}\!\otimes\!\dI] \!\left[\prod^\rightarrow_{i=s::N-1} \hspace{-0.7em}A_{k_i}^\dagger (U_{t_{i+1} }\otimes \dI)^\dagger \right]\!\left[\prod^\leftarrow_{i=s::N-1} \hspace{-0.7em} (U_{t_{i+1} }\otimes \dI)A_{k_i} \right] \![\proj{v}\!\otimes\!\dI] \!\left[\prod^\leftarrow_{i=1::s-1} \hspace{-0.7em}A_{k_i}^{\dagger}\right]\!\ket{\phi}}
    \\&\le \sum_{s=1}^N|\mathrm{e}^{\mathrm{i}\eta t_s}-1| \sqrt{\mathbb{E}_{V}  \sum_{k_1,\dots, k_{N-1}} \bra{\phi}\left(\prod^\rightarrow_{i=1::s-1} \hspace{-0.7em}A_{k_i}^{\dagger}\right)  \left(\prod^\leftarrow_{i=1::s-1} \hspace{-0.7em}A_{k_i}^{\dagger}\right)\ket{\phi}} \Bigg(\mathbb{E}_{V}  \sum_{k_1,\dots, k_{N-1}} \bra{\phi} \left[\prod^\rightarrow_{i=1::s-1} \hspace{-0.7em}A_{k_i}^{\dagger}\right] 
    \\&\qquad  [\proj{v}\!\otimes\!\dI] \left[\prod^\rightarrow_{i=s::N-1} \hspace{-0.7em}A_{k_i}^\dagger (U_{t_{i+1} }\otimes \dI)^\dagger \right]\left[\prod^\leftarrow_{i=s::N-1}  (U_{t_{i+1} }\otimes \dI)A_{k_i} \right] [\proj{v}\!\otimes\!\dI]\left[\prod^\leftarrow_{i=1::s-1} \hspace{-0.7em}A_{k_i}^{\dagger}\right]\ket{\phi} \Bigg)^{\frac{1}{2}}
    \\&= \sum_{s=1}^N|\mathrm{e}^{\mathrm{i}\eta t_s}-1| \sqrt{\frac{1}{d}}.
\end{align*}
}
Therefore, using 
\begin{align*}
    |\mathrm{e}^{\mathrm{i}\eta t}-1|&= \sqrt{(\cos(\eta t)-1)^2+\sin^2(\eta t)} =\sqrt{2(1-\cos(\eta t))}\\& \le \min\{\sqrt{2}, \sqrt{2(\eta t)^2}\}=\min\{\sqrt{2}, \sqrt{2}\eta t\},
\end{align*}
we get the following upper bound:
\begin{align*}
 \left\| \mathbb{E}_{V\sim \Haar(d)}\left[\rho_1^{\text{output}}(\proj{\phi})\right]- \rho_0^{\text{output}}(\proj{\phi}) \right\|_1 &\le \eqref{proof-sequ}+ \eqref{proof-sequ-B}
    \\&\le  2\sum_{s=1}^N |\mathrm{e}^{\mathrm{i}\eta t_s}-1| \sqrt{\frac{1}{d}} 
    \\&\le \frac{2}{\sqrt{d}}\sum_{s=1}^N  \min\{2, \sqrt{2} \eta t_s\}
\end{align*}
Finally, as $\left\| \mathbb{E}_{V\sim \Haar(d)}\left[\rho_1^{\text{output}}(\proj{\phi})\right] -\rho_0^{\text{output}}(\proj{\phi})\right\|_1 \ge \frac{1}{3}$, we deduce that
\begin{align*}
    \sum_{k=1}^N  t_k \ge \frac{\sqrt{d}}{6\sqrt{2}\eta}
\end{align*}
and 
\begin{align*}
    N  \ge \frac{\sqrt{d}}{6\sqrt{2}}. 
\end{align*}
We can set $\eps = \eta/4$ to finish the proof.
\end{proof}

\ins{\begin{remark}
    In Equation \eqref{Toy problem II}, we consider the problem of distinguishing between the trivial time evolution and what can be seen as a noisy version thereof. However, the noise model we use is somewhat artificial. It would be interesting to find examples of physically relevant Hamiltonians for which such hardness statements can be proved. We leave this question to future work.
\end{remark}}

\section{Upper bounds for Hamiltonian property testing}

\subsection{Upper bounds inherited from Hamiltonian learning}\label{sec:upper-bounds-from-hamiltonian-learning}

Before presenting our Hamiltonian property testing results, we discuss what known Hamiltonian learning results imply for testing. To the best of our knowledge, the procedures from \cite{caro2023learning, castaneda2023hamiltonian} are currently the only Hamiltonian learning algorithms from dynamics that work for arbitrary Hamiltonians without locality assumptions, and thus the only ones immediately applicable to our locality testing scenario via a naive ``learn general Hamiltonian, then check property'' approach. Analyzing their performance as testers will further highlight the importance of the chosen norms and in particular demonstrate that a different approach is needed when taking the normalized Frobenius norm as distance measure.

\cite{caro2023learning, castaneda2023hamiltonian} gave two different approaches -- the former based on Pauli shadow tomography methods applied to the Choi state of the forward short-time evolution in combination with Chebyshev interpolation for polynomial derivative estimation, the latter using forward and backward short-time evolution and block-encodings to create pseudo-Choi states as well as (classical) shadow tomography tools -- for learning an unknown Hamiltonian from query access.
For our purposes, it is important to carefully consider the performance measure in their learning task. Namely, the procedures of both papers produce $\norm{\cdot}_{\mathrm{Pauli},p}$ approximations to the coefficient vector of an arbitrary unknown Hamiltonian.
Concretely, \cite[Theorem 1.3]{caro2023learning} achieves this for $p=\infty$. That is, their Hamiltonian learning algorithm produces estimates $\hat{\alpha}_P$, $P\in \{\dI,X,Y,Z\}^{\otimes n}\setminus \{\dI^{\otimes n}\}$, that satisfy $\lvert\hat{\alpha}_P - \alpha_P\rvert\leq\varepsilon$ simultaneously for all $P\in \{\dI,X,Y,Z\}^{\otimes n}\setminus \{\dI^{\otimes n}\}$, using $\tilde{\mathcal{O}}\left(\frac{n \norm{H}_\infty^4}{\varepsilon^4}\right)$ queries to the Hamiltonian evolution, each for short time $t=\tilde{\mathcal{O}}\left(\frac{1}{\norm{H}_\infty}\right)$, thus leading to a total evolution time of $\tilde{\mathcal{O}}\left(\frac{n \norm{H}_\infty^3}{\varepsilon^4}\right)$. 
The guarantees in \cite{castaneda2023hamiltonian} are phrased for $p=2$ and use a number of queries to the Hamiltonian evolution that scales linearly with the number of Pauli terms in the Hamiltonian. While not discussed explicitly in \cite{castaneda2023hamiltonian}, this direct dependence on the number of terms can be removed when focusing on $p=\infty$.

As learning is a more demanding task than testing, the results of \cite{caro2023learning, castaneda2023hamiltonian} immediately imply (even tolerant) Hamiltonian locality testers with the same query complexities and total evolution times as their learning procedures. However, there is an important caveat: This works only for the norm $\norm{\cdot}_{\mathrm{Pauli},\infty}$. When trying to solve a Hamiltonian property testing problem \edit{w.r.t.~}{with respect to }$\norm{\cdot}_{\mathrm{Pauli},p}$ for any $1\leq p<\infty$, the only bounds that can be obtained immediately from \cite{caro2023learning, castaneda2023hamiltonian} (via Hölder's inequality) scale exponentially in $n$. This complication arises because in our testing task we do not want to make \emph{any} assumptions on the unknown Hamiltonian, in particular it can have exponentially many Pauli terms. When considering unnormalized Schatten $p$-norms $\norm{\cdot}_p$ on the level of the Hamiltonians, the situation is similarly bad if not worse, since naive attempts at controlling a $\norm{\cdot}_p$-difference even via the $\norm{\cdot}_{\mathrm{Pauli},1}$-difference incur an additional exponential overhead due to $\norm{P}_p = 2^{n/p}$ for all $n$-qubit Pauli strings $P$.  
Even normalizing the Schatten $p$-norms does not resolve this issue. For instance, by Parseval's identity, $\frac{1}{\sqrt{2^n}}\norm{\cdot}_2 = \norm{\cdot}_{\mathrm{Pauli},2}$, but we have observed above that the complexities of \cite{caro2023learning, castaneda2023hamiltonian} scale exponentially for the case of $\norm{\cdot}_{\mathrm{Pauli},2}$ and arbitrary unknown Hamiltonians with potentially exponentially many terms.

In summary, while the results of \cite{caro2023learning, castaneda2023hamiltonian} can in principle be used for Hamiltonian property testing, and even for the tolerant version thereof, they suffer from exponential query complexities and total evolution times for any of our norms of interest except for the weakest, $\norm{\cdot}_{\mathrm{Pauli},\infty}$. In particular, they do not give rise to query-efficient solutions for the operationally relevant norms $\norm{\cdot}_\infty$ and $\frac{1}{\sqrt{2^n}}\norm{\cdot}_{2}$. Additionally, their methods require potentially challenging-to-implement quantum capabilities (such as maximally entangled input states, access to both forward and backward time evolution, and/or entangled multi-copy measurements for shadows). Finally, neither of the two approaches achieves computational efficiency, even for $\norm{\cdot}_{\mathrm{Pauli},\infty}$.
Thus, while relevant for Hamiltonian learning, we consider \cite{caro2023learning, castaneda2023hamiltonian} insufficient for our Hamiltonian testing purposes and develop a new procedure that is tailored to the testing task at hand.

\paragraph{Note added.} In the property testing literature, there is another well known way of obtaining testing from learning, see for instance \cite[Proposition 2.1]{ron2008property}. Here, one first runs a proper learning algorithm for the class that one wishes to test and then checks whether the resulting hypothesis is close to the unknown object that generates the data. Instantiating this approach in our scenario to test whether an unknown Hamiltonian has property $S$ thus requires (a) a proper learning algorithm for Hamiltonians with property $S$ and (b) a procedure for tolerant Hamiltonian identity testing or for estimating the distance between a hypothesis Hamiltonian (a classical description of which is known, and which has property $S$) and a general unknown Hamiltonian, to which one has time evolution access. For the normalized Schatten $2$-norm, assuming the (non-trivial) ability to simulate the time evolution according to the hypothesis Hamiltonian, one can use the connection to the average-case fidelity between short-time-evolved input states drawn from the Haar measure (or from a $2$-design) established in \Cref{appendix:normalized-frobenius-norm} to achieve distance estimation as in (b). However, while the prior work reviewed in \Cref{subsection:related-work} achieved (a) for properties of Hamiltonians with a bounded-degree interaction graph, no efficient Hamiltonian learning algorithms for more general $S$, in particular for $k$-local Hamiltonians, were known. More recently, \cite[Result 1.6]{arunachalam2025testing} gave the first efficient algorithm for learning $k$-local Hamiltonians in normalized Schatten $2$-norm when auxiliary systems are available. Thus, the ``testing via (proper) learning'' template can now be used for Hamiltonian $k$-locality testing. However, as outlined above, this does not apply to more general $S$, requires auxiliary systems and Hamiltonian simulation capabilities, and inherits the exponential dependence on $k^2$ in both the total evolution time and the number of experiments from \cite[Result 1.6]{arunachalam2025testing}. Here, we aim for more broadly applicable and more resource-efficient Hamiltonian property testing procedures.

\subsection{Upper bounds in the randomized measurement framework} \label{sec:randomized-upper-bounds}

In this section, we will prove a general theorem that shows that efficient property testing is possible with respect to the normalized Schatten-$2$ norm, from which \Cref{inf-thm:hamiltonian-locality-testing-normalized-frobenius} follows as a special case.

\begin{definition}[Relation between states according to a property]
Let $\ket{\phi}$ and $\ket{\psi}$ be two unit vectors and $S\subset \mathds{P}_n$. We say that  $\ket{\phi}$ and $\ket{\psi}$ are equivalent under the property $S$ if they are equal or $\ket{\psi}$ can be obtained from $\ket{\phi}$ by applying a Pauli operator in $S$. We denote this  relation by $\sim_S$ and its negation as $\nsim_S$. Formally, 
 \begin{align*}
     \ket{\phi} \sim_S \ket{\psi} &\Leftrightarrow \exists \theta \in [0, 2\pi),\; \exists P\in S\cup \{\dI\} : P\ket{\phi}= \mathrm{e}^{\mathrm{i}\theta } \ket{\psi}
     \\&\Leftrightarrow \exists P\in S\cup \{\dI\} :  |\bra{\phi}P\ket{\psi}|= 1. 
 \end{align*}
 If $\ket{\phi} \nsim_S \ket{\psi}$, we say that a violation of the property $S$ is detected by the pair of unit vectors $(\ket{\phi}, \ket{\psi})$.
\end{definition}

With this definition in place, we can give the algorithm for testing the property $\Pi_S$ as \Cref{alg:property testing Hamiltonian}. Recall that the $\proj{\phi_{i,j}}$ are the MUBs constructed from stabilizer states defined in \Cref{eq:stabilizer_MUB}.
In \Cref{alg:property testing Hamiltonian}, the Hamiltonian is given as a black box that, given a time, runs the time evolution under the Hamiltonian for each input state provided and allows for any measurement at the end. The property $\Pi_S$ is given as a list of strings specifying the Pauli operators in $S$.

\SetKwComment{Comment}{/* }{ */}
\SetKwInOut{Input}{Input}
        \SetKwInOut{Output}{Output}
\begin{algorithm}[t!]
\caption{Testing Properties for Hamiltonian Evolutions}\label{alg:property testing Hamiltonian}
\LinesNumbered
\Input{A Hamiltonian $H$, a property $\Pi_S$, and an accuracy parameter $\varepsilon\in (0,1)$}
\Output{The null hypothesis $H_0$ or the alternate hypothesis $H_1$}
$t \gets \frac{\eps}{6}$\;
$N \gets   \left\lceil\frac{2\log(3)}{t^2\eps^2}\right\rceil  $\;
\For{$s\gets 1$ \KwTo $N$}{
Sample $i_s\sim\unif[d]$, $j_s\sim \unif[d+1]$\;
 Input state : $ \rho_s= \proj{\phi_{i_s,j_s}}$\;
 Evolve under $H$ for time $t$\;
 Measurement : $\cM_s = \{\proj{\phi_{i_s,\ell}}\}_{\ell}$ and observe $\ell_s\gets \cM_s(\cU_t(\rho_s))$\;
  \If{$\ket{\phi_{i_s,j_s}} \nsim_S \ket{\phi_{i_s,\ell_s}}$}
    {
        \Return $H_1$ and \textbf{stop}
    }
}
\Return $H_0$
\end{algorithm}

\begin{theorem}\label{thm:upper bound on testing} 
    Let $S\subset \mathds{P}_n$ such that  $|S\cup\{\dI\}|\le\frac{(2^n+1)^{}\eps^{4}}{144}$, and let $\varepsilon\in (0,1)$. 
    Suppose that the Hamiltonian $H$ satisfies $\tr(H)=0 $ and $\|H\|_\infty \le 1$. 
    \Cref{alg:property testing Hamiltonian} tests whether $H\in \Pi_S$ or $\frac{1}{\sqrt{2^n}}\|H-K\|_2>\eps$ for all Hamiltonians $K\in \Pi_S$ with probability at least $2/3$ using a total evolution time $\mathcal{O}\left(\frac{1}{\eps^3}\right)$, a total number of independent experiments $N=  \mathcal{O}\left(\frac{1}{\eps^4}\right)$, and a total classical processing time $\cO\left(\frac{n^2 |S\cup\{\dI\}|}{\eps^4}\right)$.
    Each experiment uses efficiently implementable states and measurements.
\end{theorem}

In fact, as we argue in the proof of \Cref{thm:upper bound on testing}, the procedure uses only stabilizer state inputs and stabilizer basis measurements. Each of these can be realized with Clifford circuits and thus with at most $\mathcal{O}(\frac{n^2}{\log n})$ many Hadamard, phase, and controlled-NOT gates \cite{aaronson2004improved}. Thus, \Cref{alg:property testing Hamiltonian} is efficient in terms of the number of experiments, the total evolution time, and the classical and quantum processing time.

Testing locality corresponds to the property $\Pi_{S_{k-\text{local}}}$, where  $S_{k-\text{local}}= \{P\in \mathds{P}_n : |P|\le k \}$ satisfies $|S|= \sum_{s=0}^k \binom{n}{s}3^s\le (3n)^{k+1}$. For this special case, we obtain:

\begin{corollary}[Testing locality \ins{-- Restatement of Theorem~\ref{inf-thm:hamiltonian-locality-testing-normalized-frobenius}}]\label{corollary:locality-testing-upper-bound}
    Let $n\ge  2$ and $\eps>0$ be such that  $ (3n)^{k+1}\le \frac{(2^{n}+1)^{}\eps^{4}}{144}$. 
    Suppose that the Hamiltonian $H$  satisfies $\tr(H)=0 $ and $\|H\|_\infty \le 1$. 
    \Cref{alg:property testing Hamiltonian}, when given the property $S=S_{k\mathrm{-loc}}$, tests whether $H$ is $k$-local or $\frac{1}{\sqrt{d}}\|H-H_{\text{local}}\|_2>\eps$ for all $k$-local Hamiltonians $H_{\text{local}}$ with probability at least $2/3$ using a total evolution time $\mathcal{O}\left(\frac{ 1}{\eps^3}\right)$, a total number of independent experiments $N=  \mathcal{O}\left(\frac{1}{\eps^4}\right)$, and a total classical processing time $\cO\left(\frac{(3n)^{k+3}}{\eps^4}\right)$.
\end{corollary}

\begin{remark}[Testing many properties] 
    In situations where we are interested in testing many properties at once or we are not confident about the exact property we want to test during the data acquisition phase, we should find a way to perform the Hamiltonian property testing with an arbitrarily small error probability. It turns out that changing the data processing (statistic/estimator) part of \Cref{alg:property testing Hamiltonian} slightly solves this issue. Concretely, using the concentration of an estimator that compares the empirical number of violations, i.e., counting how many outcomes are measured such that $\ket{\phi_{i_s,j_s}} \nsim_S \ket{\phi_{i_s,\ell_s}}$, with a threshold, we are able to achieve an error probability $\delta$ with a complexity that scales as $\log(1/\delta)$. Via a union bound and setting $\delta \mapsto \delta/M$, this allows us to test many properties at once with only an overhead that is logarithmic in $M$, number of properties. See \Cref{testing many ppts} for details. 
\end{remark}

\begin{remark}[Assumption on the set $S$] 
    The assumption on the set $S$\,---\,$|S\cup\{\dI\}|\le\frac{(2^n+1)^{}\eps^{4}}{144}$\,---\,for which we can prove a rigorous guarantee on the complexity of \Cref{alg:property testing Hamiltonian} limits the range of properties we can test. If we are only interested in having efficient tests and thus $|S|\leq\mathcal{O}(\mathrm{poly}(n))$, this assumption should always be satisfied. Nonetheless, we are able to extend the result of \Cref{thm:upper bound on testing} to any set $S$ using an additional number $n_{\mathrm{aux}}$ of ancilla qubits, where $n_{\mathrm{aux}}= \left\lceil\log_2\left(\frac{144\cdot  |S \cup \{\dI\}|}{2^n\eps^4}\right)\right\rceil$. See \Cref{assumption on S} for details. 
\end{remark}

\begin{remark}[Tolerant testing]
    In \Cref{thm:upper bound on testing}, the null hypothesis is that the unknown Hamiltonian $H$ is itself an element of $\Pi_S$. In the spirit of tolerant property testing \cite{parnas2006tolerant}, one may aim to weaken this to $H$ merely being close to the set $\Pi_S$. In \Cref{sec:tolerant-testing}, we extend our result to this tolerant Hamiltonian property testing scenario.
\end{remark}

\begin{remark}[Assumption of  independent and identically distributed (i.i.d.) input]
    We assumed that the tester runs $N$ i.i.d.\ unitary evolutions in order to decide the correct hypothesis. 
    Note that the proof works even if the unitary evolutions $\cU_1, \dots, \cU_N$ are not identical, i.e., the Hamiltonians may be different as long as they remain in the same hypothesis class. However, the assumption of independence is crucial for the proof. 
    Moreover, the naive application of the de Finetti theorem (on the Choi state, as in \cite{fawzi_learning_2024})  would require an overhead in the copy complexity that is exponential in the number of qubits.
\end{remark}

\begin{remark}[Implications for testing \edit{w.r.t.~}{with respect to }other norms]
    \Cref{thm:upper bound on testing} is phrased in terms of $\frac{1}{\sqrt{2^n}}\norm{\cdot}_2$. Recalling that $\frac{1}{\sqrt{2^n}}\norm{\cdot}_2 = \norm{\cdot}_{\mathrm{Pauli},2}$ as well as the monotonicity $\norm{\cdot}_{\mathrm{Pauli},p}\leq \norm{\cdot}_{\mathrm{Pauli},q}$ for $1\leq q\leq p\leq \infty$, we immediately see that the results of \Cref{thm:upper bound on testing} also apply to Hamiltonian property testing with $\norm{\cdot}_{\mathrm{Pauli},p}$ for any $p\geq 2$ as distance measure.
    Similarly, as the normalized Schatten $p$-norms satisfy the monotonicity property $\frac{1}{2^{n/p}}\norm{\cdot}_p \leq \frac{1}{2^{n/q}}\norm{\cdot}_q$ for $1\leq p\leq q \leq\infty$, the results of \Cref{thm:upper bound on testing} immediately carry over to testing \edit{w.r.t.~}{with respect to }$\frac{1}{2^{n/p}}\norm{\cdot}_p$ for any $1\leq p\leq 2$.
\end{remark}

\begin{proof}[Proof of \Cref{thm:upper bound on testing}]
We need to show that the error probability under the null and alternate hypotheses is at most $1/3$. 
\paragraph{Error probability under the alternate hypothesis.} Here we suppose that $H$ is $\eps$-far from $\Pi_S$. Let $\bm{i}=(i_1, \dots, i_N)$, $\bm{j}=(j_1, \dots, j_N)$ and $\bm{\ell}=(\ell_1, \dots, \ell_N)$.
The error probability is 
\begin{align*}
	\exs{\bm{i,j,\ell}}{\prs{H_1}{\ket{\phi_{i_1,\ell_1}} \sim_S \ket{\phi_{i_1,j_1}},  \; \cdots \;, \ket{\phi_{i_N,\ell_N}} \sim_S \ket{\phi_{i_N,j_N}}} }= \exs{i_1,j_1,\ell_1}{\pr{\ket{\phi_{i_1,\ell_1}} \sim_S \ket{\phi_{i_1,j_1}} }}^N.
\end{align*}
\edit{We have}{We now evaluate $\exs{i_1,j_1,\ell_1}{\pr{\ket{\phi_{i_1,\ell_1}} \sim_S \ket{\phi_{i_1,j_1}} }} = \exs{i,j,\ell}{\pr{\ket{\phi_{i,\ell}} \sim_S \ket{\phi_{i,j}} }}$ as follows:}
{\small
\begin{align*}
	&\exs{i,j,\ell}{\pr{\ket{\phi_{i,\ell}} \sim_S \ket{\phi_{i,j}}} }\\&=\frac{1}{d(d+1)}\sum_{i,j,\ell} |\bra{\phi_{i,\ell}} \mathrm{e}^{\mathrm{i}tH} \ket{\phi_{i,j}}|^2 \mathbf{1}\left(\left\{ \ket{\phi_{i,\ell}} \sim_S \ket{\phi_{i,j}}  \right\}\right)
	\\&=\frac{1}{d(d+1)}\sum_{i,j,\ell} |\bra{\phi_{i,\ell}} \mathrm{e}^{\mathrm{i}tH} \ket{\phi_{i,j}}|^2 \mathbf{1}\left(\left\{ \exists \theta,\; \exists P\in S\cup\{\dI\}: \mathrm{e}^{\mathrm{i}\theta} \ket{\phi_{i,\ell}} =P  \ket{\phi_{i,j}}  \right\}\right)
	\\&\le  \frac{1}{d(d+1)} \sum_{P\in S\cup\{\dI\}}\sum_{i,j} |\bra{\phi_{i,j}} P \mathrm{e}^{\mathrm{i}tH} \ket{\phi_{i,j}}|^2 \sum_{\ell} \mathbf{1}\left(\left\{ \exists \theta: \mathrm{e}^{\mathrm{i}\theta}  \ket{\phi_{i,\ell}} =P  \ket{\phi_{i,j}}  \right\}\right)
	\\&\overset{(a)}{\le} \frac{1}{d(d+1)} \sum_{P\in  S\cup\{\dI\}}\sum_{i,j} |\bra{\phi_{i,j}} P \mathrm{e}^{\mathrm{i}tH} \ket{\phi_{i,j}}|^2 
	\\&\overset{(b)}{=}  \sum_{P\in  S\cup\{\dI\}} \frac{\tr\left(P \mathrm{e}^{\mathrm{i}tH} \mathrm{e}^{-\mathrm{i}tH} P^\dagger  \right)+ \left|\tr\left(P \mathrm{e}^{\mathrm{i}tH}\right)\right|^2}{d(d+1)} 
	\\&\overset{(c)}{=}  \sum_{P\in  S\cup\{\dI\}} \frac{d}{d(d+1)}+  \sum_{P\in  S\cup\{\dI\}} \frac{1}{d(d+1)} \left|\sum_{m\ge 0} \frac{(\mathrm{i}t)^m}{m!} \tr(PH^m) \right|^2
\end{align*}
}
where we used in $(a)$ that $\sum_{\ell} \mathbf{1}\left(\left\{ \exists \theta: \mathrm{e}^{\mathrm{i}\theta}  \ket{\phi_{i,\ell}} =P  \ket{\phi_{i,j}}  \right\}\right)\le 1$ because, if we have $\theta_1,\ell_1,  \theta_2, \ell_2$ such that $\mathrm{e}^{\mathrm{i}\theta_1}  \ket{\phi_{i,\ell_1}} =P  \ket{\phi_{i,j}} = \mathrm{e}^{\mathrm{i}\theta_2}  \ket{\phi_{i,\ell_2}}  $, then $|\spr{\phi_{i,\ell_1}}{\phi_{i,\ell_2}}|=1$, but $\{\ket{\phi_{i,\ell}}\}_i$ is an orthonormal basis, so $\ell_1=\ell_2$. In $(b)$, we used the fact that $\{\ket{\phi_{i,j}}\}_{i,j}$ forms a $2$-design.

The first term of $(c)$ can be computed exactly: 
\begin{align*}
	\sum_{P\in  S\cup\{\dI\}} \frac{d}{d(d+1)} =\frac{|S\cup\{\dI\}|}{d+1}. 
\end{align*}
For the second term of $(c)$, we deal first with the case $P=\dI$:
\begin{align*}
	\frac{1}{d(d+1)} \left|\sum_{m\ge 0} \frac{(\mathrm{i}t)^m}{m!} \tr(PH^m)\right|^2&= \frac{1}{d(d+1)} \left|\sum_{m\ge 0} \frac{(\mathrm{i}t)^m}{m!} \tr(H^m)\right|^2
	\\&\le \frac{1}{d(d+1)} \left|  d- \frac{t^2}{2}\tr(H^2)+ \sum_{m\ge 3} \frac{(\mathrm{i}t)^m}{m!} \tr(H^m)\right|^2
		\\&\le \frac{1}{d(d+1)} \left(  d^2-dt^2\tr(H^2) + 4d^2t^4 \right)
		\\&= \frac{1}{(d+1)} \left(  d-t^2d\sum_P |\alpha_P|^2 + 4dt^4 \right) \, .
\end{align*}
In the third step, we evaluated the squared absolute value $|z|^2=\Bar{z}z$ and bound the higher order terms, using that $|\tr(H^m)| \le d\|H\|^m_\infty\le d$ and $\sum_{m\ge 4} \frac{t^m}{m!}\tr(H^m)\le 0.06\cdot dt^4  $ since $t\leq 1$. Note that third order terms vanish.

For the other cases in  $(c)$: 
{\small
\begin{align*}
    &\sum_{P\in S\setminus\{\dI\}}	\frac{1}{d(d+1)} \left|\sum_{m\ge 0} \frac{(\mathrm{i}t)^m}{m!} \tr(PH^m)\right|^2
    \\&= \frac{1}{d(d+1)} \sum_{P\in S\setminus\{\dI\}} \left|  (\mathrm{i}t)d\alpha_P + \sum_{m\ge 2} \frac{(\mathrm{i}t)^m}{m!} \tr(PH^m)\right|^2
    \\&= \!\sum_{P\in S\setminus\{\dI\}} \hspace{-0.3em} \frac{|(\mathrm{i}t)d\alpha_P|^2}{d(d\!+\!1)} \! +\!\frac{1}{d(d\!+\!1)}\hspace{-0.4em}\sum_{P\in S\setminus\{\dI\}}\hspace{-0.2em}\Big| \!  \sum_{m\ge 2}\!\frac{(\mathrm{i}t)^m}{m!} \tr(PH^m)\!\Big|^2\hspace{-0.5em}+\! 2\Re\! \sum_{P\in S\setminus\{\dI\}} \!\frac{\mathrm{i}td\alpha_P}{d(d\!+\!1)} \!\sum_{m\ge 2} \!\frac{(-\mathrm{i}t)^m}{m!} \tr(PH^m)
     \\&\le \! \sum_{P\in S\setminus\{\dI\}}\! \frac{|(\mathrm{i}t)d\alpha_P|^2}{d(d\!+\!1)} \! +\!\frac{1}{d(d\!+\!1)}\!\sum_{P}\!\Big|  \sum_{m\ge 2} \!\frac{(\mathrm{i}t)^m}{m!} \tr(PH^m)\Big|^2 \!+\!  \frac{2td }{d(d\!+\!1)}\!\sum_{m\ge 3} \!\frac{t^m}{m!}  \sum_{P\in S\setminus\{\dI\}}|\alpha_P|\cdot |\tr(PH^m)|
    \\&\le\!\sum_{P\in S\setminus\{\dI\}} \frac{dt^2\alpha_P^2}{d\!+\!1}\!  +\hspace{-0.5em}\sum_{P, m,m'\ge 2}\! \frac{(\mathrm{i}t)^m (-\mathrm{i}t)^{m'} }{d(d\!+\!1)m!m'!} \tr(PH^m) \tr(PH^{m'}) \!+\!   \frac{2t}{d\!+\!1} \! \sum_{m\ge 3} \!\frac{t^m}{m!} \sqrt{\sum_P \!\alpha_P^2\!\sum_P\!  |\tr(PH^m)|^2}
    \\&\le\frac{d}{d+1}t^2 \sum_{P\in S\setminus\{\dI\}} \alpha_P^2+ \frac{1}{d+1}  \sum_{m,m'\ge 2} \frac{(\mathrm{i}t)^m (-\mathrm{i}t)^{m'} }{m!m'!} \tr(H^mH^{m'}) +\frac{2t}{d+1}  \sum_{m\ge 3} \frac{t^m}{m!} \sqrt{d\tr(H^mH^m)}
    \\&\le \frac{d}{d+1}t^2\sum_{P\in S\setminus\{\dI\}} \alpha_P^2+ \frac{1}{d+1}  \sum_{m,m'\ge 2} \frac{t^m }{m!}\cdot \frac{t^{m'} }{m'!} \cdot d +\frac{2t}{d+1}\cdot t^3\cdot \sqrt{d\cdot d}
    \\&\le \frac{d}{d+1}t^2 \left( \sum_{P\in S\setminus\{\dI\}} \alpha_P^2\right)+ 3t^4 \,,
\end{align*}
}
where we use $\sum_P \alpha_P^2\le 1$, $\frac{1}{d}\sum_P \tr(PA)\tr(PB)=\tr(AB)$, $|\tr(H^m)| \le d\|H\|^m_\infty\le d$, $\sum_{m\ge 2}\frac{t^m}{m!} =e^t-1-t\le t^2$ and $\sum_{m\ge 3}\frac{t^m}{m!} =e^t-1-t-\frac{t^2}{2}\le t^3$.

Therefore 
\begin{align}\label{eq:err_under_H_1}
	\exs{i,j,\ell}{\pr{\ket{\phi_{i,\ell}} \sim_S \ket{\phi_{i,j}}} }&= \frac{|S\cup\{\dI\}|}{d+1} +  \sum_{P\in S \cup \{\dI\}} \frac{1}{d(d+1)} \left|\sum_{m\ge 0} \frac{(\mathrm{i}t)^m}{m!} \tr(PH^m) \right|^2 \notag
	\\&\le \frac{|S\cup\{\dI\}|}{d+1}+\frac{1}{d+1}\left(d-t^2d\sum_{P\notin S} |\alpha_P|^2 + 7dt^4  \right) \notag 
	\\&\le 1-t^2\eps^2 + \frac{|S\cup\{\dI\}|}{d+1}+ 7t^4 
	\\&\le 1-\frac{t^2\eps^2}{2} \notag
\end{align}
where we used 
\begin{itemize}
	\item For all $K\in \Pi_S$,  $\frac{1}{\sqrt{d}}\|H-K\|_2\ge \eps $ so $ \sum_{P\notin S} \alpha_P^2 \ge \eps^2$ (choose $K= \sum_{P\in S} \frac{1}{d}\tr(PH)P$),
	\item $\frac{|S\cup\{\dI\}|}{d+1}\le \frac{t^2\eps^2}{4}$ and $7t^4\le \frac{t^2\eps^2}{4}$. 
\end{itemize}
For $ t=\frac{\eps}{6}$, it is sufficient to add the condition that the size of the set $S$ satisfies 
\begin{align*}
|S\cup\{\dI\}|\le \frac{(d+1)\cdot \eps^4}{144}.
\end{align*}
We can choose the number of repetitions to be $N=  \frac{2\log(3)}{t^2\eps^2}  $ so that:
\begin{align*}
	&\exs{\bm{i,j,\ell}}{\prs{H_1}{\ket{\phi_{i_1,\ell_1}} \sim_S \ket{\phi_{i_1,j_1}},  \; \cdots \;, \ket{\phi_{i_N,\ell_N}} \sim_S \ket{\phi_{i_N,j_N}}}}
 \\&= \exs{i_1,j_1,l_1}{\pr{\ket{\phi_{i_1,\ell_1}} \sim_S \ket{\phi_{i_1,j_1}} }}^N
	\le \left(1-\frac{t^2\eps^2}{2}\right)^N\le \frac{1}{3}.
\end{align*}

\paragraph{Error probability under the null hypothesis.} Here we suppose that $H\in \Pi_S$. The error probability is 
\begin{align*}
	&\exs{\bm{i,j,\ell}}{\prs{H_1}{\ket{\phi_{i_1,\ell_1}} \nsim_S \ket{\phi_{i_1,j_1}} \; \text{or}  \; \cdots  \; \text{or}\;  \ket{\phi_{i_N,\ell_N}} \nsim_S \ket{\phi_{i_N,j_N}}} }\\&\le N \cdot \exs{i_1,j_1,l_1}{\pr{\ket{\phi_{i_1,\ell_1}} \nsim_S \ket{\phi_{i_1,j_1}} }}.
\end{align*}
Observe that under the event $\left\{ \ket{\phi_{i,\ell}} \nsim_S \ket{\phi_{i,j}}  \right\} $ we have $\ell\neq j$ since $\dI\in S\cup \{\dI\}$.
Also $\ket{\phi_{i,\ell}} \nsim_S \ket{\phi_{i,j}} $ implies 
$\bra{\phi_{i,\ell}} H \ket{\phi_{i,j}}=0$ since $H\in \Pi_S$ . To see this we can write $\proj{\phi_{i,j}}= \frac{1}{d}\sum_{p\in G_i}(-1)^{ r_j^i\circ p}S(p)$ where $r_j^i\in A_{G_i}$, then for $Q\in \mathbb P_n$ we have 
\begin{align}
	|\bra{\phi_{i,\ell}} Q \ket{\phi_{i,j}}|^2 
      &= \frac{1}{d^2} \sum_{p_1, p_2\in G_i} (-1)^{r^i_\ell\circ p_1+ r^i_j\circ p_2 }\tr(S(p_1) Q S(p_2) Q )\nonumber
	\\&=\frac{1}{d^2} \sum_{p_1, p_2\in G_i} (-1)^{r^i_\ell\circ p_1+ r^i_j\circ p_2 +q\circ p_2}\tr(S(p_1)  S(p_2) Q Q )\nonumber
	\\&=\frac{1}{d} \sum_{p\in G_i} (-1)^{r^i_\ell\circ p+ r^i_j\circ p +q\circ p}\nonumber
	\\&= \mathbf{1}\left(\left\{r^i_\ell r^i_j q \in G_i\right\}\right) \label{eq:overlap-check}\, ,
\end{align}
where we wrote $q$ for the element in $\mathbf{P}_n$ corresponding to $Q \in \mathbb P_n$. In the last line, we have used \Cref{lem:summing-signs-over-groups}. Moreover, we used the fact $r\circ p + s\circ p = (rs)\circ p$ which can be seen from $(-1)^{(rs)\circ p}P(r)P(s)P(p)=P(p)P(r)P(s)= (-1)^{r\circ p}P(r)P(p)P(s)= (-1)^{r\circ p} (-1)^{s\circ p}P(r)P(s)P(p)$. 
Thus either $\bra{\phi_{i,\ell}} Q \ket{\phi_{i,j}}=0 $ or $|\bra{\phi_{i,\ell}} Q \ket{\phi_{i,j}}|= 1 \Leftrightarrow  Q \ket{\phi_{i,j}} = \mathrm{e}^{\mathrm{i}\theta }  \ket{\phi_{i,\ell}}$.
Hence, under the event $\left\{ \ket{\phi_{i,\ell}} \nsim_S \ket{\phi_{i,j}}  \right\} $ we have $\bra{\phi_{i,\ell}} H^m \ket{\phi_{i,j}} =0$  for  $m=0,1 $. Therefore:
 {\small
\begin{align*}
    \exs{i,j,\ell}{\pr{\ket{\phi_{i,\ell}} \nsim_S \ket{\phi_{i,j}} }}
    &=\frac{1}{d(d+1)} \sum_{i=1}^{d+1}\sum_{j\neq \ell}|\bra{\phi_{i,\ell}} \mathrm{e}^{\mathrm{i}tH} \ket{\phi_{i,j}}|^2 \mathbf{1}\left(\left\{ \ket{\phi_{i,\ell}} \nsim_S \ket{\phi_{i,j}}  \right\}\right)
    \\&=  \frac{1}{d(d+1)} \sum_{i=1}^{d+1}\sum_{j\neq \ell}\left|\sum_{m\ge 0} \frac{(\mathrm{i}t)^m}{m!} \bra{\phi_{i,\ell}} H^m \ket{\phi_{i,j}}\right|^2 \mathbf{1}\left(\left\{ \ket{\phi_{i,\ell}} \nsim_S \ket{\phi_{i,j}}  \right\}\right)
    \\&\le   \frac{1}{d(d+1)} \sum_{i=1}^{d+1}\sum_{j, \ell}\left|\sum_{m\ge 2} \frac{(\mathrm{i}t)^m}{m!} \bra{\phi_{i,\ell}} H^m \ket{\phi_{i,j}}\right|^2 
    \\&=  \frac{1}{d(d+1)} \sum_{i=1}^{d+1}\sum_{j, \ell}\sum_{m,m'\ge 2} \frac{(\mathrm{i}t)^m}{m!}\cdot \frac{(-\mathrm{i}t)^{m'}}{m'!} \bra{\phi_{i,\ell}} H^m \ket{\phi_{i,j}}\bra{\phi_{i,j}} H^{m'} \ket{\phi_{i,\ell}}
    \\&=   \frac{1}{d(d+1)} \sum_{i=1}^{d+1}\sum_{m,m'\ge 2} \frac{(\mathrm{i}t)^m}{m!}\cdot \frac{(-\mathrm{i}t)^{m'}}{m'!} \tr(H^m  H^{m'})
    \\&\le  \frac{1}{d(d+1)} \sum_{i=1}^{d+1}\sum_{m,m'\ge 2} \frac{t^m}{m!}\cdot \frac{t^{m'}}{m'!} |\tr(H^m  H^{m'})|
    \\&\le  \frac{1}{d(d+1)} \sum_{i=1}^{d+1}t^4 d
    =t^4\, ,
\end{align*}}
where we used that $\{\ket{\phi_{i,j}}\}_j$ is an orthonormal basis, $|\tr(H^m)|\le d$ and $\sum_{m\ge 2} \frac{(t)^m}{m!}\le t^2$ .
   
Hence, for $t \le \frac{\eps}{6}$ and $N=  \left\lceil\frac{2\log(3)}{t^2\eps^2}\right\rceil$ we have that
\begin{align*}
	&\exs{\bm{i,j,\ell}}{\prs{H_1}{\ket{\phi_{i_1,\ell_1} }\nsim_S \ket{\phi_{i_1,j_1}} \; \text{or}  \; \cdots  \; \text{or}\;  \ket{\phi_{i_N,\ell_N}} \nsim_S \ket{\phi_{i_N,j_N}} }  }
   \\ &\le N \cdot \exs{i_1,j_1,\ell_1}{\pr{\ket{\phi_{i_1,\ell_1}} \nsim_S \ket{\phi_{i_1,j_1}} }}
	\le  \left\lceil\frac{2\log(3)}{t^2\eps^2}\right\rceil t^4 
	 \\&  \le \frac{2\log(3)}{\eps^2}t^2 + t^4 
	\le \frac{2\log(3)}{36} + \frac{1}{6^4}
    <\frac{1}{3}.
\end{align*}

\paragraph{Complexity}

\begin{itemize}
	\item Evolution time at each step $t=\frac{\eps}{6}$,
	\item Number of independent experiments $N=  \left\lceil\frac{2\log(3)}{t^2\eps^2}\right\rceil  =  \left\lceil\frac{72\log(3)} {\eps^4}\right\rceil $,
	\item Total evolution time $Nt= \left\lceil\frac{72\log(3)} {\eps^4}\right\rceil\cdot\frac{\eps}{6}\leq \frac{12\log(3)} {\eps^3} + \underbrace{\frac{\eps}{6}}_{\leq 1/6}$,
 \item Total classical processing time: In each round $1\leq s\leq N$, we check whether $\lvert \bra{\phi_{i_s,j_s}}Q\ket{\phi_{i_s,\ell_s}}\rvert=1$ for any $Q\in S\cup\{\dI\}$. Via \Cref{eq:overlap-check}, this reduces to checking whether $r^{i_s}_{\ell_s}r^{i_s}_{j_s}q\in G_{i_s}$. As $G_{i_s}$ is a maximal Abelian subgroup, it equals its own commutator, so we can equivalently check whether $r^{i_s}_{\ell_s}r^{i_s}_{j_s}q$ commutes with all $n$ generators of $G_{i_s}$. Each such commutation check can be performed in time $\mathcal{O}(n)$, see e.g. \cite[Section 3]{steeb2014spin}. This leads to a total classical processing time of $\cO\left(Nn^2 \lvert S\cup\{\dI\}\rvert\right)=\cO\left(\frac{n^2|S\cup\{\dI\}|}{\eps^4}\right)$.
    \item Complexity of state preparation and measurements: The input states used in our protocol are (randomly chosen) stabilizer states $\ket{\phi_{i,j}}$, each of which can be prepared using a Clifford circuit with $\mathcal{O}(\frac{n^2}{\log n})$ many $2$-qubit gates \cite{aaronson2004improved}. The measurements used in our protocol are \edit{w.r.t.~}{with respect to }some (randomly chosen) orthonormal basis of stabilizer states. Given any such a target basis, there always exists a Clifford unitary that maps the computational basis to the target basis.
    This can be seen via the surjective group homomorphism between $n$-qubit Cliffords and the symplectic automorphism group on $(2n)$-bit strings \cite[Proposition II.14.]{haah2016algebraic}, under which stabilizer groups correspond to (symplectically) isotropic subspaces, and then using the fact that every symplectic subspace can be obtained by applying a symplectic transformation to a span of a suitable number of canonical basis elements in $\mathbb{F}_2^{2n}$ \cite[Proposition II.8]{haah2016algebraic}.
    Thus, each of our measurements can be implemented by a Clifford circuit with $\mathcal{O}(\frac{n^2}{\log n})$ many $2$-qubit gates \cite{aaronson2004improved} followed by a computational basis measurement.
\end{itemize}

\end{proof}

\section{Lower bounds for learning a general Hamiltonian}

In this section, we prove \Cref{inf-thm:hardness-hamiltonian-learning} on lower bounds for learning a general Hamiltonian in normalized Schatten $2$-norm. 

\begin{theorem}[\ins{Formal version of \Cref{inf-thm:hardness-hamiltonian-learning}}]\label{thm:hardness-hamiltonian-learning}
    Let $n\ge  11$. 
	Suppose that the Hamiltonian $H$  satisfies $\tr(H)=0 $ and $\|H\|_\infty \le 1 $. 
 \begin{itemize}
     \item Any  coherent algorithm that constructs $\hat{H}$ such that  $\frac{1}{\sqrt{d}}\|H-\hat{H}\|_2\le \eps$  with probability at least $2/3$ has to use  a number of independent experiments  $N=\Omega\left(\frac{2^{2n}}{n}\right)$.
     \item Any non-adaptive ancilla-free incoherent algorithm that constructs $\hat{H}$ such that  $\frac{1}{\sqrt{d}}\|H-\hat{H}\|_2\le \eps$  with probability at least $2/3$ has to use a total evolution time  $\Omega\left(\frac{2^{2n}}{\eps}\right)$.
 \end{itemize}
\end{theorem}

\ins{We note that the hardness results of Theorem~\ref{inf-thm:hardness-hamiltonian-learning} immediately carries over to normalized $p$-norm with $p\geq 2$ via standard norm inequalities.}

We follow a standard strategy for proving lower bounds for learning problems~\cite{flammia2012quantum,haah2017sample,lowe2022lower,fawzi2023lower,oufkir2023sample}. 
\begin{proof}
	We use the following construction inspired by \cite{bubeck2020entanglement}:
	\begin{align*}
		H_{U}= \eps U O U^\dagger 	\quad \text{where} \quad  U\sim \Haar(d), \text{ and } O=  \operatorname{diag}(\underbrace{+1,\ldots,+1}_{\frac{d}{2}\textrm{ times}},\underbrace{-1,\ldots,-1}_{\frac{d}{2}\textrm{ times}}).
	\end{align*}
    Note that in particular $\tr[O]=0$ by construction.    
    For such a Hamiltonian we have $\|H\|_\infty \le \eps \|U\|_\infty\|O\|_\infty \|U^\dagger \|_\infty=\eps\le 1 $.
	The expected distance  between two independent Hamiltonians $H_U$ and $H_{V}$ satisfies:
		\begin{align*}
		\ex{\frac{1}{d}\|H_{U}-H_{V}\|_2^2}&= \frac{\eps^2}{d}\ex{\tr\left[(UOU^\dagger)^2\right] + 2 \tr\left[UOU^\dagger VOV^\dagger \right] +\tr\left[(VOV^\dagger)^2\right]}
		\\&= 2\eps^2 +\frac{2\eps^2}{d}\ex{\tr\left[UOU^\dagger O\right]}
		\\&= 2\eps^2 +\frac{2\eps^2}{d^2}(\tr[O])^2
		\\&=2\eps^2\, .
	\end{align*}
    Here, the second equality used invariance of the Haar measure, the third equality used $\ex{U O U^\dagger} = \frac{\tr[O]}{d}\dI $, and the last equality used $\tr[O]=0$.
	Moreover for $d\ge 2$,
		\begin{align*}
		\ex{\frac{1}{d^2}\|H_{U}-H_{V}\|_2^4}&= \frac{\eps^4}{d^2}\ex{\left(\tr\left[(UOU^\dagger)^2\right] + 2 \tr\left[UOU^\dagger VOV^\dagger \right] +\tr\left[(VOV^\dagger)^2\right]\right)^2}
		\\&= \frac{\eps^4}{d^2}\ex{\left(2d + 2 \tr{(UOU^\dagger VOV^\dagger) } \right)^2}
		\\&= 4\eps^4+ \frac{4\eps^4}{d}\ex{\tr\left[UOU^\dagger O\right]}+  \frac{4\eps^4}{d^2}\ex{\left(  \tr\left[OU^\dagger O U \right] \right)^2}
		\\&\le 4\eps^4+ \frac{8\eps^2}{d^2}\le 6\eps^2\, .
	\end{align*}
    Here, the third equality again used Haar invariance, while the fourth equality used $\ex{U O U^\dagger} = \frac{\tr[O]}{d}\dI=0$ as well as (relying on \cite[Corollary 13]{mele2023introduction})
    \begin{align*}
        &\ex{\left(  \tr\left[(OU^\dagger O U \right] \right)^2}\\ &= \tr\left[O^{\otimes 2}\left(\ex{(U^\dagger O U)^{\otimes 2}}\right)\right]\\
        &= \tr\left[O^{\otimes 2}\frac{1}{d^2-1}\left( \tr[O^{\otimes 2}]\mathds{1} + F\tr[O^{\otimes 2}F] - \frac{1}{d}F\tr[O^{\otimes 2}] - \frac{1}{d}\tr[F O^{\otimes 2}]\mathds{1}\right)\right]\\
        &= \frac{1}{d^2-1}\left(\left(\tr[O]\right)^4+ \left(\tr[O^2]\right)^2 - \frac{2}{d}\left(\tr[O^2]\right)\left(\tr[O]\right)^2\right)
        \\&=\frac{d^2}{d^2-1}
        \\&\leq 2\, ,
    \end{align*}
    with $F$ the flip operator.
	So by Hölder's inequality:
	\begin{align*}
		\ex{\frac{1}{\sqrt{d}}\|H_{U}-H_{V}\|_2}&\ge\sqrt{\frac{	\ex{\frac{1}{d}\|H_{U}-H_{V}\|_2^2}^3}{	\ex{\frac{1}{d^2}\|H_{U}-H_{V}\|_2^4}}}\ge  \sqrt{\frac{(2\eps^2)^3}{6\eps^4}}> 1.1\eps\, .
	\end{align*}
	Observe that the function $f:(U,V)\mapsto \frac{1}{\sqrt{d}}\|H_{U}-H_{V}\|_2 $ is Lipschitz:
     {\small
	\begin{align*}
	&|f(U,V)-f(U',V')|\\&= \frac{\eps}{\sqrt{d}}\left|\|UOU^\dagger - VOV^\dagger\|_2-\|U'OU'^\dagger - V'OV'^\dagger\|_2 \right|
	\\&\le \frac{\eps}{\sqrt{d}}\left|\|UOU^\dagger - VOV^\dagger-U'OU'^\dagger + V'OV'^\dagger\|_2 \right|
	\\&\le \frac{\eps}{\sqrt{d}}\left|\|UOU^\dagger -UOU'^\dagger\|_2+\|UOU'^\dagger -U'OU'^\dagger\|_2+\|VOV^\dagger-VOV'^\dagger\|_2 +\|VOV'^\dagger - V'OV'^\dagger\|_2 \right|
	\\&\le  \frac{\eps}{\sqrt{d}}\left|\|U^\dagger -U'^\dagger\|_2+\|U -U'\|_2+\|V^\dagger-V'^\dagger\|_2 +\|V- V'\|_2 \right|
	\\&=\frac{2\eps}{\sqrt{d}}\left|\|U -U'\|_2 +\|V- V'\|_2 \right|
	\\&\le \frac{2\sqrt{2}\eps}{\sqrt{d}}\|(U,V) -(U',V')\|_2 \, ,
	\end{align*}}
    where $\|(U,V) -(U',V')\|_2\coloneqq \sqrt{\|U -U'\|_2^2 +\|V- V'\|_2^2}$.
	Hence, by the concentration inequality of Lipschitz functions \edit{w.r.t.~}{with respect to }the Haar measure \cite[Corollary 17]{meckes2013spectral}: 
	\begin{align*}
		\pr{f(U,V)-\ex{f(U,V)}\le  -s}&\le \exp\Big(-\tfrac{ds^2}{12\big(\tfrac{2\sqrt{2}\eps}{\sqrt{d}}\big)^2}\Big) = \exp\left(-\frac{s^2d^2}{96\eps^2}\right)\, .	\end{align*}
	  So, since $\ex{f}\ge 1.1\eps$:
	\begin{align*}
			\pr{f(U,V)\le \eps}	\le \pr{f(U,V)-\ex{f(U,V)}\le -0.1\eps}&\le \exp\left(-\frac{\eps^2d^2}{9600\eps^2}\right)=\exp\left(-\frac{d^2}{9600}\right)\, .
	\end{align*}
		Therefore, by iteratively picking independent Haar-random unitaries, we can construct a family $F= \{H_{x}=\eps U_xOU_x^\dagger \}_{x\in [M]}$ where $M=\exp\left(\frac{d^2}{38400}\right)$ that is $\eps$-separated with high probability. In fact, by the union bound, we have that:
	\begin{align*}
		\pr{F\; \text{is not } \eps\text{-separated}}&= \pr{\exists x\neq x' \in [M] :\frac{1}{\sqrt{d}}\|H_{x}-H_{{x'}}\|_2 < \eps }
		\\&\le M^2\pr{\frac{1}{\sqrt{d}}\|H_{1}-H_{2}\|_2 < \eps }
		\\&\le \exp\left(\frac{d^2}{19200}\right)\cdot \exp\left(-\frac{d^2}{9600}\right)
		\\&\le \exp\left(-\frac{d^2}{19200}\right)\, .
	\end{align*}
	Let $X\sim \unif[M]$ and $Y$ be  the observations of a correct algorithm. By Fano's inequality we have:
	\begin{align*}
		\cI(X: Y)\ge (2/3)\log(M)-\log(2)\ge \frac{d^2}{38400}-\log(2)
	\end{align*}

\paragraph{Coherent algorithms.}
Let $x\in [M]$. A coherent algorithm using the time evolution $t\mapsto \cU_x^t(\cdot)= \mathrm{e}^{-\mathrm{i}t H_x} (\cdot) \mathrm{e}^{\mathrm{i}t H_x}$ chooses the input state $\rho$, the evolution times $t_1, \dots, t_N$ and the channels $\cN_{1}, \dots, \cN_{N-1}$ and measures the output state:
\begin{align*}
\sigma_{x}^{N} = [\cU_{x}^{t_{N}}\otimes \id] \circ \cN_{N-1} \circ \cdots \circ \cN_{1}\circ [\cU_{x}^{t_{1}}\otimes \id](\rho). \end{align*}
Without loss of generality, we can assume that the channels $\cN_{1}, \dots, \cN_{N-1}$ are unitary (Stinespring dilation theorem) and the time evolution unitaries $\{\cU^{t_k}_x\}_k$ act on different systems $\{A_k\}_k$ of dimension $d$ (we can include swap channels in $\cN_1, \dots , \cN_{N-1}$ if necessary). The global system is thus $A_1 \cdots A_N E$ where $E$ is an ancilla system of arbitrary dimension. For $k\in [N]$, we denote $\sigma_{x}^{k} = \cU_{x}^{t_{k}} \circ \cN_{k-1} \circ \cdots \circ \cN_{1}\circ \cU_{x}^{t_{1}}(\rho)$, $\sigma_{x}^0 = \rho$ and $\cN_{0}=\id$ so that we have 
\begin{align*}
\sigma_{x}^{k} = \cU_{x}^{t_{k}} \circ \cN_{k-1}(\sigma_{x}^{k-1}). 
\end{align*}
Denote by $\pi_{k} = \frac{1}{M}\sum_{x=1}^M \cU_{x}^{t_{k}} \circ \cN_{k-1}(\sigma_{x}^{k-1})$ and $\xi_{k} = \frac{1}{M}\sum_{x=1}^M \cN_{k-1}(\sigma_{x}^{k-1})$.

For any bipartite state $\zeta_{AB}$, the von Neumann entropy satisfies the subadditivity property   $ S(AB)_{\zeta}\le S(A)_{\zeta}+S(B)_{\zeta}$ and the triangle inequality $|S(A)_{\zeta} - S(B)_{\zeta}|\le S(AB)_{\zeta}$ \cite{Nielsen2010Dec}, so for all $k\in [n]$, we obtain 

\begin{align}
S\left(A_{1} \cdots A_{k} A_{k+1}\cdots A_{N} E\right)_{\pi_{k}} &\le S\left(A_{1} \cdots A_{k-1} A_{k+1}\cdots A_{N} E\right)_{\pi_{k}} + S(A_{k})_{\pi_{k}}, \text{ and } \label{eq:entropy1}
 \\ - S\left(A_{1} \cdots A_{k-1}A_{k} A_{k+1}\cdots A_{N} E\right)_{\xi_{k}}&\le -S\left(A_{1} \cdots A_{k-1} A_{k+1}\cdots A_{N} E\right)_{\xi_{k}}  +S\left(A_{k} \right)_{\xi_{k}}.\label{eq:entropy2}
\end{align}
Hence the mutual information between $X$ and the observation of the coherent algorithm $Y$ can be bounded as follows
\begin{align*}
    \cI(X:Y)&\overset{(a)}{\le}  S\left(\frac{1}{M}\sum_{x=1}^M \sigma_{x}^N\right)-\frac{1}{M}\sum_{x=1}^MS\left( \sigma_{x}^N\right)
    \\&\overset{(b)}{=} \sum_{k=1}^{N} S\left( \frac{1}{M}\sum_{x=1}^M\sigma_{x}^{k}\right)-\sum_{k=1}^{N}S\left( \frac{1}{M}\sum_{x=1}^M\sigma_{x}^{k-1}\right)
    \\&\overset{(c)}{=}\sum_{k=1}^{N} S\left(\frac{1}{M}\sum_{x=1}^M \cU_{x}^{t_{k}} \circ \cN_{k-1}(\sigma_{x}^{k-1})\right)-\sum_{k=1}^{N}S\left( \frac{1}{M}\sum_{x=1}^M\cN_{k-1}(\sigma_{x}^{k-1})\right)
    \\&=\sum_{k=1}^{N} S\left(A_{1} \cdots A_{k-1}A_{k} A_{k+1}\cdots A_{N} E\right)_{\pi_{k}} -\sum_{k=1}^{N}S\left(A_{1} \cdots A_{k-1}A_{k} A_{k+1}\cdots A_{N} E\right)_{\xi_{k}} 
    \\&\overset{(d)}{\le} \sum_{k=1}^{N} S\left(A_{1} \cdots A_{k-1} A_{k+1}\cdots A_{N} E\right)_{\pi_{k}} +\sum_{k=1}^{N} S(A_{k})_{\pi_{k}}
    \\&\quad  -\sum_{k=1}^{N}S\left(A_{1} \cdots A_{k-1} A_{k+1}\cdots A_{N} E\right)_{\xi_{k}}  + \sum_{k=1}^{N}S\left(A_{k} \right)_{\xi_{k}} 
    \\&\overset{(e)}{=}  \sum_{k=1}^{N} S(A_{k})_{\pi_{k}}+ \sum_{k=1}^{N} S\left(A_{k} \right)_{\xi_{k}}
    \\&\le 2N\log(d)\ , 
\end{align*}
where $(a)$ uses Holevo's theorem \cite{holevo1973bounds}; $(b)$ is a telescopic sum and uses the fact that $S\left( \sigma_{x}^N\right) = S\left( \rho\right)$ for all $x$ as all operations we apply are unitary; $(c)$ uses the assumption that $\cN_{k-1}$ is a unitary channel; $(d)$ uses the inequalities \eqref{eq:entropy1}\&\eqref{eq:entropy2}; and $(e)$ uses the fact  $\ptr{A_{k}}{\pi_{k}} = \ptr{A_{k}}{\xi_{k}}$.

Since $\cI(X:Y)\ge \frac{d^2}{38400}-\log(2)$ we deduce that 
\begin{equation*}
    N \ge \frac{d^2}{76800\log(d)}-\frac{\log(2)}{2\log(d)}. 
\end{equation*}

	\paragraph{Ancilla-free incoherent non-adaptive algorithms.}
	In this setting, the algorithm selects a set of input states $\{\rho_{\ell}\}_{\ell\in [N]}$, time evolutions $\{t_{\ell}\}_{\ell\in [N]}$, and measurement devices $\{\cM_{\ell}\}_{\ell\in [N]}$, which we assume without loss of generality to be of the form $\cM_{\ell}=\{\lambda_{i_{\ell}}\proj{\phi^{\ell}_{i_{\ell}}}\}_{i_{\ell}}$ for projectors $\proj{\phi^{\ell}_{i_{\ell}}}$ and non-negative coefficients $\lambda_{i_{\ell}}$ satisfying $\sum_{i_{\ell}} \lambda_{i_{\ell}}=d$. At step $\ell\in [N]$, the input state is transmitted through the channel $\cU_x^{t_{\ell}}(\rho)= \mathrm{e}^{-\mathrm{i}t_{\ell} H_x} \rho \mathrm{e}^{\mathrm{i}t_{\ell} H_x}$, and the output state is measured using the device $\cM_{\ell}$, resulting in the outcome $I_{\ell}$.

	For a non-adaptive algorithm, the observations $Y=(I_1, \dots, I_N)$ are independent conditioned on $X$, so 
    \begin{align*}
		\cI(X: Y)&=\cI(X: I_1, \dots, I_N) 
       = S(I_1,\dots, I_N ) - S(I_1, \dots, I_N| X)
        \\&\le \sum_{\ell=1}^N S(I_{\ell}) - S( I_{\ell}| X)
       = \sum_{\ell=1}^N \cI(X: I_{\ell}).
	\end{align*}
	Fix $l\in [N]$ and recall the notation $\cU_x^{t_{\ell}}(\rho)= \mathrm{e}^{-\mathrm{i}t_{\ell} H_x} \rho \mathrm{e}^{\mathrm{i}t_{\ell} H_x} $ where $H_x= \eps U_x O U_x^{\dagger}$, we have the joint distribution of $(X, I_{\ell})$:
	\begin{align*}
		q(x, i_{\ell})=\frac{1}{M}\lambda_{i_{\ell}}\bra{\phi_{i_{\ell}}^{\ell}} \cU_x^{t_{\ell}}(\rho_{\ell}) \ket{\phi_{i_{\ell}}^{\ell}}.
	\end{align*}
	So the mutual information is:
	\begin{align*}
		\cI(X:I_{\ell})&=\sum_{x, i_{\ell}}q(x, i_{\ell}) \log\left(\frac{q(x, i_{\ell})}{q(x)q(i_{\ell})}\right)
		\\&=\sum_{x, i_{\ell}}\frac{1}{M}\lambda_{i_{\ell}}\bra{\phi_{i_{\ell}}^{\ell}} \cU_x^{t_{\ell}}(\rho_{\ell}) \ket{\phi_{i_{\ell}}^{\ell}} \log\left(\frac{\frac{1}{M}\lambda_{i_{\ell}}\bra{\phi_{i_{\ell}}^{\ell}} \cU_x^{t_{\ell}}(\rho_{\ell}) \ket{\phi_{i_{\ell}}^{\ell}}}{\frac{1}{M}\cdot\sum_y \frac{1}{M}\lambda_{i_{\ell}}\bra{\phi_{i_{\ell}}^{\ell}} \cU_y^{t_{\ell}}(\rho_{\ell}) \ket{\phi_{i_{\ell}}^{\ell}}}\right)
		\\&=\Sigma_1^{\ell}+ \Sigma_2^{\ell}\, ,
	\end{align*}
	where
	\begin{align*}
		\Sigma_1^{\ell}&=  \sum_{x, i_{\ell}}\frac{1}{M}\lambda_{i_{\ell}}\bra{\phi_{i_{\ell}}^{\ell}} \cU_x^{t_{\ell}}(\rho_{\ell}) \ket{\phi_{i_{\ell}}^{\ell}} \log\left(\frac{ \bra{\phi_{i_{\ell}}^{\ell}} \cU_x^{t_{\ell}}(\rho_{\ell}) \ket{\phi_{i_{\ell}}^{\ell}}}{ \exs{U\sim \Haar(d) }{\bra{\phi_{i_{\ell}}^{\ell}} \cU_U^{t_{\ell}}(\rho_{\ell}) \ket{\phi_{i_{\ell}}^{\ell}}}}\right),
 \\ \Sigma_2^{\ell}&= \sum_{x, i_{\ell}}\frac{1}{M}\lambda_{i_{\ell}}\bra{\phi_{i_{\ell}}^{\ell}} \cU_x^{t_{\ell}}(\rho_{\ell}) \ket{\phi_{i_{\ell}}^{\ell}} \log\left(\frac{\exs{U\sim \Haar(d) }{ \bra{\phi_{i_{\ell}}^{\ell}} \cU_U^{t_{\ell}}(\rho_{\ell})\ket{\phi_{i_{\ell}}^{\ell}}}}{\frac{1}{M}\sum_y  \bra{\phi_{i_{\ell}}^{\ell}} \cU_y^{t_{\ell}}(\rho_{\ell}) \ket{\phi_{i_{\ell}}^{\ell}}}\right), 
 	\end{align*}
 and $\cU_U^{t_{\ell}}(\rho)=\mathrm{e}^{-\mathrm{i}t_{\ell}\eps UOU^\dagger }\rho \mathrm{e}^{\mathrm{i}t_{\ell} \eps UOU^\dagger} $.
 
We have  for $H=\eps UOU^\dagger$, $\mathrm{e}^{itH}=\cos(t\eps)\dI+ \mathrm{i}\sin(t\eps)UOU^\dagger $ hence
{\small
\begin{align}
		\bra{\phi_{i_{\ell}}^{\ell}} \cU_x^{t_{\ell}}(\rho_{\ell}) \ket{\phi_{i_{\ell}}^{\ell}}\nonumber&=	\bra{\phi_{i_{\ell}}^{\ell}} \mathrm{e}^{-\mathrm{i}t_{\ell}H_x}\rho_{\ell}\mathrm{e}^{\mathrm{i}t_{\ell}H_x} \ket{\phi_{i_{\ell}}^{\ell}}\nonumber
  \\&=\cos^2(t_{\ell}\eps)\bra{\phi_{i_{\ell}}^{\ell}} \rho_{\ell} \ket{\phi_{i_{\ell}}^{\ell}} +\sin^2(t_{\ell}\eps)\bra{\phi_{i_{\ell}}^{\ell}}UOU^\dagger \rho_{\ell} UOU^\dagger \ket{\phi_{i_{\ell}}^{\ell}}\nonumber
		\\&\quad +2\Im \cos(t_{\ell}\eps)\sin(t_{\ell}\eps) \bra{\phi_{i_{\ell}}^{\ell}} \rho_{\ell}  UOU^\dagger\ket{\phi_{i_{\ell}}^{\ell}},\nonumber
			\\\ex{\bra{\phi_{i_{\ell}}^{\ell}} \cU_U^{t_{\ell}}(\rho_{\ell}) \ket{\phi_{i_{\ell}}^{\ell}}}\nonumber&=\cos^2(t_{\ell}\eps)\bra{\phi_{i_{\ell}}^{\ell}} \rho_{\ell} \ket{\phi_{i_{\ell}}^{\ell}} +\sin^2(t_{\ell}\eps)\ex{\bra{\phi_{i_{\ell}}^{\ell}}UOU^\dagger \rho_{\ell} UOU^\dagger \ket{\phi_{i_{\ell}}^{\ell}}}\nonumber
			\\&=\cos^2(t_{\ell}\eps)\bra{\phi_{i_{\ell}}^{\ell}} \rho_{\ell} \ket{\phi_{i_{\ell}}^{\ell}} +\sin^2(t_{\ell}\eps)\left(\frac{d}{d^2-1}-\frac{d\bra{\phi_{i_{\ell}}^{\ell}} \rho_{\ell} \ket{\phi_{i_{\ell}}^{\ell}}}{d(d^2-1)}\right)\nonumber
			\\&\ge \cos^2(t_{\ell}\eps)\bra{\phi_{i_{\ell}}^{\ell}} \rho_{\ell} \ket{\phi_{i_{\ell}}^{\ell}} +\frac{\sin^2(t_{\ell}\eps)}{d+1},\nonumber
				\\\ex{\bra{\phi_{i_{\ell}}^{\ell}} \cU_U^{t_{\ell}}(\rho_{\ell}) \ket{\phi_{i_{\ell}}^{\ell}}^2}\nonumber&\le \cos^4(t_{\ell}\eps)\bra{\phi_{i_{\ell}}^{\ell}} \rho_{\ell} \ket{\phi_{i_{\ell}}^{\ell}}^2 +\sin^4(t_{\ell}\eps)\ex{\bra{\phi_{i_{\ell}}^{\ell}}UOU^\dagger \rho_{\ell} UOU^\dagger \ket{\phi_{i_{\ell}}^{\ell}}^2} \nonumber
				\\&\quad +4 \cos^2(t_{\ell}\eps)\sin^2(t_{\ell}\eps) \ex{\bra{\phi_{i_{\ell}}^{\ell}} UOU^\dagger\rho_{\ell}  \ket{\phi_{i_{\ell}}^{\ell}} \bra{\phi_{i_{\ell}}^{\ell}} \rho_{\ell}  UOU^\dagger\ket{\phi_{i_{\ell}}^{\ell}} }\nonumber
				\\&\quad +2 \cos^2(t_{\ell}\eps)\sin^2(t_{\ell}\eps) \ex{\bra{\phi_{i_{\ell}}^{\ell}} \rho_{\ell}  \ket{\phi_{i_{\ell}}^{\ell}} \bra{\phi_{i_{\ell}}^{\ell}} UOU^\dagger\rho_{\ell}  UOU^\dagger\ket{\phi_{i_{\ell}}^{\ell}} }\nonumber
				\\&\quad +4\Im  \cos(t_{\ell}\eps)\sin^3(t_{\ell}\eps) \ex{\bra{\phi_{i_{\ell}}^{\ell}}UOU^\dagger \rho_{\ell} UOU^\dagger \ket{\phi_{i_{\ell}}^{\ell}} \bra{\phi_{i_{\ell}}^{\ell}} \rho_{\ell}  UOU^\dagger\ket{\phi_{i_{\ell}}^{\ell}} }\nonumber
				\\&\le \cos^4(t_{\ell}\eps)\bra{\phi_{i_{\ell}}^{\ell}} \rho_{\ell} \ket{\phi_{i_{\ell}}^{\ell}}^2 +\sin^4(t_{\ell}\eps)\cdot\cO\left(\frac{1}{d^2}\right)\label{eq:Haar-S3} 
				\\&\quad + \cos^2(t_{\ell}\eps)\sin^2(t_{\ell}\eps) \cdot\cO\left(\frac{\bra{\phi_{i_{\ell}}^{\ell}} \rho_{\ell} \ket{\phi_{i_{\ell}}^{\ell}}}{d}\right). \nonumber
\end{align}
}
For $\ex{\bra{\phi_{i_{\ell}}^{\ell}} \cU_U^{t_{\ell}}(\rho_{\ell}) \ket{\phi_{i_{\ell}}^{\ell}}}$, the calculations are similar as in \Cref{eq:Haar-S2}.
To show Inequality \eqref{eq:Haar-S3}, we used Weingarten calculus \Cref{sec:weingarten facts} and the  remark that $\|O\|_{\infty}=1, \tr{(O)}=0, O^2=\dI$ and $|\W(\pi, d)|\le \cO\left( \frac{1}{d^n}\right)$ for $\pi\in \fS_n$ and $n=2,3,4$. 

In particular, we have 
{\small
\begin{align*}
	0&\le  \frac{\bra{\phi_{i_{\ell}}^{\ell}} \cU_x^{t_{\ell}}(\rho_{\ell}) \ket{\phi_{i_{\ell}}^{\ell}}}{\ex{\bra{\phi_{i_{\ell}}^{\ell}} \cU_U^{t_{\ell}}(\rho_{\ell}) \ket{\phi_{i_{\ell}}^{\ell}}}}
	\\&\le \frac{\cos^2(t_{\ell}\eps)\bra{\phi_{i_{\ell}}^{\ell}} \rho_{\ell} \ket{\phi_{i_{\ell}}^{\ell}} +\sin^2(t_{\ell}\eps)\bra{\phi_{i_{\ell}}^{\ell}}UOU^\dagger \rho_{\ell} UOU^\dagger \ket{\phi_{i_{\ell}}^{\ell}}
		+2\Im \cos(t_{\ell}\eps)\sin(t_{\ell}\eps) \bra{\phi_{i_{\ell}}^{\ell}} \rho_{\ell}  UOU^\dagger\ket{\phi_{i_{\ell}}^{\ell}} }{\cos^2(t_{\ell}\eps)\bra{\phi_{i_{\ell}}^{\ell}} \rho_{\ell} \ket{\phi_{i_{\ell}}^{\ell}} +\frac{\sin^2(t_{\ell}\eps)}{d+1}}
	\\&\le 1+ (d+1)+ \sqrt{d+1}\le 2d \, 
\end{align*}
}
for $d\ge 5$. 
Hence 
{\small
\begin{align*}
	&  \bra{\phi_{i_{\ell}}^{\ell}} \cU_x^{t_{\ell}}(\rho_{\ell}) \ket{\phi_{i_{\ell}}^{\ell}}\cdot \left|\frac{\bra{\phi_{i_{\ell}}^{\ell}} \cU_x^{t_{\ell}}(\rho_{\ell}) \ket{\phi_{i_{\ell}}^{\ell}}}{\ex{\bra{\phi_{i_{\ell}}^{\ell}} \cU_U^{t_{\ell}}(\rho_{\ell}) \ket{\phi_{i_{\ell}}^{\ell}}}}-1\right|
	\\&= \frac{\bra{\phi_{i_{\ell}}^{\ell}} \cU_x^{t_{\ell}}(\rho_{\ell}) \ket{\phi_{i_{\ell}}^{\ell}}}{\ex{\bra{\phi_{i_{\ell}}^{\ell}} \cU_U^{t_{\ell}}(\rho_{\ell}) \ket{\phi_{i_{\ell}}^{\ell}}}}\cdot\Bigg|\sin^2(t_{\ell}\eps)\bra{\phi_{i_{\ell}}^{\ell}}UOU^\dagger \rho_{\ell} UOU^\dagger \ket{\phi_{i_{\ell}}^{\ell}}
		\\&\qquad\qquad \qquad \qquad \qquad \qquad  +2\Im \cos(t_{\ell}\eps)\sin(t_{\ell}\eps) \bra{\phi_{i_{\ell}}^{\ell}} \rho_{\ell}  UOU^\dagger\ket{\phi_{i_{\ell}}^{\ell}} -\tfrac{\sin^2(t_{\ell}\eps)}{d+1}\Bigg|
	\\&\le 6dt_{\ell}\eps
\end{align*}
}
Here, we used that $\sin^2(x) \leq x$ and that $\cos(x)\sin(x)\leq x$ for $x \geq 0$. So by taking the average over the outcome $i_{\ell}$ and assuming that $\sum_{\ell=1}^N t_{\ell}\le \frac{d^2}{\eps}$ (since otherwise we already have the statement that we set out to prove) we have
\begin{align*}
   Y_x   &=\sum_{\ell=1}^N\sum_{ i_{\ell}}\lambda_{i_{\ell}}\bra{\phi_{i_{\ell}}^{\ell}} \cU_x^{t_{\ell}}(\rho_{\ell}) \ket{\phi_{i_{\ell}}^{\ell}}\cdot \left(\frac{\bra{\phi_{i_{\ell}}^{\ell}} \cU_x^{t_{\ell}}(\rho_{\ell}) \ket{\phi_{i_{\ell}}^{\ell}}}{\ex{\bra{\phi_{i_{\ell}}^{\ell}} \cU_U^{t_{\ell}}(\rho_{\ell}) \ket{\phi_{i_{\ell}}^{\ell}}}}-1\right)  \in   [-6d^4,6d^4].
\end{align*}
Hence, by Hoeffding's inequality applied to the i.i.d.  random variables $\{Y_x\}_{x\in [M]}$:
\begin{align*}
   & \pr{ \left| \frac{1}{M}\sum_{x}Y_x - \ex{Y_U}\right|>\sqrt{\frac{12^2d^8\log(20)}{M}}}
    \le 2\exp\left(-\tfrac{M\cdot \frac{12^2d^8}{M}\log(20)}{12^2d^8}\right)=\frac{1}{10}.
\end{align*}
Therefore with probability at least $9/10$, for all $l\in [N]$, 
 using the inequality $\log(x)\le x-1$, we have:
 {\small
\begin{align}\label{UB-Sigma1}
&\sum_{\ell=1}^N\Sigma_1^{\ell} -\sqrt{\frac{12^2d^8\log(20)}{M}} \notag
\\&=  \sum_{\ell=1}^N\sum_{x, i_{\ell}}\frac{1}{M}\lambda_{i_{\ell}}\bra{\phi_{i_{\ell}}^{\ell}} \cU_x^{t_{\ell}}(\rho_{\ell}) \ket{\phi_{i_{\ell}}^{\ell}} \log\left(\frac{ \bra{\phi_{i_{\ell}}^{\ell}} \cU_x^{t_{\ell}}(\rho_{\ell}) \ket{\phi_{i_{\ell}}^{\ell}}}{ \ex{\bra{\phi_{i_{\ell}}^{\ell}} \cU_U^{t_{\ell}}(\rho_{\ell}) \ket{\phi_{i_{\ell}}^{\ell}}}}\right)- \sqrt{\frac{12^2d^8\log(20)}{M}}\notag 
		\\&\le \frac{1}{M}\sum_x \sum_{\ell=1}^N\sum_{ i_{\ell}}\frac{1}{M}\lambda_{i_{\ell}}\bra{\phi_{i_{\ell}}^{\ell}} \cU_x^{t_{\ell}}(\rho_{\ell}) \ket{\phi_{i_{\ell}}^{\ell}} \left(\frac{ \bra{\phi_{i_{\ell}}^{\ell}} \cU_x^{t_{\ell}}(\rho_{\ell}) \ket{\phi_{i_{\ell}}^{\ell}}}{ \ex{\bra{\phi_{i_{\ell}}^{\ell}} \cU_U^{t_{\ell}}(\rho_{\ell}) \ket{\phi_{i_{\ell}}^{\ell}}}}-1\right)- \sqrt{\frac{12^2d^8\log(20)}{M}}\notag 
  \\&\le  \exs{V\sim \Haar(d)}{\sum_{\ell=1}^N\sum_{ i_{\ell}}\lambda_{i_{\ell}}\bra{\phi_{i_{\ell}}^{\ell}} \cU_V^{t_{\ell}}(\rho_{\ell}) \ket{\phi_{i_{\ell}}^{\ell}} \left(\frac{ \bra{\phi_{i_{\ell}}^{\ell}} \cU_V^{t_{\ell}}(\rho_{\ell}) \ket{\phi_{i_{\ell}}^{\ell}}}{ \exs{U\sim \Haar(d)}{\bra{\phi_{i_{\ell}}^{\ell}} \cU_U^{t_{\ell}}(\rho_{\ell}) \ket{\phi_{i_{\ell}}^{\ell}}}}-1\right)}
  \notag 
	\\&\le  \sum_{\ell=1}^N\sum_{i_{\ell}} \lambda_{i_{\ell}}\frac{\ex{\bra{\phi_{i_{\ell}}^{\ell}} \cU_U^{t_{\ell}}(\rho_{\ell}) \ket{\phi_{i_{\ell}}^{\ell}}^2}}{\ex{\bra{\phi_{i_{\ell}}^{\ell}} \cU_U^{t_{\ell}}(\rho_{\ell}) \ket{\phi_{i_{\ell}}^{\ell}}}} -1\notag
	\\&\le \sum_{\ell=1}^N\sum_{i_{\ell}} \lambda_{i_{\ell}}\frac{\cos^4(t\eps)\bra{\phi_{i_{\ell}}^{\ell}} \rho_{\ell} \ket{\phi_{i_{\ell}}^{\ell}}^2 +\sin^4(t\eps)\cO\left(\frac{1}{d^2}\right)
		+ \cos^2(t\eps)\sin^2(t\eps) \cO\left(\frac{1}{d}\bra{\phi_{i_{\ell}}^{\ell}} \rho_{\ell} \ket{\phi_{i_{\ell}}^{\ell}}\right)}{\cos^2(t_{\ell}\eps)\bra{\phi_{i_{\ell}}^{\ell}} \rho_{\ell} \ket{\phi_{i_{\ell}}^{\ell}} +\frac{\sin^2(t_{\ell}\eps)}{d+1}} -1\notag
	\\&\le \sum_{\ell=1}^N\sum_{i_{\ell}}\lambda_{i_{\ell}} \cos^4(t_{\ell}\eps) \bra{\phi_{i_{\ell}}^{\ell}} \rho_{\ell} \ket{\phi_{i_{\ell}}^{\ell}} -1
	+   \sum_{ i_{\ell}} \lambda_{i_{\ell}} \sin^2(t_{\ell}\eps)\cO\left(\frac{1}{d}\right)
		+   \sum_{ i_{\ell}}\lambda_{i_{\ell}} \sin^2(t_{\ell}\eps)\cO\left(\frac{1}{d}\right)\notag
		\\&\le \sum_{\ell=1}^N   [\cos^4(t_{\ell}\eps)-1]+\cO\left(\sin^2(t_{\ell}\eps)\right)+ \sqrt{\frac{12^2d^8\log(20)}{M}}
		\le  \sum_{\ell=1}^N \cO( \min\{t_{\ell}^2\eps^2,1\})\notag 
  \\&\le \cO\left( \sum_{\ell=1}^Nt_{\ell}\eps \right) 
\end{align}
}
where we use $\cos(x)\le 1$, $\sin(x)\le \min\{x,1\}$.
On the other hand, for the second sum $\Sigma_2$, we can use again the inequality $\log(x)\le x-1$:

\begin{align*}
	\Sigma_2^{\ell}&=	\sum_{x, i_{\ell}}\frac{1}{M}\lambda_{i_{\ell}}\bra{\phi_{i_{\ell}}^{\ell}} \cU_x^{t_{\ell}}(\rho_{\ell}) \ket{\phi_{i_{\ell}}^{\ell}} \log\left(\frac{\ex{ \bra{\phi_{i_{\ell}}^{\ell}} \cU_U^{t_{\ell}}(\rho_{\ell})\ket{\phi_{i_{\ell}}^{\ell}}}}{\frac{1}{M}\sum_y  \bra{\phi_{i_{\ell}}^{\ell}} \cU_y^{t_{\ell}}(\rho_{\ell}) \ket{\phi_{i_{\ell}}^{\ell}}}\right)
	\\&\le 	\sum_{x, i_{\ell}}\frac{1}{M}\lambda_{i_{\ell}}\bra{\phi_{i_{\ell}}^{\ell}} \cU_x^{t_{\ell}}(\rho_{\ell}) \ket{\phi_{i_{\ell}}^{\ell}} \left(\frac{\ex{ \bra{\phi_{i_{\ell}}^{\ell}} \cU_U^{t_{\ell}}(\rho_{\ell})\ket{\phi_{i_{\ell}}^{\ell}}}}{\frac{1}{M}\sum_y  \bra{\phi_{i_{\ell}}^{\ell}} \cU_y^{t_{\ell}}(\rho_{\ell}) \ket{\phi_{i_{\ell}}^{\ell}}}-1\right)
	\\&= 	\sum_{ i_{\ell}}\lambda_{i_{\ell}}\ex{ \bra{\phi_{i_{\ell}}^{\ell}} \cU_U^{t_{\ell}}(\rho_{\ell})\ket{\phi_{i_{\ell}}^{\ell}}}- 	\sum_{x, i_{\ell}}\frac{1}{M}\lambda_{i_{\ell}}\bra{\phi_{i_{\ell}}^{\ell}} \cU_x^{t_{\ell}}(\rho_{\ell}) \ket{\phi_{i_{\ell}}^{\ell}}
	\\&= \ex{\tr\left(\cU_U^{t_{\ell}}(\rho_{\ell})\right)}- \sum_x \frac{1}{M} \tr\left(\cU_x^{t_{\ell}}(\rho_{\ell})\right)
	\\&=0.
\end{align*}
Therefore,
\begin{align*}
		\frac{d^2}{38400}-\log(2)\le \cI(X: I_1, \dots, I_N)&= \sum_{\ell=1}^N \cI(X: I_{\ell})=\sum_{\ell=1}^N (\Sigma_1^{\ell} + \Sigma_2^{\ell})
  \\&\le \cO\left(\sum_{\ell=1}^N t_{\ell} \eps\right)+ \sqrt{\frac{12^2d^8\log(20)}{M}} 
   \\&\le  \cO\left(\sum_{\ell=1}^N t_{\ell} \eps\right) + 1
	\end{align*}
for $n\ge 11$ since $M= \exp\left(\frac{d^2}{38400}\right)$. Finally:
 \begin{align*}
     \sum_{\ell=1}^N t_{\ell}\ge \Omega\left(\frac{d^2}{\eps}\right). 
 \end{align*}
\end{proof}

\section*{Acknowledgments}

MCC thanks Marcel Hinsche and Marios Ioannou for helpful discussions.
We would like to thank Omar Fawzi and Daniel Stilck Fran{\c c}a for useful discussions.
Moreover, AB would like to thank Mehdi Mhalla for a helpful discussion on stabilizer states. 
Finally, the authors thank several anonymous reviewers for helpful feedback and suggestions that helped us remove the assumption of bounded ancilla for coherent algorithms in Theorem \ref{inf-thm:hardness-hamiltonian-learning}.
AB is supported by the French National Research Agency in the framework of the ``France 2030” program (ANR-11-LABX-0025-01) for the LabEx PERSYVAL. MCC was supported by a DAAD PRIME fellowship. 
AO acknowledges funding by the European Research Council (ERC Grant Agreement No. 948139).

\newpage
\setcounter{secnumdepth}{0}
\defbibheading{head}{\section{References}}
\sloppy
\printbibliography[heading=head]

\newpage
\appendix
\setcounter{secnumdepth}{2}

\section{Relating Hamiltonian distances to time evolution distances}\label{appendix:hamiltonian-distance-vs-time-evollution-distance}

In this appendix, we give operational interpretations for the operator norm distance between Hamiltonians as a worst-case distance and for the normalized Frobenius norm distance between Hamiltonians as an average-case distance. We do so by relating them to corresponding distances between the associated unitary evolutions in the limit of short times.

\subsection{Operator norm distance between Hamiltonians}\label{appendix:operator-norm}

The operator norm distance between two Hamiltonians is connected to the following distance measures on the level of associated time evolution unitaries:
\begin{itemize}
    \item[(a)] Worst-case fidelity between (pure) output states of the unitaries,
    \item[(b)] $1$-to-$1$ norm distance between the unitary channels (i.e., worst-case $1$-norm distance between output states over all input states without auxiliary system),
    \item[(c)] Diamond norm distance between unitary channels (i.e., worst-case $1$-norm distance between output states over all input states with auxiliary system),
    \item[(d)] Operator norm distance between the unitaries up to a global phase, 
    \begin{equation*}
        \operatorname{dist}_\infty(U,V) = \min_{\varphi\in [0,2\pi)}\norm{U - \mathrm{e}^{i\varphi}V}_\infty\, .
    \end{equation*}
\end{itemize}
We first recall the distance measures involved  and the relations between them. Then, we demonstrate their connection to the operator norm distance between Hamiltonians.

First, we know that the $1$-norm distance between two pure states is closely connected to their fidelity via $\|\proj{\phi}-\proj{\psi}\|_1=2\sqrt{1-|\spr{\phi}{\psi}|^2}$, by the equality case in the  Fuchs-van de Graaf inequalities \cite{fuchs1999cryptographic}. This implies the following connection between the above distance measures (a) and (b):
\begin{align*}
    \lVert \mathcal{U} - \mathcal{V}\rVert_{1\to 1}
    = \max_{\ket{\psi}} \lVert U\proj{\psi}U^\dagger - V\proj{\psi}V^\dagger\rVert_1
    = 2\sqrt{1 - \min_{\ket{\psi}}\lvert\bra{\psi}U^\dagger V\ket{\psi}\rvert^2}\, .
\end{align*}
Next, as we're dealing with unitary channels, the $1$-to-$1$ and diamond norm distances coincide, $\lVert \mathcal{U} - \mathcal{V}\rVert_{1\to 1}=\lVert \mathcal{U} - \mathcal{V}\rVert_{\diamond}$ (see \cite[Theorem 3.55]{watrous2018theory}).
Finally, from \cite[Proposition I.6]{haah2023queryoptimal}, we know that for two unitaries $U,V$, the diamond norm distance between the associated channels $\mathcal{U},\mathcal{V}$ and the operator norm distance between the unitaries up to a global phase coincide up to multiplicative constants: 
\begin{equation*}
    \frac{1}{2}\norm{\mathcal{U} - \mathcal{V}}_\diamond
    \leq \operatorname{dist}_\infty (U,V)
    \leq \norm{\mathcal{U} - \mathcal{V}}_\diamond\, .
\end{equation*}
So, all four distance measures (a)-(d) above are equivalent.

To understand them on the level of Hamiltonians, let $H,\Tilde{H}$ be two $n$-qubit Hamiltonians with $\tr[H]=\tr[\Tilde{H}]$. We write the associated unitary time evolutions as $U_t = \mathrm{e}^{-\mathrm{i}tH}$ and $\Tilde{U}_t = \mathrm{e}^{-\mathrm{i}t\Tilde{H}}$. Then, as $t\to 0$, we have for any input state $\ket{\psi}$:
\begin{align*}
    \lvert\bra{\psi}U_t^\dagger \Tilde{U}_t\ket{\psi}\rvert^2
    &= \lvert\bra{\psi} \mathds{1} + \mathrm{i}tH - \frac{t^2}{2}H^2 -\mathrm{i}t\Tilde{H} - \frac{t^2}{2}\Tilde{H}^2 + t^2 H\Tilde{H}\ket{\psi}\rvert^2 + \mathcal{O}(t^3)\\
    &= \left(1 - \frac{t^2}{2}\bra{\psi} (H-\Tilde{H})^2\ket{\psi}\right)^2 + t^2 \left(\bra{\psi} (H-\Tilde{H})\ket{\psi}\right)^2+ \mathcal{O}(t^3)\\
    &= 1 - t^2\bra{\psi} (H-\Tilde{H})^2\ket{\psi} + t^2 \left(\bra{\psi} (H-\Tilde{H})\ket{\psi}\right)^2+ \mathcal{O}(t^3) \, .
\end{align*}
Therefore, we have
\begin{align}
    \lVert \mathcal{U} - \mathcal{V}\rVert_{\diamond}\nonumber
    &=\lVert \mathcal{U} - \mathcal{V}\rVert_{1\to 1}\nonumber\\
    &= 2\sqrt{1 - \min_{\ket{\psi}} \lvert\bra{\psi}U_t^\dagger \Tilde{U}_t\ket{\psi}\rvert^2}\nonumber\\
    &= 2\sqrt{t^2\max_{\ket{\psi}}\left( \bra{\psi} (H-\Tilde{H})^2\ket{\psi} - \left(\bra{\psi} (H-\Tilde{H})\ket{\psi}\right)^2\right) + \mathcal{O}(t^3)}\nonumber\\
    &\begin{cases}\label{eq:case-distinction-bound}
        \leq 2\sqrt{t^2 \max_{\ket{\psi}} \bra{\psi} (H-\Tilde{H})^2\ket{\psi}+ \mathcal{O}(t^3)}\\
        \geq 2\sqrt{t^2\left(\frac{1}{2}(\lambda_{\mathrm{max}}^2 + \lambda_{\mathrm{min}}^2) - \left(\frac{1}{2}(\lambda_{\mathrm{max}} + \lambda_{\mathrm{min}})\right)^2\right)+ \mathcal{O}(t^3)}
    \end{cases}\\
    & \begin{cases}
        = 2\sqrt{t^2 \lVert H - \Tilde{H}\rVert^2_\infty+ \mathcal{O}(t^3)}\\
        = 2\sqrt{t^2\cdot\frac{(\lambda_{\mathrm{max}} - \lambda_{\mathrm{min}})^2}{4} + \mathcal{O}(t^3)}
    \end{cases}\nonumber\\
    &\begin{cases}
        = 2t \lVert H - \Tilde{H}\rVert_\infty + \mathcal{O}(t^{2})\\
        \geq 2\sqrt{t^2\cdot\frac{\max(\lambda_{\mathrm{max}}^2,\lambda_{\mathrm{min}}^2 )}{4} + \mathcal{O}(t^3)}
    \end{cases}\nonumber\\
    & \begin{cases}
        = 2t \lVert H - \Tilde{H}\rVert_\infty + \mathcal{O}(t^{2})\\
        = t \lVert H - \Tilde{H}\rVert_\infty + \mathcal{O}(t^{2})
    \end{cases}\, . \nonumber
\end{align}
Here, we used $\lambda_{\mathrm{max}/\mathrm{min}}$ to denote the maximal and minimal eigenvalues of $H-\tilde{H}$. 
For the lower bound in \Cref{eq:case-distinction-bound}, we made the specific choice $\ket{\psi}=\frac{1}{\sqrt{2}}(\ket{\psi_{\mathrm{max}}}+\ket{\psi_{\mathrm{min}}})$, where $\ket{\psi_{\mathrm{max}/\mathrm{min}}}$ are the (orthogonal) eigenvectors of $H-\tilde{H}$ for the eigenvalues $\lambda_{\mathrm{max}/\mathrm{min}}$.
In the second-to-last step, we used that $\tr[H-\Tilde{H}]=0$ implies that $\lambda_{\mathrm{max}}\geq 0$ and $\lambda_{\mathrm{min}}\leq 0$.
Additionally, in the second-to-last as well as in the last step, we used the Taylor expansion $\sqrt{1+x}=1+\mathcal{O}(x)$ to get $\sqrt{t^2 \lVert H - \Tilde{H}\rVert^2_\infty+ \mathcal{O}(t^3)}= t\lVert H - \Tilde{H}\rVert_\infty\sqrt{1 + \mathcal{O}(t)} = t\lVert H - \Tilde{H}\rVert_\infty \left(1 + \mathcal{O}(t)\right)$ and similarly $\sqrt{t^2\cdot\frac{\lambda_{\mathrm{max}}^2}{4} + \mathcal{O}(t^3)}=t\cdot \frac{\lambda_{\mathrm{max}}}{2}\sqrt{1 + \mathcal{O}(t)}= t\cdot \frac{\lambda_{\mathrm{max}}}{2} \left(1 + \mathcal{O}(t)\right)$.
This derivation tells us that, for short times, we can understand the worst-case distances (a)-(d) via the operator norm distance between the two Hamiltonians underlying the unitary evolutions.

\subsection{Normalized Frobenius norm distance between Hamiltonians}\label{appendix:normalized-frobenius-norm}

The normalized Frobenius norm distance between two Hamiltonians is connected to the following tightly related distance measures on the level of the associated time evolution unitaries:
\begin{itemize}
    \item[(a)] Frobenius norm distance between (normalized) Choi states, 
    \item[(b)] Normalized Frobenius norm distance between unitaries up to a global phase, 
    \item[(c)] Average-case squared Frobenius norm distance over Haar-random input states, 
    \item[(d)] Average-case squared trace norm distance over Haar-random input states, 
    \item[(e)] Average-case fidelity between output states over Haar-random input states.
\end{itemize}
Note that due to our focus on unitary evolutions, pure input states lead to pure output states, so that the relevant average-case squared Frobenius and trace norm distances are related by constant factors. 
Thus, (c) and (d) are immediately related. 
And via the standard relation between the trace distance and the fidelity between pure states, we can immediately translate between (d) and (e).
We now review the connections between (a), (b), and (c) established in prior work, and then relate (a) to the normalized Frobenius norm distance between Hamiltonians.

For two $n$-qubit channels $\mathcal{N}$ and $\mathcal{M}$, we define $D(\mathcal{N},\mathcal{M}) =\frac{1}{\sqrt{2}}\norm{C(\mathcal{N})-C(\mathcal{M})}_2$, where $C(\mathcal{N}) = (\mathcal{N}\otimes \operatorname{id})(\Omega)=(\mathcal{N}\otimes \operatorname{id})(\proj{\Omega})$ denotes the (normalized) Choi state obtained by applying the channel to a canonical maximally entangled state. From \cite[Proposition 15]{bao2023testing}, we know: 
If $\mathcal{N},\mathcal{M}$ are unital, then
\begin{equation*}
    \mathbb{E}_{\ket{\psi}\sim\Haar_n}\left[\norm{\mathcal{N}(\proj{\psi}) - \mathcal{M}(\proj{\psi})}_2^2\right] 
    = \frac{2^{n+1}}{2^{n}+1} D(\mathcal{N},\mathcal{M})^2\, .
\end{equation*}
Moreover, if we define, for two unitaries $U,V$,
\begin{equation*}
    \operatorname{dist}_2(U,V)
    = \frac{1}{\sqrt{2^{n}}}\min_{\varphi\in [0,2\pi)}\norm{U - \mathrm{e}^{i\varphi}V}_2\, ,
\end{equation*}
and if we denote the associated unitary channels by $\mathcal{U}$ and $\mathcal{V}$, then \cite[Lemma 14]{bao2023testing} tells us
\begin{equation*}
    \frac{1}{\sqrt{2}}\operatorname{dist}_2(U,V)
    \leq D(\mathcal{U},\mathcal{V})
    \leq \operatorname{dist}_2(U,V)\, .
\end{equation*}
Now, consider two $n$-qubit Hamiltonians $H,\tilde{H}$ with $\tr[H]=\tr[\Tilde{H}]$, giving rise to unitary time evolutions $U_t = \mathrm{e}^{-\mathrm{i}tH}$ and $\Tilde{U}_t = \mathrm{e}^{-\mathrm{i}t\Tilde{H}}$. Then, as $t\to 0$, we have:
\begin{align*}
    D(\mathcal{U}_t, \Tilde{\mathcal{U}}_t)
    &= \frac{1}{\sqrt{2}}\norm{(U_t\otimes \mathds{1})\Omega (U_t^\dagger \otimes \mathds{1}) - (\Tilde{U}_t \otimes \mathds{1})\Omega (\Tilde{U}_t^\dagger \otimes \mathds{1})}_2\\
    &= \frac{1}{\sqrt{2}}\norm{\Omega - \mathrm{i}t[H\otimes \mathds{1}, \Omega]- \left(\Omega - \mathrm{i}t[\Tilde{H}\otimes \mathds{1}, \Omega]\right)}_2 + \mathcal{O}(t^2)\\
    &= \frac{t}{\sqrt{2}} \norm{[(H - \Tilde{H})\otimes \mathds{1},\Omega]}_2 + \mathcal{O}(t^2)\\
    &= \frac{t}{\sqrt{2^n}} \norm{H - \Tilde{H}}_2 + \mathcal{O}(t^2)\, .
\end{align*}
Thus, in the limit of short times, $\frac{1}{\sqrt{2^n}} \norm{H - \Tilde{H}}_2$ closely relates to $D(\mathcal{U}_t, \Tilde{\mathcal{U}}_t)$ and thereby also to $\operatorname{dist}(U_t,\Tilde{U}_t)$ as well as to $\sqrt{\mathbb{E}_{\ket{\psi}\sim\Haar_n}\left[\norm{U_t \proj{\psi} U_t - \Tilde{U}_t\proj{\psi}\Tilde{U}_t}_2^2\right]}$.
This tells us that for short times, the different average-case distance measures (a)-(e) are connected to the normalized Frobenius distance between the underlying Hamiltonians.

\section{More general variants of Hamiltonian property testing}

In this appendix, we discuss three variants and extensions of the Hamiltonian property testing procedure presented in \Cref{sec:randomized-upper-bounds}. First, we demonstrate how a small modification to our original procedure allows us to test many properties simultaneously. Second, we show that we can test arbitrarily large properties when allowing the testing procedure to use auxiliary qubits. Third, we present a variant of the procedure that can be used for tolerant testing.

\subsection{Testing many properties}\label{testing many ppts}

\SetKwComment{Comment}{/* }{ */}
\SetKwInOut{Input}{Input}
        \SetKwInOut{Output}{Output}
\begin{algorithm}[t!]
\caption{Testing  Multiple Properties for Hamiltonian Evolutions }\label{alg:multiple property testing Hamiltonian}
\LinesNumbered
\Input{A Hamiltonian $H$, a set of properties $\Pi_{S_m}$ for $m \in \{1, \ldots, M\}$, and an accuracy parameter $\varepsilon\in (0,1)$}
\Output{The null hypothesis $H^m_0$ or the alternate hypothesis $H^m_1$ for each $m \in \{1, \ldots, M\}$}
$t \gets \frac{\eps}{6}$\;
$N \gets  \left\lceil\frac{100\log(M/\delta)}{t^2\eps^2}\right\rceil$\;

    \For{$s\gets 1$ \KwTo $N$}{
    Sample $i_s\sim\unif[d]$, $j_s\sim \unif[d+1]$\;
     Input state : $ \rho_s= \proj{\phi_{i_s,j_s}}$\;
     Evolve under $H$ for time $t$\;
     Measurement : $\cM_s = \{\proj{\phi_{i_s,\ell}}\}_{\ell}$ and observe $\ell_s\gets \cM_s(\cU_t(\rho_s))$\;
    }
    \For{$m \gets 1$ \KwTo $M$}{
     \eIf{$\frac{1}{N}\sum_{s=1}^N \mathbf{1}(\{\ket{\phi_{i_s,\ell_s}} \nsim_{S_m} \ket{\phi_{i_s,j_s}}\})\le \frac{3}{8}t^2\eps^2$}
        {
            \Return $H^m_0$
        }
     {
            \Return $H^m_1$
        }
}

\end{algorithm}

If we want to test many properties $S_1, \dots , S_M$ with the same data acquired, we need to change the algorithm decision rule. The new decision rule is 
\begin{center}
    At step $x=1,\dots, N$, let $\cE_x=\{ \ket{\phi_{i_x,\ell_x}} \sim_S \ket{\phi_{i_x,j_x}}\}$ and let $\bar{\cE}_x$ be its complement. Answer the null hypothesis iff
	\begin{align*}
		\frac{1}{N}\sum_{x=1}^N \mathbf{1}(\{\bar{\cE}_x\})\le \frac{3}{8}t^2\eps^2.
	\end{align*}
\end{center}
As we will see below, the error probability can be shown to be at most $\delta$ by Chernoff-Hoeffding's inequality~\cite{hoeffding_probability_1963} for the following set of parameters:

    \begin{itemize}
        \item Evolution time at each step $t=\frac{\eps}{6}$,
    	\item Number of independent experiments $N= \left\lceil\frac{100\log(1/\delta)}{t^2\eps^2}\right\rceil= \left\lceil\frac{60^2\log(1/\delta)}{\eps^4}\right\rceil$,
    	\item Total evolution time $Nt = \left\lceil\frac{60^2\log(1/\delta)}{\eps^4}\right\rceil\cdot \frac{\eps}{6} \leq  \frac{600 \log(1/\delta)}{\eps^3} + \underbrace{\frac{\eps}{6}}_{\leq 1/6}$.
    \end{itemize}

So, to test $M$ properties we can take the error probability to be $\delta\rightarrow \frac{\delta}{M}$ and apply a union bound. The new complexity can be taken to be 

\begin{itemize}
        \item Evolution time at each step $t=\frac{\eps}{6}$,
    	\item Number of independent experiments $N= \left\lceil\frac{100\log(M/\delta)}{t^2\eps^2}\right\rceil= \left\lceil\frac{60^2\log(M/\delta)}{\eps^4}\right\rceil$,
    	\item Total evolution time
        \begin{align*}
            Nt 
            &= \left\lceil\frac{60^2\log(M/\delta)}{\eps^4}\right\rceil\cdot \frac{\eps}{6} \leq  \frac{600\log(M/\delta)}{\eps^3} + \underbrace{\frac{\eps}{6}}_{\leq 1/6}\, .
        \end{align*}
    \end{itemize}
The logarithmic scaling with $M$ and $1/\delta$ is good, however, the dependency in $\eps$ might be sub-optimal.

With multiple properties, there is a null and alternate hypothesis for each property. The null hypothesis $H_0^i$ at $i \in \{1, \ldots, M\}$ becomes that $H \in \Pi_{S_i}$ and the alternate hypothesis $H_1^i$ is that $H$ is $\eps$ far from being in $\Pi_{S_i}$. The algorithm for the testing of multiple properties is therefore \Cref{alg:multiple property testing Hamiltonian}.

\begin{theorem}\label{thm:upper bound on multiple property testing} Let $S_i\subset \mathds{P}_n$ such that  $|S_i\cup\{\dI\}|\le\frac{(2^n+1)\eps^{4}}{144}$ for all $i \in \{1, \ldots, M\}$. 
    Suppose that the Hamiltonian $H$  satisfies $\tr(H)=0 $ and $\|H\|_\infty \le 1$. For each $i \in \{1, \ldots, M\}$,
    \Cref{alg:multiple property testing Hamiltonian} tests whether $H\in \Pi_{S_i}$ or whether $\frac{1}{\sqrt{2^n}}\|H-K\|_2>\eps$ for all  Hamiltonians $K\in \Pi_{S_i}$. The algorithm succeeds with probability at least $1-\delta$ and uses a total evolution time $\mathcal{O}\left(\frac{\log(M/\delta)}{\eps^3}\right)$, a number of independent experiments $N= \mathcal{O}\left(\frac{\log(M/\delta)}{\eps^4}\right)$ , and a total classical processing time $\tilde{\cO}\left(n^2 \sum_{i=1}^M \frac{|S_i\cup\{\dI\}|}{\eps^4}\right)$. Each experiment uses efficiently implementable states and measurements.
\end{theorem}

As in \Cref{thm:upper bound on testing}, we in fact have stronger guarantees than just efficient implementability for the states and measurements. Also here, the input states are $n$-qubit stabilizer states and the output measurements are \edit{w.r.t.~}{with respect to }some $n$-qubit stabilizer state bases.

\begin{proof}[Proof Sketch]
    From the proof of correctness in \Cref{sec:randomized-upper-bounds} we have the following inequalities, conditioned on either the null or alternate hypothesis being true:
    \begin{align*}
    	\mathbb{P}_{H_1}(\cE)=\mathbb{P}_{H_1}(\ket{\phi_{i,\ell}} \sim_S \ket{\phi_{i,j}}) &\le 1-\frac{t^2\eps^2}{2}\, ,
       \\ \mathbb{P}_{H_0}(\bar{\cE})=\mathbb{P}_{H_0}(\ket{\phi_{i,\ell}} \nsim_S \ket{\phi_{i,j}} ) &\le t^4 \le \frac{t^2\eps^2}{4}\, ,
    \end{align*}
    for $t=\frac{\eps}{6}$ and $|S\cup\{\dI\}|\le\frac{(d+1)^{}\eps^{4}}{144}$. Now we have a linear behavior of the KL divergence: 
    \begin{align*}
        \KL\left(\frac{3}{8}t^2\eps^2\Big\| \frac{1}{2}t^2\eps^2\right),\;\KL\left(\frac{3}{8}t^2\eps^2\Big\| \frac{1}{4}t^2\eps^2\right)\ge \frac{1}{100}\cdot t^2\eps^2.
    \end{align*}
    Here, we have used that 
    \begin{align*}
        \KL(p\|\alpha p) &\geq (-\log{\alpha}+\alpha - 1) p\, , \\
        \KL(\alpha p\|p) &\geq (\alpha \log{\alpha}+ 1-\alpha) p\, ,
    \end{align*}
    for $\alpha \in (0,1)$ and $p \in [0,1]$, which can be inferred easily from $\log(1+x) \geq x/(1+x)$ for $x > - 1$. Under the null hypothesis, since $\frac{3}{8}t^2\eps^2> \frac{1}{4}t^2\eps^2\ge  \prs{H_0}{	\bar{\cE}}$, we can apply the Chernoff-Hoeffding inequality (see \cite{hoeffding_probability_1963} or \cite[Theorem 2.1]{mulzer2018five}):
    \begin{align*}
    	\prs{H_0}{	\frac{1}{N}\sum_{x=1}^N \mathbf{1}(\{\bar{\cE}_x\})> \frac{3}{8}t^2\eps^2}&\le \exp\left(-N\KL\left(\frac{3}{8}t^2\eps^2\Big\| \prs{H_0}{	\bar{\cE}}\right)\right)
    	\\&\le \exp\left(-N\KL\left(\frac{3}{8}t^2\eps^2\Big\| \frac{1}{4}t^2\eps^2\right)\right)
    	\\&\le \exp\left(-N\cdot \tfrac{1}{100}\cdot t^2\eps^2\right)
        \le \delta
    \end{align*}
    for $N=\left\lceil\frac{100\log(1/\delta)}{t^2\eps^2}\right\rceil$. Here, we have used in the second inequality that $\alpha \mapsto \KL(p\|\alpha p)$ is decreasing on $(0,1)$ for any $p \in [0,1]$, which can be verified by differentiating in $\alpha$. 
    
    The alternate hypothesis case is similar. 
    The bounds on the classical processing time and on the quantum circuit sizes required for state preparation and measurements follow as in the proof of \Cref{thm:upper bound on testing}.
\end{proof}

\begin{remark}\label{remark:application-sparse}
    We highlight two applications of \Cref{thm:upper bound on multiple property testing}.
    The first is Hamiltonian sparsity testing, that is, the task of testing whether an unknown Hamiltonian $H$ has an at most $k$-sparse Pauli basis expansion or is $\varepsilon$-far \edit{w.r.t.~}{with respect to }$\frac{1}{\sqrt{d}}\|\cdot\|_2$ from all such Hamiltonians, where $k=\mathcal{O}(1)$.
    This can be embedded into the scenario of \Cref{thm:upper bound on multiple property testing} by testing $M=\binom{4^n-1}{k}=\mathcal{O}(4^{nk})$ many properties of size $k$ simultaneously, each corresponding to a possible set of at most $k$ Pauli terms that appear with non-zero coefficients. As our bounds in \Cref{thm:upper bound on multiple property testing} scale logarithmically with the number of properties, \Cref{alg:multiple property testing Hamiltonian} solves this sparsity testing problem with an efficient number of queries and an efficient total total evolution time.
    If $\frac{1}{\sqrt{2^n}}\norm{H}_2=1$, meaning $\sum_{P\in\mathbb{P}_n} \lvert\alpha_P\rvert^2 = 1$, then \Cref{alg:multiple property testing Hamiltonian} can even be used to estimate the support of $H$. Namely, in this case, each of the $M$ properties that gets assigned the corresponding null hypothesis in the second for-loop corresponds to a Pauli support such that the corresponding coefficients of $H$ have a squared $\ell_2$-norm of $\geq 1-\varepsilon^2$.

    Second, we can use \Cref{thm:upper bound on multiple property testing} to test whether an unknown $H$ is a low-intersection Hamiltonian. This property served as an important assumption in recent work on Hamiltonian learning, for instance in \cite{haah2022optimal, huang2023heisenberg}.
    We call a Hamiltonian $H=\sum_{P\in\mathbb{P}_n} \alpha_P P$ a $(k,\mathfrak{d})$-intersection Hamiltonian if $|P|\leq k$ for all $P\in\mathbb{P}_n$ with $\alpha_P\neq 0$ and if $\lvert\{Q\in\mathbb{P}_n~|~\alpha_Q\neq 0~\wedge~ \mathrm{supp}(P)\cap\mathrm{supp}(Q)\neq\emptyset\}\rvert\leq \mathfrak{d}$ for all $P\in \mathbb{P}_n$ with $\alpha_P\neq 0$. Here, $\mathrm{supp}(P)=\{1\leq i\leq n~|~ P_i\neq \dI\}$ denotes the support of an $n$-qubit Pauli string.
    We speak of a low-intersection Hamiltonian if both $k,\mathfrak{d}=\mathcal{O}(1)$.
    Given fixed $k$ and $\mathfrak{d}$, there are at most $2^{\mathcal{O}(n^{\mathrm{poly}(k,\mathfrak{d})})}$ many different dual interaction graphs (see \cite {haah2022optimal} for a definition) that a $(k,\mathfrak{d})$-intersection Hamiltonian can have. (A loose bound that does not take $\mathfrak{d}$ into account can be seen as follows: By the locality constraint, admissible dual interaction graphs have at most $\sum_{\ell=0}^k \binom{n}{\ell} 3^\ell\leq \mathcal{O}(n^{k+1})$ vertices and thus at most $\mathcal{O}(n^{2(k+1)})$ edges. As each edge can either be present or not, there are at most $2^{\mathcal{O}(n^{2(k+1)})}$ many admissible dual interaction graphs.) Therefore, by simultaneously testing all the size-$\mathcal{O}(n^{k+1})$ properties corresponding to different valid dual interaction graphs, \Cref{alg:multiple property testing Hamiltonian} can test whether an unknown $H$ is a $(k,\mathfrak{d})$-intersection Hamiltonian; and thanks to \Cref{thm:upper bound on multiple property testing}, the query complexity and total evolution time for doing so can be bounded in terms of $\log 2^{\mathcal{O}(n^{\mathrm{poly}(k,\mathfrak{d})})}= \mathcal{O}(n^{\mathrm{poly}(k,\mathfrak{d})})$.
\end{remark}

\subsection{Testing arbitrarily large properties}\label{assumption on S}

\SetKwComment{Comment}{/* }{ */}
\SetKwInOut{Input}{Input}
        \SetKwInOut{Output}{Output}
\begin{algorithm}[t!]
\caption{Testing Arbitrary Properties for Hamiltonian Evolutions}\label{alg:property testing Hamiltonian big S}
\LinesNumbered
\Input{A Hamiltonian $H$, a property $\Pi_S$, and an accuracy parameter $\varepsilon\in (0,1)$}
\Output{The null hypothesis $H_0$ or the alternate hypothesis $H_1$}
$n_{\mathrm{aux}}\gets \left\lceil\log_2\left(\frac{144\cdot  |S \cup \{\dI\}|}{2^n\eps^4}\right)\right\rceil $ \;
$S_{n_{\mathrm{aux}}}\gets (S\cup\{\dI\})\otimes \dI_{2^{n_{\mathrm{aux}}}}$\;
$t \gets \frac{\eps}{6}$\;
$N \gets  \left\lceil\frac{2\log(3)}{t^2\eps^2}\right\rceil$\;
\For{$s\gets 1$ \KwTo $N$}{
Sample $i_s\sim\unif[2^{n+{n_{\mathrm{aux}}}}]$, $j_s\sim \unif[2^{n+{n_{\mathrm{aux}}}}+1]$\;
 Input state : $ \rho_s= \proj{\phi_{i_s,j_s}}$\;
 Evolve the first $n$ qubits under $H$ for time $t$\;
 Measurement : $\cM_s = \{\proj{\phi_{i_s,\ell}}\}_{\ell}$ and observe $\ell_s\gets \cM_s\left[(\cU_t\otimes \id_{n_{\mathrm{aux}}})(\rho_s)\right]$\;
  \If{$\ket{\phi_{i_s,j_s}} \nsim_{S_{n_{\mathrm{aux}}}} \ket{\phi_{i_s,\ell_s}}$}
    {
        \Return $H_1$ and \textbf{stop}
    }
}
\Return $H_0$
\end{algorithm}

In \Cref{thm:upper bound on testing}, we can test a property $\Pi_S$ if the size of the set $S$ satisfies:
\[|S\cup\{\dI\}|\le\frac{(d+1)^{}\eps^{4}}{144}.\]
To lift this assumption, we propose to use  auxiliary systems. The idea is that if we can add an $n_{\mathrm{aux}}$-qubit ancilla, then we can query the unitary
\begin{align*}
	\mathrm{e}^{\mathrm{i}tH}\otimes \dI_{2^{n_{\mathrm{aux}}}}= 	\mathrm{e}^{\mathrm{i}tH\otimes \dI_{2^{n_{\mathrm{aux}}}} }. 
\end{align*}
So we can think of testing $H\otimes \dI_{2^{n_{\mathrm{aux}}}}$ instead of $H$, and the relevant property now is $S_{n_{\mathrm{aux}}}=(S\cup\{\dI\})\otimes \dI_{2^{n_{\mathrm{aux}}}}$. This is not exactly true, because we cannot enforce that $H\otimes \dI_{2^{n_{\mathrm{aux}}}} $ is $\eps$-far from $\Pi_{S_{n_{\mathrm{aux}}}} $ given that $H $ is $\eps$-far from $\Pi_{S} $. 
However, in terms of distance, we only need  that $\sum_{P\notin S} \alpha_P^2\ge \eps^2$ for our proof. To see this, let us bound $\pr{\cE}$ under both the null and alternate hypothesis, where we write again $\cE=\{ \ket{\phi_{i,\ell}} \sim_S \ket{\phi_{i,j}}\}$ and $\bar{\cE}$ for its complement.

First, under the null hypothesis, the main ingredient in the proof in \Cref{sec:randomized-upper-bounds} is the observation that if $H\in \Pi_S$ and $\ket{\phi_{i,\ell}}\nsim_{S} \ket{\phi_{i,j}}$, then for all $j\neq l$ we have $\bra{\phi_{i,\ell}}H^m \ket{\phi_{i,j}}=0$ for $m=0,1$. This is not affected if we add ancilla as
$\ket{\phi_{i,\ell}}\nsim_{S_{n_{\mathrm{aux}}}} \ket{\phi_{i,j}} \Rightarrow \forall P\in S_{n_{\mathrm{aux}}} :|\bra{\phi_{i,\ell}}P \ket{\phi_{i,j}}|\neq 1  \Rightarrow \forall P\in S\cup\{\dI\} :|\bra{\phi_{i,\ell}}P\otimes \dI_{2^{n_{\mathrm{aux}}}} \ket{\phi_{i,j}}|\neq 1$, thus $\bra{\phi_{i,\ell}}(H\otimes \dI_{2^{n_{\mathrm{aux}}}})^m \ket{\phi_{i,j}} = \bra{\phi_{i,\ell}}H^m\otimes \dI_{2^{n_{\mathrm{aux}}}} \ket{\phi_{i,j}}=0$ for $m=0,1$. Hence, under the condition that the null hypothesis is true, 
\begin{align*}
\prs{H_0}{\bar{\cE}} &\le t^4\, ,
\end{align*}
by the same reasoning as before.
Next, under the alternate hypothesis, we have (by the previous calculations): 
\begin{align*}
		\prs{H_1}{\cE} &\le  \sum_{P\in  S_{n_{\mathrm{aux}}} \cup \{\dI\}} \frac{2^{n+{n_{\mathrm{aux}}}}}{2^{n+{n_{\mathrm{aux}}}}(2^{n+{n_{\mathrm{aux}}}}+1)}  \\ & \qquad+\sum_{P\in  S_{n_{\mathrm{aux}}} \cup \{\dI\}} \frac{1}{2^{n+{n_{\mathrm{aux}}}}(2^{n+{n_{\mathrm{aux}}}}+1)} \left|\sum_{m\ge 0} \frac{(\mathrm{i}t)^m}{m!} \tr(P(H\otimes \dI_{2^{n_{\mathrm{aux}}}})^m) \right|^2
		\\&=   \frac{|S \cup \{\dI\}|}{(2^{n+{n_{\mathrm{aux}}}}+1)}+  \sum_{P\in  S\cup \{\dI\}} \frac{1}{2^{n+{n_{\mathrm{aux}}}}(2^{n+{n_{\mathrm{aux}}}}+1)} \left|\sum_{m\ge 0} \frac{(\mathrm{i}t)^m}{m!} \tr(PH^m)2^{n_{\mathrm{aux}}} \right|^2
		\\&\le    \frac{|S \cup \{\dI\}|}{(2^{n+{n_{\mathrm{aux}}}}+1)}+  \sum_{P\in  S \cup \{\dI\}} \frac{1}{2^{2n}} \left|\sum_{m\ge 0} \frac{(\mathrm{i}t)^m}{m!} \tr(PH^m) \right|^2\, ,
\end{align*}
since $P\in S\Leftrightarrow P\otimes \dI_{2^{n_{\mathrm{aux}}}}\in S_{n_{\mathrm{aux}}}$. Then we can continue as in the proof of \Cref{thm:upper bound on testing} to have 
\begin{align*}
	\prs{H_1}{\cE} &\le    \frac{|S\cup \{\dI\}|}{(2^{n+{n_{\mathrm{aux}}}}+1)}+  \sum_{P\in  S \cup \{\dI\}} \frac{1}{2^{2n}} \left|\sum_{m\ge 0} \frac{(\mathrm{i}t)^m}{m!} \tr(PH^m) \right|^2
	\\&\le  1-t^2\eps^2+ \frac{|S\cup \{\dI\}|}{(2^{n+{n_{\mathrm{aux}}}}+1)}+ 7t^4\, .
\end{align*}
Hence, if we choose $t= \frac{\eps}{6}$, then we have $\prs{H_0}{\bar{\cE}} \le \frac{t^2\eps^2}{4} $. Thus, we have to fulfill the following in order to ensure $\prs{H_1}{\cE} \le1-\frac{t^2\eps^2}{2}$:
\begin{align*}
	\frac{|S\cup \{\dI\}|}{2^{n+{n_{\mathrm{aux}}}}}\le \frac{t^2\eps^2}{4}= \frac{\eps^4}{144} \quad \text{and} \quad 7t^4\le  \frac{t^2\eps^2}{4} \, .
\end{align*}
The second inequality holds with our choice of $t$ as $|S \cup \{\dI\}|\ge 1$, and the first inequality holds if 
\begin{align*}
	n_{\mathrm{aux}}\ge   \log_2\left(\frac{144\cdot  |S \cup \{\dI\}|}{2^n\eps^4}\right) \, .
\end{align*}
As we have always $|S \cup \{\dI\}|\le 2^{2n}$, a number of auxiliary qubits 
\begin{align*}
    n_{\mathrm{aux}} \ge  \log_2\left(\frac{144\cdot  2^{n}}{\eps^4}\right)= n + \log_2\left(\frac{144}{\eps^4}\right)
\end{align*}
is enough to ensure $\prs{H_1}{\cE} \le1-\frac{t^2\eps^2}{2}$.
With this amount of ancilla qubits, we can test \emph{any} set $S$, regardless of its size. 

Since we have the inequalities 
\begin{align*}
	\mathbb{P}_{H_1}(\cE) \le 1-\frac{t^2\eps^2}{2}\quad \text{and} \quad 
 \mathbb{P}_{H_0}(\bar{\cE}) \le \frac{t^2\eps^2}{4} \quad \text{for} \quad t=\frac{\eps}{6}\, ,
\end{align*}
the query complexity and the total evolution time are the same as for the ancilla-free algorithm. 
The classical processing time changes compared to \Cref{thm:upper bound on testing} in that we now need to check properties of $(n+{n_{\mathrm{aux}}})$-qubit Paulis. 
With the same arguments as in the proof of \Cref{thm:upper bound on testing}, we see that the relation $\sim_{S_{n_{\mathrm{aux}}}}$ can be checked with $\mathcal{O}((n+{n_{\mathrm{aux}}})^2 \lvert S_{n_{\mathrm{aux}}}\rvert) = \Tilde{\cO}(n^2 |S\cup\{\dI\}|)$ classical processing time. Here, we used $\Tilde{\cO}$ as a notational simplification that hides a factor $\mathrm{polylog}(1/\varepsilon)$.
Finally, the input states are now $(n+{n_{\mathrm{aux}}})$-qubit stabilizer states and the measurements are in some bases of $(n+{n_{\mathrm{aux}}})$-qubit stabilizer states, so the quantum circuit sizes can be bounded as before upon replacing $n$ by $n+{n_{\mathrm{aux}}}$, leading to quantum circuit complexities of $\cO\left(\frac{(n+{n_{\mathrm{aux}}})^2}{\log(n+{n_{\mathrm{aux}}})}\right) = \Tilde{\cO}\left(\frac{n^2}{\log n}\right)$, where the $\Tilde{\cO}$ again hides factors that are polylogarithmic in $\frac{1}{\varepsilon}$.

Thus, we have shown the following theorem in this section:

\begin{theorem}\label{thm:upper bound on testing big S} 
    Let $S\subset \mathds{P}_n$. Suppose that the Hamiltonian $H$  satisfies $\tr(H)=0 $ and $\|H\|_\infty \le 1$. Using $n_{\mathrm{aux}}=  \left\lceil\log_2\left(\frac{144\cdot  |S \cup \{\dI\}|}{2^n\eps^4}\right)\right\rceil$ ancilla qubits, 
    \Cref{alg:property testing Hamiltonian big S} tests whether $H\in \Pi_S$  or $\frac{1}{\sqrt{2^n}}\|H-K\|_2>\eps$ for all  Hamiltonians $K\in \Pi_S$ with probability at least $2/3$ using a total evolution time $\mathcal{O}\left(\frac{1}{\eps^3}\right)$, a total number of independent experiments $N=\mathcal{O}\left(\frac{1}{\eps^4}\right)$, and a total classical processing time $\tilde{\cO}\left(\frac{n^2 |S\cup\{\dI\}|}{\eps^4}\right)$. Each experiment uses efficiently implementable states and measurements on $\cO (n + \log(1/\varepsilon))$ many qubits.
\end{theorem}

\subsection{Tolerant Hamiltonian property testing}\label{sec:tolerant-testing}

So far, we have focused on the standard setting of property testing. Now, we turn our attention to tolerant testing \cite{parnas2006tolerant}. In our Hamiltonian case, we formulate a tolerant property testing problem as follows:

\begin{problem}[Tolerant Hamiltonian property testing]\label{def:tolerant-hamiltonian-testing}
    Given a property $\Pi_S$ associated to a subset $S\subseteq \mathds{P}_n$, a norm $\nnorm{\cdot}$, and two accuracy parameters $0\leq \varepsilon_1<\varepsilon_2<1$, we denote by $\mathcal{T}_{\nnorm{\cdot}}^{\Pi_S}(\varepsilon_1,\varepsilon_2)$ the following Hamiltonian property testing problem:
    Given access to the time evolution according to an unknown Hamiltonian $H$, decide, with success probability $\geq 2/3$, whether
    \begin{enumerate}
        \item[(i)] $H$ is $\varepsilon_1$-close to having property $\Pi_S$, that is, there exists $\tilde{H}\in \Pi_S$ such that $\nnorm{ H - \Tilde{H}}\leq\varepsilon_1$, or 
        \item[(ii)] $H$ is $\varepsilon_2$-far from having property $\Pi_S$, that is, $\forall \Tilde{H}\in \Pi_S$: $\nnorm{ H - \Tilde{H}}\geq \varepsilon_2$\, .
    \end{enumerate}
     If $H$ satisfies neither (i) nor (ii), then any output of the tester is considered valid. 
\end{problem}

Clearly, the tolerant testing problem $\mathcal{T}_{\nnorm{\cdot}}^{\Pi_S}(\varepsilon_1,\varepsilon_2)$ is at least as hard as $\mathcal{T}_{\nnorm{\cdot}}^{\Pi_S}(\varepsilon_2)$, so the lower bounds of \Cref{sec:hamiltonian-testing-lower-bounds} for locality testing \edit{w.r.t.~}{with respect to }unnormalized Schatten $p$-norms carry over straightforwardly. More interestingly, in this section, we show, via a variant of the analysis from \Cref{sec:randomized-upper-bounds}, that an analogue of \Cref{thm:upper bound on testing} holds for \Cref{alg:tolerant property testing Hamiltonian}.

\SetKwComment{Comment}{/* }{ */}
\SetKwInOut{Input}{Input}
        \SetKwInOut{Output}{Output}
\begin{algorithm}[t!]
\caption{Tolerantly Testing Properties for Hamiltonian Evolutions}\label{alg:tolerant property testing Hamiltonian}
\LinesNumbered
\Input{A Hamiltonian $H$, a property $\Pi_S$, and accuracy parameters $0\leq \varepsilon_1<\varepsilon_2<1$}
\Output{The null hypothesis $H_0$ or the alternate hypothesis $H_1$}
$t \gets \sqrt{\frac{1}{20}(\eps_2^2-\eps_1^2)}$\;
$N \gets \left\lceil\frac{30\log(3) (\eps_2^2+\eps_1^2)}{t^2(\eps_2^2-\eps_1^2)^2}\right\rceil  $\;
\For{$s\gets 1$ \KwTo $N$}{
Sample $i_s\sim\unif[d]$, $j_s\sim \unif[d+1]$\;
 Input state : $ \rho_s= \proj{\phi_{i_s,j_s}}$\;
 Evolve under $H$ for time $t$\;
 Measurement : $\cM_s = \{\proj{\phi_{i_s,\ell}}\}_{\ell}$ and observe $\ell_s\gets \cM_s(\cU_t(\rho_s))$\;
}
 \eIf{$\frac{1}{N}\sum_{s=1}^N \mathbf{1}(\{\ket{\phi_{i_s,\ell_s}} \nsim_{S} \ket{\phi_{i_s,j_s}}\})\le  \frac{1}{5}t^2(2\eps_2^2+3\eps_1^2)$}
        {
            \Return $H_0$
        }
     {
            \Return $H_1$
        }
\end{algorithm}

\begin{theorem}\label{thm:upper bound on tolerant testing} 
    Let $0\leq \varepsilon_1<\varepsilon_2<1$, and let $S\subset \mathds{P}_n$ such that  $|S\cup\{\dI\}|\le\frac{(2^n+1)^{} (\eps_2^2-\eps_1^2)^{2}}{400}$.
    Suppose that the Hamiltonian $H$ satisfies $\tr(H)=0 $ and $\|H\|_\infty \le 1$. 
    \Cref{alg:tolerant property testing Hamiltonian} tests whether there exists $\tilde{H}\in\Pi_S$ with $\frac{1}{\sqrt{2^n}}\|H-\Tilde{H}\|_2\leq\varepsilon_1$ or $\frac{1}{\sqrt{2^n}}\|H-K\|_2>\eps_2$ for all Hamiltonians $K\in \Pi_S$ with probability at least $2/3$ using a total evolution time $\mathcal{O}\left(\frac{1}{(\eps_2-\eps_1)^{2.5} \eps_2^{0.5}}\right)$, a total number of independent experiments $N=  \mathcal{O}\left(\frac{1}{(\eps_2-\eps_1)^3\eps_2}\right)$, and a total classical processing time $\cO\left(\frac{n^2 |S\cup\{\dI\}|}{(\eps_2-\eps_1)^3 \eps_2}\right)$.
    Each experiment uses efficiently implementable states and measurements.
\end{theorem}
\begin{proof}[Proof Sketch]
    First, note that the error probability under the alternate hypothesis (i.e., in case $H$ is $\varepsilon_2$-far from having property $\Pi_S$) can be upper bounded
    with exactly the same reasoning as in the proof of \Cref{thm:upper bound on testing} (Eq.~\ref{eq:err_under_H_1}) to obtain
    \begin{align*}
	\exs{i,j,\ell}{\pr{\ket{\phi_{i,\ell}} \sim_S \ket{\phi_{i,j}}} }&\le 1-t^2\eps_2^2 + \frac{|S\cup\{\dI\}|}{d+1}+ 7t^4 \le 1-t^2\eps_2^2 +  8t^4
\end{align*}
for $|S\cup\{\dI\}|\le (d+1) t^4$. 
So, we only have to adapt the analysis of the error probability under the null hypothesis.

Therefore, suppose that $H$ is $\varepsilon_1$-close to having property $\Pi_S$ and let $\Tilde{H}\in\Pi_S$ be such that $\frac{1}{\sqrt{2^n}}\|H-\Tilde{H}\|_2\leq\varepsilon_1$. Following the reasoning from the proof of \Cref{thm:upper bound on testing}, we have to upper bound the expression 
    \begin{equation*}
        \exs{i,j,\ell}{\pr{\ket{\phi_{i,\ell}} \nsim_S \ket{\phi_{i,j}} }}
        =\frac{1}{d(d+1)} \sum_{i=1}^{d+1}\sum_{j\neq \ell}|\bra{\phi_{i,\ell}} \mathrm{e}^{\mathrm{i}tH} \ket{\phi_{i,j}}|^2 \mathbf{1}\left(\left\{ \ket{\phi_{i,\ell}} \nsim_S \ket{\phi_{i,j}}  \right\}\right)\, .
    \end{equation*}
    Using that under the event $\left\{ \ket{\phi_{i,\ell}} \nsim_S \ket{\phi_{i,j}}  \right\}$, we have both $\bra{\phi_{i,\ell}} \dI \ket{\phi_{i,j}} =0$ and $\bra{\phi_{i,\ell}} \tilde{H} \ket{\phi_{i,j}} =0$, we can expand the exponential series and obtain:
    \begin{align}
        &\exs{i,j,\ell}{\pr{\ket{\phi_{i,\ell}} \nsim_S \ket{\phi_{i,j}} }}\\
        &=\frac{1}{d(d+1)} \sum_{i=1}^{d+1}\sum_{j\neq \ell}|\bra{\phi_{i,\ell}} \mathrm{e}^{\mathrm{i}tH} \ket{\phi_{i,j}}|^2 \mathbf{1}\left(\left\{ \ket{\phi_{i,\ell}} \nsim_S \ket{\phi_{i,j}}  \right\}\right) \notag
        \\
        &=\frac{1}{d(d+1)} \sum_{i=1}^{d+1}\sum_{j\neq \ell}\left|\sum_{m\ge 0} \frac{(\mathrm{i}t)^m}{m!} \bra{\phi_{i,\ell}} H^m \ket{\phi_{i,j}}\right|^2 \mathbf{1}\left(\left\{ \ket{\phi_{i,\ell}} \nsim_S \ket{\phi_{i,j}}  \right\}\right) \notag
        \\&\le   \frac{1}{d(d+1)} \sum_{i=1}^{d+1}\sum_{j, \ell}\left| \mathrm{i} t \bra{\phi_{i,\ell}}(H-\tilde{H})\ket{\phi_{i,j}} +  \sum_{m\ge 2} \frac{(\mathrm{i}t)^m}{m!} \bra{\phi_{i,\ell}} H^m \ket{\phi_{i,j}}\right|^2 \notag
        \\&= \frac{1}{d(d+1)} \sum_{i=1}^{d+1}\sum_{j, \ell} t^2 \left\lvert\bra{\phi_{i,\ell}}(H-\tilde{H})\ket{\phi_{i,j}}\right\rvert^2\label{eq:tolerant-proof-intermediate-1}
        \\&\hphantom{=}~ +2 \frac{1}{d(d+1)} \Re\sum_{i=1}^{d+1}\sum_{j, \ell} \mathrm{i} t\bra{\phi_{i,\ell}}(H-\tilde{H})\ket{\phi_{i,j}} \sum_{m\ge 2} \frac{(-\mathrm{i}t)^m}{m!} \bra{\phi_{i,j}} H^m \ket{\phi_{i,\ell}}\label{eq:tolerant-proof-intermediate-2}
        \\&\hphantom{=}~ +\left| \sum_{m\ge 2} \frac{(\mathrm{i}t)^m}{m!} \bra{\phi_{i,\ell}} H^m \ket{\phi_{i,j}}\right|^2\label{eq:tolerant-proof-intermediate-3} \, .
    \end{align}
    The term in \eqref{eq:tolerant-proof-intermediate-3} has already been upper bounded by $t^4$ in the proof of \Cref{thm:upper bound on testing}. Now, we address the remaining two terms. First, note that
    \begin{align*}
        \eqref{eq:tolerant-proof-intermediate-1}
        &= \frac{1}{d(d+1)} \sum_{i=1}^{d+1}\sum_{j, \ell} t^2 \tr[\ket{\phi_{i,\ell}}\bra{\phi_{i,\ell}}(H-\tilde{H})\ket{\phi_{i,j}}\bra{\phi_{i,j}}(H-\tilde{H})]\\
        &= \frac{1}{d(d+1)} \sum_{i=1}^{d+1} t^2 \tr[(H-\tilde{H})^2]
        = t^2 \frac{\tr[(H-\tilde{H})^2]}{d}
        \leq  t^2 \varepsilon_1^2\, ,
    \end{align*}
    where the second step used that $\{\ket{\phi_{i,j}}\}_j$ forms an ONB for every $i$, and where the last step used that $\frac{1}{\sqrt{2^n}}\|H-\Tilde{H}\|_2\leq\varepsilon_1$.
    Now, we bound the final remaining term:
    \begin{align*}
        \eqref{eq:tolerant-proof-intermediate-2}
        &=  2 \frac{1}{d(d+1)} \Re\sum_{i=1}^{d+1}\sum_{j, \ell} \sum_{m\ge 2} \frac{(-1)^m (\mathrm{i} t)^{m+1}}{m!} \tr\left[\ket{\phi_{i,\ell}}\bra{\phi_{i,\ell}}(H-\tilde{H})\ket{\phi_{i,j}}\bra{\phi_{i,j}} H^m\right]\\
        &= 2 \frac{1}{d} \Re\sum_{m\ge 2} \frac{(-1)^m (\mathrm{i} t)^{m+1}}{m!} \tr\left[(H-\tilde{H}) H^{m}\right]\\
        &= 2 \frac{1}{d} \sum_{m\ge 2,\; m \textrm{ odd}} \frac{(-1)^m (\mathrm{i} t)^{m+1}}{m!} \tr\left[(H-\tilde{H})H^{m}\right]\\
        &= 2 \frac{1}{d} \sum_{m\ge 3,\; m \textrm{ even}} \frac{(-1)^{m-1} (\mathrm{i} t)^{m}}{(m-1)!} \tr\left[(H-\tilde{H})H^{m-1}\right]\\
        &\leq 3 t^4\, .
    \end{align*}
   where the second step used that $\{\ket{\phi_{i,j}}\}_j$ forms an ONB for every $i$; the third step used that $H=H^\dagger$ and $\tilde{H}=\tilde{H}^\dagger$, so that $\tr\left[(H-\tilde{H})H^{m}\right]\in\mathbb{R}$ for all $m$; the fourth step used $|\tr[(H-\tilde{H})H^{m-1}]| \le \|H-\tilde{H}\|_1\|H\|^{m-1}_\infty \le\sqrt{d}\|H-\tilde{H}\|_2 \le d$ and thus $\sum_{m\ge 4} \frac{t^m}{(m-1)!}|\tr((H-\tilde{H})H^{m-1})|\le 0.22\cdot dt^4$ since $t\leq 1$.

    Combining the upper bounds on \eqref{eq:tolerant-proof-intermediate-1}, \eqref{eq:tolerant-proof-intermediate-2}, and \eqref{eq:tolerant-proof-intermediate-3}, we have shown that
    \begin{equation}
        \exs{i,j,\ell}{\pr{\ket{\phi_{i,\ell}} \nsim_S \ket{\phi_{i,j}} }}
        \leq  t^2 \varepsilon_1^2 + 3  t^4 + t^4
        \leq t^2 \varepsilon_1^2 + 4 t^4\, .
    \end{equation}
    Following the logic laid out in \Cref{testing many ppts}, we set $t^2=\frac{1}{20}(\eps_2^2-\eps_1^2)$ and 
    the new decision rule is: 
\begin{center}
    At step $s=1,\dots, N$, let $\cE_s=\{ \ket{\phi_{i_s,\ell_s}} \sim_S \ket{\phi_{i_s,j_s}}\}$ and let $\bar{\cE}_s$ be its complement. Answer the null hypothesis iff
	\begin{align*}
		\frac{1}{N}\sum_{s=1}^N \mathbf{1}(\{\bar{\cE}_s\})\le \frac{1}{5}t^2(2\eps_2^2+3\eps_1^2).
	\end{align*}
\end{center}
To control the error probability under the null hypothesis, we 
apply the  Chernoff-Hoeffding inequality  \cite{hoeffding_probability_1963}:
    \begin{align*}
    	\prs{H_0}{	\frac{1}{N}\sum_{s=1}^N \mathbf{1}(\{\bar{\cE}_s\})> \frac{1}{5}t^2(2\eps_2^2+3\eps_1^2)}
     &\le \exp\left(-N\KL\left(\frac{1}{5}t^2(2\eps_2^2+3\eps_1^2) \Big\| \prs{H_0}{	\bar{\cE}}\right)\right)
    	\\&\le \exp\left(-N\KL\left(\frac{1}{5}t^2(2\eps_2^2+3\eps_1^2) \Big\| t^2\eps_1^2+4t^4 \right)\right)
     \\&= \exp\left(-N\KL\left(\frac{1}{5}t^2(2\eps_2^2+3\eps_1^2)  \Big\| \frac{1}{5}t^2(\eps_2^2+4\eps_1^2) \right)\right)
        \\&=\exp\left(-N\cdot \frac{1}{10}\cdot t^2\cdot \frac{(\eps_1^2-\eps_2^2)^2}{2\eps_2^2+3\eps_1^2} \right)
        \\&\le \delta
    \end{align*}
    where we used our choice of $t$ and $\KL(x\|  y )\ge \frac{1}{2x}(x-y)^2  $
    for $x\ge y$ and $N=\left\lceil\frac{30\log(1/\delta) (\eps_2^2+\eps_1^2)}{t^2(\eps_2^2-\eps_1^2)^2}\right\rceil.$
    
    The alternate hypothesis case is similar. 
\begin{align*}
    	\prs{H_1}{	\frac{1}{N}\sum_{s=1}^N \mathbf{1}(\{\bar{\cE}_s\})\le \frac{1}{5}t^2(2\eps_2^2+3\eps_1^2)}
     &\le \exp\left(-N\KL\left(\frac{1}{5}t^2(2\eps_2^2+3\eps_1^2) \Big\| \prs{H_1}{	\bar{\cE}}\right)\right)
    	\\&\le \exp\left(-N\KL\left(\frac{1}{5}t^2(2\eps_2^2+3\eps_1^2) \Big\| t^2\eps_2^2-8t^4 \right)\right)
     \\&= \exp\left(-N\KL\left(\frac{1}{5}t^2(2\eps_2^2+3\eps_1^2)  \Big\| \frac{1}{5}t^2(3\eps_2^2+2\eps_1^2) \right)\right)
        \\&=\exp\left(-N\cdot \frac{1}{10}\cdot t^2\cdot \frac{(\eps_1^2-\eps_2^2)^2}{3\eps_2^2+2\eps_1^2} \right)
        \\&\le \delta
    \end{align*}
    where we used our choice of $t$ and $\KL(x\|  y )\ge \frac{1}{2y}(x-y)^2  $
    for $x\le y$  and $N=\left\lceil\frac{30\log(1/\delta) (\eps_2^2+\eps_1^2)}{t^2(\eps_2^2-\eps_1^2)^2}\right\rceil =\cO\left(\frac{\log(1/\delta) \eps_2^2}{(\eps_2^2-\eps_1^2)^3}\right)=\cO\left(\frac{\log(1/\delta) }{(\eps_2-\eps_1)^3 \eps_2}\right)$. 
    This leads to a total evolution time $Nt = \mathcal{O}\left(\frac{\log(1/\delta) \sqrt{(\eps_2-\eps_1)(\eps_2+\eps_1)}}{(\eps_2-\eps_1)^3 \eps_2}\right) = \mathcal{O}\left(\frac{\log(1/\delta) }{(\eps_2-\eps_1)^{2.5} \eps_2^{0.5}}\right)$.
\end{proof}

\section{Deferred proofs}\label{appendix:deferred-proofs}

\edit{The following Lemma uses the notation from the proof of Theorem~\ref{thm:lower-Schatten-indep}}{We now prove Lemma~\ref{lemma:gen-le-cam-main-text}, using the notation from the proof of Theorem~\ref{thm:lower-Schatten-indep}}.

\begin{lemma}[Generalized Le Cam \ins{-- Restatement of Lemma~\ref{lemma:gen-le-cam-main-text}}] \label{lem:gen-le-cam}
      Let $n\geq \Omega(1)$.
      Let $k\le \cO\left(\frac{n}{\log(n)}\right)$.
      \edit{We have for $\alpha\le \frac{1}{10N}:$}{For any $\alpha\le \frac{1}{10N}$, if there is an incoherent algorithm that correctly distinguishes between $P$ and $Q_V$ with success probability at least $2/3$, then}
     \begin{align*}
         \exs{V\sim \Haar(d)}{\TV(P, Q_V^\alpha)}\ge \frac{1}{18}.
     \end{align*}
 \end{lemma}
We use Le Cam's method \cite{lecam1973convergence}:

\begin{proof}[Proof of \Cref{lem:gen-le-cam}]
 Let $x_k=\lambda^{(k)}_{i_k}\bra{\phi^{(k)}_{i_k}}\rho_k \ket{\phi^{(k)}_{i_k}}  $ and  $y_k=\lambda^{(k)}_{i_k}\bra{\phi^{(k)}_{i_k}}U_{v,t_k}\rho_k U_{v,t_k}^\dagger\ket{\phi^{(k)}_{i_k}}$. Note that since $\tr(\rho_k)=1$ we have $\sum_{i_k } x_k= \sum_{i_k} y_k=1$.

Let $\cE$ be the event that  $\eta(V\proj{0}V^\dagger-\dI/d)  $ is not $(\eta/4)$-close to any $k$-local Hamiltonian. On the one hand, by the correctness of the algorithm and the data processing inequality, we have that:
\begin{align*}
    \TV(P, \exs{V\sim \Haar(d)|\cE}{Q_{V}})\ge \TV(\Ber(1/3)\| \Ber(2/3))=\frac{1}{3}. 
\end{align*}
On the other hand, since we have $\pr{\cE}\ge 1-\exp(-\Omega(d))$ by \Cref{lem:concentration-operator-norm}, we have by the triangle inequality:
{\small
\begin{align*}
   &\TV( \exs{V\sim \Haar(d)}{Q_V},  \exs{V\sim \Haar(d)|\cE}{Q_{V} })
   \\&=  \frac{1}{2}\sum_{i}\left|\exs{V\sim \Haar(d)}{Q_V(i)} - \exs{V\sim \Haar(d)}{Q_V(i)\frac{\mathbf{1}(\{\cE\})}{\pr{\cE}}} \right|  
    \\&=  \frac{1}{2}\frac{1}{\pr{\cE}}\sum_{i}\left|\pr{\cE}\exs{V\sim \Haar(d)}{Q_V(i)} - \exs{V\sim \Haar(d)}{Q_V(i)\mathbf{1}(\{\cE\})} \right| 
      \\&\le   \frac{1}{2}\frac{1}{\pr{\cE}}\sum_{i}\left|\pr{\cE^c}\exs{V\sim \Haar(d)}{Q_V(i)}\mathbf{1}(\{\cE\})\right|+  \frac{1}{2}\frac{1}{\pr{\cE}}\sum_{i}\left|\pr{\cE} \exs{V\sim \Haar(d)}{Q_V(i)\mathbf{1}(\{\cE^c\})} \right| 
      \\&\le  \frac{1}{2}\left(\frac{\pr{\cE^c}\pr{\cE}}{\pr{\cE}}+ \pr{\cE^c}\right)
      \\&\le \exp(-\Omega(d)). 
\end{align*}}
So by the triangle inequality
{\small
\begin{align*}
\exs{V\sim \Haar(d)}{  \TV(P, Q_{V})}  &\ge  \TV(P, \exs{V\sim \Haar(d)}{Q_{V}})  \\&\ge \TV(P, \exs{V\sim \Haar(d)|\cE}{Q_{V}})-\TV( \exs{V\sim \Haar(d)}{Q_V},  \exs{V\sim \Haar(d)|\cE}{Q_{V} })
  \\&\ge \frac{1}{3}-\exp(-\Omega(d))\ge \frac{2}{9}
\end{align*}}
for $d$ (or, equivalently, $n$) larger than some constant. 

The expected $\TV$ distance between $P$ and $ Q_V^\alpha$  can be expressed as follows: 
{\small
\begin{align*}
    &2\exs{V\sim \Haar(d)}{ \TV(P,  Q_V^\alpha)}
    \\&= \exs{V\sim \Haar(d)}{\sum_{i_1,\dots,i_N} \left| \prod_{k=1}^N \lambda^{(k)}_{i_k}\bra{\phi^{(k)}_{i_k}}(\alpha\rho_k+(1-\alpha) U_{v,t_k}\rho_k U_{v,t_k}^\dagger)\ket{\phi^{(k)}_{i_k}} - \prod_{k=1}^N \lambda^{(k)}_{i_k}\bra{\phi^{(k)}_{i_k}}\rho_k\ket{\phi^{(k)}_{i_k}}\right|}
    \\&= \exs{V\sim \Haar(d)}{ \sum_{i_1,\dots,i_N} \left| \prod_{k=1}^N (\alpha x_k+ (1-\alpha)y_k) - \prod_{k=1}^N (\alpha x_k+ (1-\alpha)x_k) \right|}
    \\&=\exs{V\sim \Haar(d)}{  \sum_{i_1,\dots,i_N} \left| \sum_{S\subset [N]} \alpha^{|S|}(1-\alpha)^{N-|S|}   \prod_{k\in S}  x_k\left(\prod_{k\notin S}y_k  - \prod_{k\notin S} x_k\right)\right|}
    \\&\ge \mathbb{E}_{V\sim \Haar(d)}\Bigg[\sum_{i_1,\dots,i_N} \left| \sum_{S= \emptyset} \alpha^{|S|}(1-\alpha)^{N-|S|}   \prod_{k\in S}  x_k\left(\prod_{k\notin S}y_k  - \prod_{k\notin S} x_k\right)\right|
    \\&\quad -\sum_{i_1,\dots,i_N} \left| \sum_{\emptyset \neq S\subset [N]} \alpha^{|S|}(1-\alpha)^{N-|S|}   \prod_{k\in S}  x_k\left(\prod_{k\notin S}y_k  - \prod_{k\notin S} x_k\right)\right|\Bigg]
\end{align*}}
When $S=\emptyset$, we recover the $\TV$ distance between $P$ and $Q$ up to a factor as follows:
\begin{align*}
    &\exs{V\sim \Haar(d)}{ \sum_{i_1,\dots,i_N} \left| \sum_{S= \emptyset} \alpha^{|S|}(1-\alpha)^{N-|S|}   \prod_{k\in S}  x_k\left(\prod_{k\notin S}y_k  - \prod_{k\notin S} x_k\right)\right|}
    \\&=\exs{V\sim \Haar(d)}{ \sum_{i_1,\dots,i_N} \left| (1-\alpha)^{N}  \left(\prod_{k=1}^N y_k  - \prod_{k=1}^N x_k\right)\right|}
    \\&=\exs{V\sim \Haar(d)}{ (1-\alpha)^{N}  \sum_{i_1,\dots,i_N} \left|  \prod_{k=1}^N \lambda^{(k)}_{i_k}\bra{\phi^{(k)}_{i_k}}U_{v,t_k}\rho_k U_{v,t_k}^\dagger)\ket{\phi^{(k)}_{i_k}}  - \prod_{k=1}^N \lambda^{(k)}_{i_k}\bra{\phi^{(k)}_{i_k}}\rho_k \ket{\phi^{(k)}_{i_k}}\right|}
    \\&= 2(1-\alpha)^N\exs{V\sim \Haar(d)}{ \TV(P, Q_V)}\ge \frac{4(1-\alpha)^N}{9}.
\end{align*}
When $S\neq\emptyset$, we can use the triangle inequality:
\begin{align*}
    &\sum_{i_1,\dots,i_N} \left| \sum_{\emptyset \neq S\subset [N]} \alpha^{|S|}(1-\alpha)^{N-|S|}   \prod_{k\in S}  x_k\left(\prod_{k\notin S}y_k  - \prod_{k\notin S} x_k\right)\right|
    \\&\le  \sum_{i_1,\dots,i_N}  \sum_{\emptyset \neq S\subset [N]} \alpha^{|S|}(1-\alpha)^{N-|S|}   \prod_{k\in S}  x_k\left(\prod_{k\notin S}y_k  +\prod_{k\notin S} x_k\right)
    \\&=  \sum_{\emptyset \neq S\subset [N]} \alpha^{|S|}(1-\alpha)^{N-|S|}  \sum_{i_1,\dots,i_N}  \prod_{k\in S}  x_k\left(\prod_{k\notin S}y_k  +\prod_{k\notin S} x_k\right)
    \\&=  \sum_{\emptyset \neq S\subset [N]} \alpha^{|S|}(1-\alpha)^{N-|S|}    \prod_{k\in S} \sum_{i_k} x_k\left(\prod_{k\notin S}\sum_{i_k} y_k  +\prod_{k\notin S} \sum_{i_k}  x_k\right)
    \\&= 2(1-(1-\alpha)^N).
\end{align*}
Therefore:
\begin{align*}
    2\exs{V\sim \Haar(d)}{ \TV(P,  Q_V^\alpha)}&\ge  \frac{4(1-\alpha)^N}{9} - 2(1-(1-\alpha)^N).
\end{align*}
Let $\alpha= \frac{c}{N}$ where $0 \leq c\leq 1/10$ is a small constant. We have:
 \begin{align*}
     \frac{4}{9}(1-\alpha)^N-2(1-(1-\alpha)^N)=   \frac{22}{9}\left( 1-\frac{c}{N}\right)^N-2\ge \frac{22}{9}\left( 1-c\right)-2= \frac{4-22c}{9}\geq \frac{1}{9}. 
 \end{align*}
 Therefore 
 \begin{align*}
     \exs{V\sim \Haar(d)}{ \TV(P,  Q_V^\alpha)}&\ge  \frac{1}{18}.
 \end{align*}
 \end{proof}

\begin{lemma}\label{lem:elem-calculations}
    For all $x \in [\frac{1}{10N}, \infty)$,
\begin{align*}
    (-\log)(x)\le -(x-1)+2\log(10N)(x-1)^2\, .
\end{align*}
\end{lemma}

\begin{proof}
   Let $f(x)= \log(x) -(x-1) + 2\log(10N)(x-1)^2$. We have for $c= \frac{1}{4\log(10N)}<1$: $$f'(x)= \frac{1}{x}-1+4\log(10N)(x-1)= \frac{(1-x)}{cx}\left(c-x  \right)\, ,$$  
which is positive for $x \in (0,c) \cup (1, \infty)$ and negative for $c < x< 1$, $\lim_{x\rightarrow 0^+} f(x)=-\infty$ and $f(1)=0$. Hence, there is a $0<c'<c$ such that: $$x\ge c' \Longleftrightarrow f(x) \ge 0.$$ But we have $$f\left(\frac{1}{10N}\right)= -\log(10N)-\frac{1}{10N}+1 +2\log(10N)\left(1-\frac{1}{10N}\right)^2\ge 1-\frac{1}{10N}  >0 \, .$$  Thus, we can take $c'<\frac{1}{10N}$ and for all $x\ge \frac{1}{10N}>c'$ we have $f(x)\ge 0$.
\end{proof}

\section{Weingarten Calculus} \label{sec:weingarten facts}
As we use a random Hamiltonian constructed from sampling a $\Haar$-random unitary matrix in our lower bound proofs, we need some facts from Weingarten calculus in order to compute the corresponding expectation values with respect to the Haar measure.  
If $\pi$ is a permutation of $[n]$, let $\W(\pi,d)$ denote the Weingarten function of dimension $d$. The following lemma is useful for our results. 
\begin{lemma}[{\cite{gu2013moments}}]\label{lem:Wg} Let $U$ be a $\Haar$-distributed unitary $(d\times d)$-matrix and let $\{A_i,B_i\}_{i=1}^n$ be a sequence of complex $(d\times d)$-matrices. We have the following formula for the expectation value:
\begin{align*}
&\ex{\tr(UB_1U^*A_1U\dots UB_nU^*A_n)}
\\&=\sum_{\alpha,\beta \in \fS_n}\W(\beta\alpha^{-1},d)\tr_{\beta^{-1}}(B_1,\dots,B_n)\tr_{\alpha\gamma_n}(A_1,\dots,A_n),
\end{align*}where $\gamma_n=(12\dots n)$ and $\tr_{\sigma}(M_1,\dots,M_n)=\Pi_j \tr(\Pi_{i\in C_j} M_i)$ for $\sigma=\Pi_j C_j $ and $C_j$ are cycles. 

\end{lemma}
We will also need some values of Weingarten function:
\begin{lemma}[\cite{collins2006integration}]\label{lem:wg2}
The function $\W(\pi,d)$ has the following values:
\begin{itemize}
    \item $\W((1),d)=\frac{1}{d}$,
    \item $\W((12),d)=\frac{-1}{d(d^2-1)}$,
    \item $\W((1)(2),d)=\frac{1}{d^2-1}$,
      \item $\W((123),d)=\frac{2}{d(d^2-1)(d^2-4)}$,
    \item $\W((12)(3),d)=\frac{-1}{(d^2-1)(d^2-4)}$,
    \item $\W((1)(2)(3),d)=\frac{d^2-2}{d(d^2-1)(d^2-4)}$,
    \item $\W((1234),d)= -\frac{5}{d^7-14d^5+49d^3-36d}$,
     \item $\W((12)(34),d)= \frac{d^2+6}{d^8-14d^6+49d^4-36d^2}$,
       \item $\W((123)(4),d)= \frac{2d^2-3}{d^8-14d^6+49d^4-36d^2}$,
       \item $\W((12)(3)(4),d)= -\frac{1}{d^5-10d^3+9d}$,
    \item $\W((1)(2)(3)(4),d)= \frac{d^4-8d^2+6}{d^8-14d^6+49d^4-36d^2}$.
\end{itemize}
\end{lemma}

\end{document}